\documentclass[11pt]{amsart}

\usepackage{amssymb}
\usepackage{fullpage}

\newtheorem{theorem}{Theorem}[section]
\newtheorem{lemma}[theorem]{Lemma}
\newtheorem{corollary}[theorem]{Corollary}
\newtheorem{proposition}[theorem]{Proposition}

\newtheorem{remark}[theorem]{Remark}

\numberwithin{equation}{section}
\newcommand{\Ai}{\text{Ai\,}}

\newcommand{\re}{\text{Re\,}}

\newcommand{\sgn}{\text{sgn\,}}

\begin{document}
\setcounter{page}{1}

\thanks{Supported by the grant KAW 2010.0063 from the Knut and Alice Wallenberg
Foundation}

\title[Two time distribution in Brownian directed percolation]
{Two time distribution in Brownian directed percolation}
\author[K.~Johansson]{Kurt Johansson}

\address{
Department of Mathematics,
KTH Royal Institute of Technology,
SE-100 44 Stockholm, Sweden}

\email{kurtj@kth.se}

\begin{abstract}
In the zero temperature Brownian semi-discrete directed polymer we study the joint distribution of two last-passage times 
at positions ordered in the time-like direction. This is the situation when we have the slow de-correlation phenomenon. 
We compute the limiting joint distribution function in a scaling limit. This limiting distribution is given by
an expansion in determinants which is not a Fredholm expansion. 
A somewhat similar looking formula was derived non-rigorously in a related model by Dotsenko.

\end{abstract}

\maketitle

\section{Introduction and results}\label{sect1}
Let $B_i(t)$, $t\ge 0$, $i\ge 1$, be independent standard Brownian motions. We consider the zero temperature Brownian semi-discrete
directed polymer, \cite{Bar}, \cite{GraTrWi},\cite{BoJeu}, \cite{OCon}. The last-passage time in this model is defined by
\begin{equation}\label{1.1}
H(\mu,n)=\sup_{0=\tau_0<\tau_1<\dots<\tau_n=\mu}\sum_{i=1}^nB_i(\tau_i)-B_i(\tau_{i-1}).
\end{equation}
We are interested in the asymptotics of the joint distribution function
\begin{equation}\label{1.2}
\mathbb{P}[H(\mu_1,n_1)\le\xi_1,H(\mu_2,n_2)\le\xi_2]
\end{equation}
when $(\mu_1,n_1)$ and $(\mu_2,n_2)$ are ordered in the time-like direction, $\mu_1<\mu_2$, $n_1<n_2$. The random variable (\ref{1.1})
is distributed as the largest eigenvalue of a GUE random matrix, \cite{Bar}. More precisely,
\begin{equation}\label{1.3}
\mathbb{P}[H(\mu,n)\le\xi]=\frac 1{Z_{\mu,n}}\int_{(-\infty,\xi]^n}\prod_{1\le j<k\le n}(x_k-x_j)^2\prod_{j=1}^ne^{-\frac{x_j^2}{2\mu}}\,d^nx.
\end{equation}
By standard results this leads to the following limit law for $H(\mu,n)$. Let $t,\nu$ and $\eta$ be fixed. Then
\begin{equation}\label{1.4}
\lim_{M\to\infty}\mathbb{P}\left[H(tM-\nu(tM)^{2/3},[tM+\nu(tM)^{2/3}])\le 2tM+(\eta-\nu^2)(tM)^{1/3}\right]=F_2(\eta),
\end{equation}
where $F_2$ is the GUE Tracy-Widom distribution,
\begin{equation}\label{1.5}
F_2(\eta)=\det (I-K_{\Ai})_{L^2(\eta,\infty)}.
\end{equation}
Here $K_{\Ai}$ is the Airy kernel,
\begin{equation}\label{1.6}
K_{\Ai}(x,y)=\int_0^\infty\Ai(x+\tau)\Ai(y+\tau)\,d\tau.
\end{equation}
When  $(\mu_1,n_1)$ and $(\mu_2,n_2)$ have a space-like ordering, $\mu_1<\mu_2$, $n_1>n_2$, the asymptotics for (\ref{1.2}) 
analogous to (\ref{1.4}) can be computed and expressed in terms of a Fredholm determinant with the extended Airy kernel. 
This leads to the possibility of proving convergence to the Airy process along space-like paths, \cite{BorOl}, \cite{Fer}. However, the case when 
$(\mu_1,n_1)$ and $(\mu_2,n_2)$ are ordered in the time-like direction (more precisely along a characteristic, see e.g. \cite{Fer}) 
has not been considered previously except non-rigorously in a related model by 
Dotsenko, \cite{Dots}, using the replica method. The main result of this paper is given in the next theorem.

\begin{theorem}\label{thm1.1}
Let $0<t_1<t_2$, $\eta_1,\eta_2,\nu_1,\nu_2\in\mathbb{R}$ be given. Set
\begin{equation}\label{1.7}
\alpha=(t_1/\Delta t)^{1/3},
\end{equation}
where $\Delta t=t_2-t_1$, and let $F_{\text{tt}}(\eta_1, \eta_2;\alpha, \nu_1, \nu_2)$ be given by (\ref{1.18}) below. 
Introduce the scaling
\begin{equation}\label{scaling}
\mu_i=t_iM-\nu_i(t_iM)^{2/3},\,n_i=t_iM+\nu_i(t_iM)^{2/3},\,\xi_i=2t_iM+(\eta_i-\nu_i^2)(t_iM)^{1/3},
\end{equation}
$i=1,2$. With this scaling, define
\begin{equation}\label{1.5'}
F_M(\eta_1, \eta_2;t_1,t_2, \nu_1, \nu_2)=\mathbb{P}[H(\mu_1,n_1)\le\xi_1,H(\mu_2,n_2)\le\xi_2]
\end{equation}
Then,
\begin{equation}\label{1.6'}
\lim_{M\to\infty} F_M(\eta_1, \eta_2;t_1,t_2, \nu_1, \nu_2)=F_{\text{tt}}(\eta_1, \eta_2;\alpha, \nu_1, \nu_2).
\end{equation}

\end{theorem}

The theorem will be proved in section \ref{sect4}.

In order to give the formula for the limiting distribution function we first need to define some functions. Set
\begin{align}\label{1.8}
\Delta\nu&=\nu_2\left(\frac{t_2}{\Delta t}\right)^{2/3}-\nu_1\left(\frac{t_1}{\Delta t}\right)^{2/3},
\\
\Delta\eta&=(\eta_2-\nu_2^2)\left(\frac{t_2}{\Delta t}\right)^{1/3}-(\eta_1-\nu_1^2)\left(\frac{t_1}{\Delta t}\right)^{1/3}
+\Delta\nu^2.\notag
\end{align}
Let
\begin{equation}\label{1.9}
\phi_1(x,y)=-\alpha e^{\alpha\Delta\nu x-\nu_1y}\int_0^\infty e^{(\nu_1-\alpha\Delta\nu)\tau}K_{\Ai}(\eta_1-\tau,\eta_1-y)
K_{\Ai}(\Delta\eta+\alpha\tau,\Delta\eta+\alpha x)\,d\tau,
\end{equation}
\begin{equation}\label{1.10}
\psi_1(x,y)=\alpha e^{\alpha\Delta\nu x-\nu_1y}\int_0^\infty e^{-(\nu_1-\alpha\Delta\nu)\tau}K_{\Ai}(\eta_1+\tau,\eta_1-y)
K_{\Ai}(\Delta\eta-\alpha\tau,\Delta\eta+\alpha x)\,d\tau,
\end{equation}
\begin{equation}\label{1.11}
\phi_2(x,y)=\alpha e^{\alpha\Delta\nu(x-y)}K_{\Ai}(\Delta\eta+\alpha x,\Delta\eta+\alpha y),
\end{equation}
and
\begin{equation}\label{1.12}
\phi_3(x,y)=e^{\nu_1(x-y)}K_{\Ai}(\eta_1-x,\eta_1-y).
\end{equation}
Define
\begin{equation}\label{1.13}
\phi(x,y)=\phi_1(x,y)+1(y\ge 0)\phi_2(x,y)-1(x<0)\phi_3(x,y),
\end{equation}
and
\begin{equation}\label{1.14}
\psi(x,y)=-\psi_1(x,y)-1(y>0)\phi_2(x,y)+1(x\le0)\phi_3(x,y),
\end{equation}
where $1(\cdot)$ is the indicator function. We will use the following notation in block matrices. If $f$ is a function of two real variables,
$\mathbf{x}\in\mathbb{R}^s$ and $\mathbf{y}\in\mathbb{R}^t$ we write
\begin{equation}\label{1.15}
f(\mathbf{x},\mathbf{y})=(f(x_i,y_j))_{\substack{ 1\le i\le s\\1\le j\le t}}\,\,\,,
\end{equation}
for a matrix block. When $s$ or $t$ is equal to zero this block is just empty and left out of the block matrix. 
Let $r_1,r_2,s,t\ge 0$, $\mathbf{x}\in\mathbb{R}^{r_1}$, $\mathbf{x'}\in\mathbb{R}^{s}$, $\mathbf{y}\in\mathbb{R}^{r_2}$,
$\mathbf{y'}\in\mathbb{R}^{t}$ and $0\in\mathbb{R}$. Define the determinants
\begin{equation}\label{1.16}
W_{r_1,s,r_2,t}^{(1)}(\mathbf{x},\mathbf{x'},\mathbf{y},\mathbf{y'})=
\left|\begin{matrix} \psi(\mathbf{x},\mathbf{x}) &\psi(\mathbf{x},\mathbf{x'}) &\psi(\mathbf{x},0) &\psi(\mathbf{x},\mathbf{y})
&\psi(\mathbf{x},\mathbf{y'})\\
\phi(\mathbf{x'},\mathbf{x}) &\phi(\mathbf{x'},\mathbf{x'}) &\phi(\mathbf{x'},0) &\phi(\mathbf{x'},\mathbf{y})
&\phi(\mathbf{x'},\mathbf{y'})\\
\psi(0,\mathbf{x}) &\psi(0,\mathbf{x'}) &\psi(0,0) &\psi(0,\mathbf{y})
&\psi(0,\mathbf{y'})\\
\phi(\mathbf{y},\mathbf{x}) &\phi(\mathbf{y},\mathbf{x'}) &\phi(\mathbf{y},0) &\phi(\mathbf{y},\mathbf{y})
&\phi(\mathbf{y},\mathbf{y'})\\
\psi(\mathbf{y'},\mathbf{x}) &\psi(\mathbf{y'},\mathbf{x'}) &\psi(\mathbf{y'},0) &\psi(\mathbf{y'},\mathbf{y})
&\psi(\mathbf{y'},\mathbf{y'})
\end{matrix}\right|
\end{equation}
(the determinant is of size $r_1+s+r_2+t+1$) and
\begin{equation}\label{1.17}
W_{r_1,s,r_2,t}^{(2)}(\mathbf{x},\mathbf{x'},\mathbf{y},\mathbf{y'})=
\left|\begin{matrix} \psi(\mathbf{x},\mathbf{x}) &\psi(\mathbf{x},\mathbf{x'}) &\psi(\mathbf{x},0) &\psi(\mathbf{x},\mathbf{y})
&\psi(\mathbf{x},\mathbf{y'})\\
\phi(\mathbf{x'},\mathbf{x}) &\phi(\mathbf{x'},\mathbf{x'}) &\phi(\mathbf{x'},0) &\phi(\mathbf{x'},\mathbf{y})
&\phi(\mathbf{x'},\mathbf{y'})\\
\phi(0,\mathbf{x}) &\phi(0,\mathbf{x'}) &\phi(0,0) &\phi(0,\mathbf{y})
&\phi(0,\mathbf{y'})\\
\phi(\mathbf{y},\mathbf{x}) &\phi(\mathbf{y},\mathbf{x'}) &\phi(\mathbf{y},0) &\phi(\mathbf{y},\mathbf{y})
&\phi(\mathbf{y},\mathbf{y'})\\
\psi(\mathbf{y'},\mathbf{x}) &\psi(\mathbf{y'},\mathbf{x'}) &\psi(\mathbf{y'},0) &\psi(\mathbf{y'},\mathbf{y})
&\psi(\mathbf{y'},\mathbf{y'})
\end{matrix}\right|.
\end{equation}
We can now give the expression for the distribution function $F_{\text{tt}}(\eta_1, \eta_2;\alpha, \nu_1, \nu_2)$ in theorem \ref{thm1.1}.
Define
\begin{align}\label{1.18}
&F_{\text{tt}}(\eta_1^\ast, \eta_2;\alpha, \nu_1, \nu_2)
\\&=F_2(\eta_2)-
\sum\limits_{r,s,t=0}^\infty\frac 1{(r!)^2s!t!}\int\limits_{\eta_1^\ast}^\infty d\eta_1\int\limits_{(-\infty,0]^r}d^rx\int\limits_{(-\infty,0]^s}d^sx'
\int\limits_{[0,\infty)^r}d^ry\int\limits_{[0,\infty)^t}d^ty' 
W_{r,s,r,t}^{(1)}(\mathbf{x},\mathbf{x'},\mathbf{y},\mathbf{y'})\notag\\
&-\sum\limits_{r=1}^\infty\sum\limits_{s,t=0}^\infty\frac 1{r!(r-1)!s!t!}\int\limits_{\eta_1^\ast}^\infty\,d\eta_1\int\limits_{(-\infty,0]^r}d^rx
\int\limits_{(-\infty,0]^s}d^sx'
\int\limits_{[0,\infty)^{r-1}}d^{r-1}y\int_{[0,\infty)^t}d^ty' 
W_{r,s,r-1,t}^{(2)}(\mathbf{x},\mathbf{x'},\mathbf{y},\mathbf{y'})\notag,
\end{align}
where $F_2$ is the Tracy-Widom distribution given by (\ref{1.5}). Recall that the Tracy-Widom distribution $F_2$ in (\ref{1.5}) can be written as a
Fredholm expansion. The two-time distribution function $F_{\text{tt}}$ is not given by
a Fredholm expansion although the expansion in (\ref{1.18}) has some similarities with a block Fredholm expansion.

We will derive the formulas that we will use to prove (\ref{1.6'}) by thinking of $H(\mu,n)$ as a limit of a last-passage 
time in a discrete model. Let $\left(w(i,j)\right)_{i,j\ge 1}$ be independent geometric random variables with parameter $q$,
$$
\mathbb{P}[w(i,j)=k]=(1-q)q^k,\quad k\ge 0.
$$
Consider the last-passage times
\begin{equation}\label{1.19}
G(m,n)=\max_{\pi:(1,1)\nearrow (m,n)} \sum_{(i,j)\in\pi} w(i,j),
\end{equation}
where the maximum is over all up/right paths from $(1,1)$ to $(m,n)$, see \cite{JoSh}. We have the following limit law
\begin{equation}\label{1.20}
\frac{G([\mu T],n)-\frac q{1-q}[\mu T]}{\frac{\sqrt{q}}{1-q}\sqrt{T}}
\to  H(\mu,n)
\end{equation}
in distribution as $T\to\infty$, see \cite{Bar}. The distribution function
$\mathbb{P}[G(m_1,n_1)\le v_1, G(m_2, n_2)\le v_2]$ will be analyzed using a formula from \cite{JoMar}, see section \ref{sect2} below.

\begin{remark}\label{rem1} 
{\rm As mentioned above Dotsenko has given a non-rigorous derivation of the limiting distribution function
$F_{\text{tt}}$. The formula in \cite{Dots} has similarities with (\ref{1.18}) but we have not attempted to prove that
they are the same. Dotsenko also used a similar derivation in the space-like direction, \cite{Dots2}, see also \cite{ImSaSp}. }
\end{remark}

\begin{remark}\label{rem2}
{\rm This paper is a contribution to the understanding of models in the so called KPZ universality class, which have been of great interest in the last 15 years. We will not survey this development here, see for example the papers \cite{BorCorw}, \cite{BorPet}, \cite{Corw}, \cite{JoHouch},
\cite{Quas} and references therein. In particular the results of this paper could be of interest in understanding the so called Airy sheet, a conjectural limiting object for many models, see \cite{CorwQua} and \cite{CorLiuWa}. One aspect about last-passage percolation models in the time direction has been studied previously, namely the so called slow de-correlation phenomenon, see \cite{CorFerPec}, \cite{Fer}. This means that the scaling
exponent in the time direction (characteristic direction) is 1; we need $\mu_1$ and $\Delta\mu$ to be of order $M$ above to get a 
non-trivial limit.   }
\end{remark}

\begin{remark}\label{rem3}
{\rm It is not so hard to check, disregarding technical details, that $F(\eta_1,\eta_2;\alpha)\to F_2(\eta_1)$ as $\eta_2\to\infty$ and 
$F(\eta_1,\eta_2;\alpha)\to F_2(\eta_2)$ as $\eta_1\to\infty$. We also expect that $F(\eta_1,\eta_2;\alpha)\to F_2(\eta_1)F_2(\eta_2)$
as $\alpha\to 0+$. This limit can be checked heuristically but appears to be rather subtle and we will not discuss it further.} 
\end{remark}

\begin{remark}\label{rem3}
{\rm Below we will derive formulas for the corresponding problem for the last-passage times $G(m,n)$ before
taking the limit to $H(\mu,n)$. It should be possible to carry out the whole proof below but with $G(m,n)$ instead, but some of the computations
in section \ref{sect3} appear to be harder. The role of the Hermite polynomials there would then be replaced by the Meixner polynomials.}
\end{remark}

The outline of the paper is as follows. In section \ref{sect2} we will prove a formula for the joint distribution function 
$\mathbb{P}[G(m_1,n_1)\le v_1, G(m_2, n_2)\le v_2]$ based on results from \cite{JoMar}. By taking a limit this leads to a formula 
for (\ref{1.2}). This computation involves certain symmetrization identities that will be proved in section \ref{sect5}.
In section \ref{sect3} the formula from section \ref{sect2} will be rewritten and expanded in terms of determinants. Section \ref{sect4}
gives the proof of theorem \ref{thm1.1} based on the expansion, certain asymptotic limits and estimates. These limits and estimates are
finally proved in section \ref{sect6}.

Throughout this paper $\gamma_r$ will denote a positively oriented circle around the origin with radius $r$, and $\Gamma_d$ will
denote a straight line through $d$ parallell to the imaginary axis and oriented upwards.

\bigskip
{\bf Acknowledgement.} This work was inspired by a talk by Victor Dotsenko at the Simons Symposium
{\it The Kardar-Parisi-Zhang Equation and Universality Class}, which appeared as \cite{Dots}. I thank the 
Simons Foundation for the
invitation. The work was started while visiting the Institute of Advanced Study. I thank the Institute
for inviting me to the special year on {\it Non-equilibrium Dynamics and Random Matrices}. I thank an anonymous referee
for many valuable comments and suggestions.


\section{A formula for the joint distribution function}\label{sect2}

Let $G(m,n)$, $m,n\ge 1$, be the last-passage times defined by (\ref{1.19}), and write
$$
\mathbf{G}(m)=(G(m,1),\dots,G(m,n)).
$$
We put $\mathbf{G}(0)=\mathbf{0}$. Introduce the difference operators $\Delta f(x)=f(x+1)-f(x)$, $\Delta^{-1}f(x)
=\sum_{y=-\infty}^{x-1} f(y)$, where $f:\mathbb{Z}\to\mathbb{R}$ is a given function. Set, for $m\ge 1$, $x\in\mathbb{Z}$,
\begin{equation}\label{2.1}
w_m(x)=(1-q)^m\binom{x+m-1}{x} q^x1(x\ge 0).
\end{equation} 
Also, we let $W_n=\{x\in \mathbb{Z}^n\,;\,x_1\le x_2\le\dots\le x_n\}$. In \cite{JoMar} the following result was proved, inspired by \cite{Warr}.

\begin{proposition}\label{prop2.1} For $x,y\in W_n$ and $m>\ell\ge 0$,
\begin{equation}\label{2.2}
\mathbb{P}[\mathbf{G}(m)=y\,|\,\mathbf{G}(\ell)=x]=\det \left(\Delta^{j-i} w_{m-\ell}(y_j-x_i)\right)_{1\le i,j\le n}.
\end{equation} 
In particular
\begin{equation}\label{2.3}
\mathbb{P}[\mathbf{G}(m)=x]= \det\left(\Delta^{j-i} w_{m}(x_j)\right)_{1\le i,j\le n}.
\end{equation} 
\end{proposition}

It follows from (\ref{2.2}) and (\ref{2.3}) that 
\begin{align}
&\mathbb{P}[\mathbf{G}(m_1)=x,\mathbf{G}(m_2)=y]=\mathbb{P}[\mathbf{G}(m_2)=y\,|\,\mathbf{G}(m_1)=x]\mathbb{P}[\mathbf{G}(m_1)=x]
\notag\\
&=\det \left(\Delta^{j-i} w_{m_2-m_1}(y_j-x_i)\right)_{1\le i,j\le n}\det\left(\Delta^{j-i} w_{m_1}(x_j)\right)_{1\le i,j\le n}.
\notag
\end{align}
Thus,
\begin{align}\label{2.4}
P&:=\mathbb{P}[G(m_1,n_1)\le v_1, G(m_2,n_2)\le v_2]
\\
&=\sum\limits_{u=-\infty}^{v_1}\sum\limits_{\substack{x\in W_{n_2}\\x_{n_1}=u}}\sum\limits_{\substack{y\in W_{n_2}\\y_{n_2}\le v_2}}
\det\left(\Delta^{j-i} w_{m_1}(x_j)\right)_{1\le i,j\le n_2}\det \left(\Delta^{j-i} w_{m_2-m_1}(y_j-x_i)\right)_{1\le i,j\le n_2},
\notag
\end{align}
where $1\le m_1<m_2$ and $1\le n_1\le n_2$. This formula is the starting point of our analysis. In order to get a more useful formula
we rewrite it in terms of multiple contour integrals. We can write $w_m$ in (\ref{2.1}) as
\begin{equation}\label{2.5}
w_m(x)=\frac{(1-q)^m}{2\pi i}\int_{\gamma_r}\frac{dz}{(1-qz)^mz^{x+1}},
\end{equation}
where $\gamma_r$ is a positively oriented circle around the origin with radius $r$ and $0<r<1/q$. This gives
\begin{equation}\label{2.6}
\Delta^kw_m(x)=\frac{(1-q)^m}{2\pi i}\int_{\gamma_r}\frac{(1-z)^kdz}{(1-qz)^mz^{x+k+1}},
\end{equation}
for all $k\in\mathbb{Z}$ if $0<r<1$. Inserting (\ref{2.6}) into (\ref{2.4}) will after some rather lengthy and non-trivial 
manipulations lead to the following formula for $P$.

\begin{proposition}\label{prop2.2}
Let $P$ be defined by (\ref{2.4}), and let $0<s_1<r_1<1$, $0<r_2<s_2<1$. Assume that $m_1<m_2$ and $n_1<n_2$ and let
$\Delta n=n_2-n_1$, $\Delta m=m_2-m_1$.
Then
\begin{align}\label{2.7}
&P=\sum\limits_{u=-\infty}^{v_1}\frac{(1-q)^{m_2n_2}(-1)^{n_2(n_2-1)/2}}{(2\pi i)^{2n_2}n_1!^2(\Delta n)!^2}
\int\limits_{\gamma_{s_1}^{n_1}}d^{n_1}z\int\limits_{\gamma_{s_2}^{\Delta n}}d^{\Delta n}z
\int\limits_{\gamma_{r_1}^{n_1}}d^{n_1}w\int\limits_{\gamma_{r_2}^{\Delta n}}d^{\Delta n}w \\
&\times\det \left(z_j^{i-1}\right)_{1\le i,j\le n_2}\det \left(w_j^{i-1}\right)_{1\le i,j\le n_2}
\det \left(\frac 1{w_j-z_i}\right)_{1\le i,j\le n_1}\det \left(\frac 1{z_j-w_i}\right)_{n_1<i,j\le n_2}\notag\\
&\times\prod\limits_{j=n_1+1}^{n_2}\frac{1-z_j}{1-w_j}\left(1-\prod\limits_{j=1}^{n_1}\frac{z_j}{w_j}\right)
\prod\limits_{j=1}^{n_2}\frac{w_j^{u-v_2-\Delta n}}
{z_j^{u+n_1}(1-z_j)^{\Delta n}(1-qz_j)^{m_1}(1-w_j)^{n_1}(1-qw_j)^{\Delta m}}.
\notag
\end{align}
\end{proposition}
Here we have used the notation
\begin{equation}\label{2.8}
\int\limits_{\gamma_{s_1}^{n_1}}d^{n_1}z\int\limits_{\gamma_{s_2}^{\Delta n}}d^{\Delta n}z=
\int\limits_{\gamma_{s_1}}dz_1\dots\int\limits_{\gamma_{s_1}}dz_{n_1}\int\limits_{\gamma_{s_2}}dz_{n_1+1}\dots\int\limits_{\gamma_{s_2}}dz_{n_2}.
\end{equation}

Before we can prove (\ref{2.7}) we need some preliminary results.

\begin{lemma}\label{lem2.3}
We have the following two algebraic symmetrization identities,
\begin{align}\label{2.9}
&\sum\limits_{\sigma\in S_n}\sgn(\sigma)\prod\limits_{j=1}^n\left(\frac{1-w_{\sigma(j)}}{w_{\sigma(j)}}\right)^j
\frac 1{(1-w_{\sigma(1)})(1-w_{\sigma(1)}w_{\sigma(2)})\cdots (1-w_{\sigma(1)}\cdots
w_{\sigma(n)})}\\
&=(-1)^{\frac{n(n-1)}2}\prod\limits_{j=1}^n \frac 1{w_j^n}\det\left(w_j^{i-1}\right)_{1\le i,j\le n},\notag
\end{align}
and
\begin{align}\label{2.10}
&\sum\limits_{\sigma_1,\sigma_2\in S_n}\sgn(\sigma_1\sigma_2)\prod\limits_{j=1}^n
\left(\frac{w_{\sigma_2(j)}(1-z_{\sigma_1(j)})}{z_{\sigma_1(j)}(1-w_{\sigma_2(j)})}\right)^j
\frac 1{\left(1-\frac{z_{\sigma_1(1)}}{w_{\sigma_2(1)}}\right)\left(1-\frac{z_{\sigma_1(1)}z_{\sigma_1(2)}}{w_{\sigma_2(1)}w_{\sigma_2(2)}}\right)\cdots
\left(1-\frac{z_{\sigma_1(1)}\cdots z_{\sigma_1(n)}}{w_{\sigma_2(1)}\cdots w_{\sigma_2(n)}}\right)}\\
&=\prod\limits_{j=1}^n \frac {w_j^{n+1}(1-z_j)^n}{z_j^n(1-w_j)^n}\det\left(\frac 1{w_j-z_i}\right)_{1\le i,j\le n}.\notag
\end{align}
\end{lemma}
The first identity is a direct consequence of one of the Tracy-Widom ASEP identities, \cite{TrWi}. The other identity is new as far as we know. The lemma will be proved in sect. \ref{sect5}.

We can now give the proof of proposition \ref{prop2.2}.
\begin{proof}({\it Proposition \ref{prop2.2}}) Inserting (\ref{2.6}) into (\ref{2.4}) gives
\begin{align}\label{2.11}
P&=\sum\limits_{u=-\infty}^{v_1}\sum\limits_{\substack{x\in W_{n_2}\\x_{n_1}=u}}\sum\limits_{\substack{y\in W_{n_2}\\y_{n_2}\le v_2}}
\det\left(\frac{(1-q)^{m_1}}{2\pi i}\int_{\gamma_{r_1}}\left(\frac{1-z_j}{z_j}\right)^{j-i}\frac{dz_j}{(1-qz_j)^{m_1}z_j^{x_j+1}}\right)_{1\le i,j\le n_2}
\\
&\det\left(\frac{(1-q)^{\Delta m}}{2\pi i}\int_{\gamma_{r_2}}
\left(\frac{1-w_i}{w_i}\right)^{j-i}\frac{dw_i}{(1-qw_i)^{\Delta m}w_i^{y_j-x_i+1}}\right)_{1\le i,j\le n_2},
\notag
\end{align}
where $0<r_1,r_2<1$. Now, the first determinant in (\ref{2.11}) can be rewritten as
\begin{align}\label{2.12}
&\det\left(\frac{(1-q)^{m_1}}{2\pi i}\int_{\gamma_{r_1}}\left(\frac{1-z_j}{z_j}\right)^{j-i}\frac{dz_j}{(1-qz_j)^{m_1}z_j^{x_j+1}}\right)_{1\le i,j\le n_2}
\\
&=\frac{(1-q)^{m_1n_2}}{(2\pi i)^{n_2}}\int_{\gamma_{r_1}^{n_2}}d^{n_2}z\prod\limits_{j=1}^{n_2}\left(\frac{1-z_j}{z_j}\right)^{j}
\prod\limits_{j=1}^{n_2}\frac 1{(1-qz_j)^{m_1}z_j^{x_j+1}}\det\left(\left(\frac{z_j}{1-z_j}\right)^{i}\right)_{1\le i,j\le n_2}
\notag\\
&=\frac{(1-q)^{m_1n_2}}{(2\pi i)^{n_2}}\int_{\gamma_{r_1}^{n_2}}d^{n_2}z\prod\limits_{j=1}^{n_2}\left(\frac{1-z_j}{z_j}\right)^{j}
\prod\limits_{j=1}^{n_2}\frac 1{(1-qz_j)^{m_1}(1-z_j)^{n_2}z_j^{x_j}}\det\left(z_j^{i-1}\right)_{1\le i,j\le n_2}.
\notag
\end{align}
Here we used the identity
\begin{equation}
\det\left(\left(\frac{z_j}{1-z_j}\right)^{i}\right)_{1\le i,j\le n_2}=\left(\prod\limits_{j=1}^{n_2}\frac{z_j}{(1-z_j)^{n_2}}\right)
\det\left(z_j^{i-1}\right)_{1\le i,j\le n_2}.
\notag
\end{equation}
This follows from the following computation,
\begin{align}
&\det\left(\left(\frac{z_j}{1-z_j}\right)^{i}\right)_{1\le i,j\le n_2}=
\det\left(\left(\frac{z_j}{1-z_j}\right)^{i-1}\right)_{1\le i,j\le n_2}\prod\limits_{j=1}^{n_2}\frac{z_j}{1-z_j}
\notag\\
&=\prod\limits_{j=1}^{n_2}\frac{z_j}{1-z_j}\prod_{1\le i<j\le n_2}\left(\frac{z_j}{1-z_j}-\frac{z_i}{1-z_i}\right)
\notag\\
&=\prod\limits_{j=1}^{n_2}\frac{z_j}{1-z_j}\prod_{1\le i<j\le n_2}\frac 1{(1-z_j)(1-z_i)}\prod_{1\le i<j\le n_2}
\left(z_j(1-z_i)-z_i(1-z_j)\right)
\notag\\
&=\prod\limits_{j=1}^{n_2}\frac{z_j}{1-z_j}\prod\limits_{j=2}^{n_2}\frac 1{(1-z_j)^{j-1}}
\prod\limits_{i=1}^{n_2-1}\frac 1{(1-z_i)^{n_2-i}}\prod_{1\le i<j\le n_2}(z_j-z_i)
\notag\\
&=\prod\limits_{j=1}^{n_2}\frac{z_j}{1-z_j}\prod\limits_{j=1}^{n_2}\frac 1{(1-z_j)^{j-1}}\frac 1{(1-z_j)^{n_2-j}}
\det\left(z_j^{i-1}\right)_{1\le i,j\le n_2}
=
\det\left(z_j^{i-1}\right)_{1\le i,j\le n_2}\prod\limits_{j=1}^{n_2}\frac{z_j}{(1-z_j)^{n_2}}.
\notag
\end{align}

Consider now the second determinant in (\ref{2.11}) together with the $y$-summation. We get
\begin{align}\label{2.13}
&\sum\limits_{\substack{y\in W_{n_2}\\y_{n_2}\le v_2}}\det\left(\frac{(1-q)^{\Delta m}}{2\pi i}\int_{\gamma_{r_2}}
\left(\frac{1-w_i}{w_i}\right)^{j-i}\frac{dw_i}{(1-qw_i)^{\Delta m}w_i^{y_j-x_i+1}}\right)_{1\le i,j\le n_2}
\\
&=\frac{(1-q)^{\Delta mn_2}}{(2\pi i)^{n_2}}\int_{\gamma_{r_2}^{n_2}}d^{n_2}w\sum\limits_{\substack{y\in W_{n_2}\\y_{n_2}\le v_2}}
\sum\limits_{\sigma\in S_{n_2}}\sgn(\sigma)\prod\limits_{j=1}^{n_2}\left(\frac{1-w_{\sigma(j)}}{w_{\sigma(j)}}\right)^{j-\sigma(j)}
\frac{1}{(1-qw_{\sigma(j)})^{\Delta m}w_{\sigma(j)}^{y_j-x_{\sigma(j)}+1}}
\notag \\
&=\frac{(1-q)^{\Delta mn_2}}{(2\pi i)^{n_2}}\int_{\gamma_{r_2}^{n_2}}d^{n_2}w\prod\limits_{j=1}^{n_2}\left(\frac{w_j}{1-w_j}\right)^{j}
\frac{w_j^{x_j-1}}{(1-qw_j)^{\Delta m}}
\notag\\
&\times\sum\limits_{\sigma\in S_{n_2}}\sgn(\sigma)\prod\limits_{j=1}^{n_2}\left(\frac{1-w_{\sigma(j)}}{w_{\sigma(j)}}\right)^{j}
\left(\sum\limits_{\substack{y\in W_{n_2}\\y_{n_2}\le v_2}}\prod\limits_{j=1}^{n_2}\frac 1{w_{\sigma(j)}^{y_j}}\right).
\notag
\end{align}
Since $0<r_2<1$ we see, by summing the geometric series, that
\begin{align}\label{2.14}
\sum\limits_{\substack{y\in W_{n_2}\\y_{n_2}\le v_2}}\prod\limits_{j=1}^{n_2}\frac 1{w_{\sigma(j)}^{y_j}}
&=\sum\limits_{y_{n_2}=-\infty}^{v_2}(w_{\sigma(1)}\dots w_{\sigma(n_2)})^{-y_{n_2}}
\sum\limits_{y_{n_2-1}=-\infty}^{y_{n_2}}(w_{\sigma(1)}\dots w_{\sigma(n_2-1)})^{y_{n_2}-y_{n_2-1}}\dots
\sum\limits_{y_{1}=-\infty}^{y_2}w_{\sigma(1)}^{y_2-y_1}\notag\\
&=\prod\limits_{j=1}^{n_2}\frac 1{w_j^{v_2}}\frac 1{(1-w_{\sigma(1)})(1-w_{\sigma(1)}w_{\sigma(2)})\cdots (1-w_{\sigma(1)}\cdots
w_{\sigma(n_2)})}.
\end{align}
Combining (\ref{2.14}) with the identity (\ref{2.9}) we get
\begin{equation}
\sum\limits_{\sigma\in S_{n_2}}\sgn(\sigma)\prod\limits_{j=1}^{n_2}\left(\frac{1-w_{\sigma(j)}}{w_{\sigma(j)}}\right)^{j}
\left(\sum\limits_{\substack{y\in W_{n_2}\\y_{n_2}\le v_2}}\prod\limits_{j=1}^{n_2}\frac 1{w_{\sigma(j)}^{y_j}}\right)
=(-1)^{\frac{n_2(n_2-1)}2}\prod\limits_{j=1}^{n_2}\frac 1{w_j^{n_2+v_2}}\det\left(w_j^{i-1}\right)_{1\le i,j\le n_2}.\notag
\end{equation}
We can now use this identity in (\ref{2.13}) and obtain
\begin{align}\label{2.15}
&\sum\limits_{\substack{y\in W_{n_2}\\y_{n_2}\le v_2}}\det\left(\frac{(1-q)^{\Delta m}}{2\pi i}\int_{\gamma_{r_2}}
\left(\frac{1-w_i}{w_i}\right)^{j-i}\frac{dw_i}{(1-qw_i)^{\Delta m}w_i^{y_j-x_i+1}}\right)_{1\le i,j\le n_2}
\\
&=\frac{(1-q)^{\Delta mn_2}(-1)^{\frac{n_2(n_2-1)}2}}{(2\pi i)^{n_2}}\int_{\gamma_{r_2}^{n_2}}d^{n_2}w
\prod\limits_{j=1}^{n_2}\left(\frac{w_j}{1-w_j}\right)^{j}
\frac{w_j^{x_j-v_2-n_2-1}}{(1-qw_j)^{\Delta m}}\det\left(w_j^{i-1}\right)_{1\le i,j\le n_2}.
\notag
\end{align}
Next, insert (\ref{2.12}) and (\ref{2.15}) into (\ref{2.11}) to get
\begin{align}\label{2.16}
P&=\sum\limits_{u=-\infty}^{v_1}\sum\limits_{\substack{x\in W_{n_2}\\x_{n_1}=u}}
\frac{(1-q)^{m_2n_2}(-1)^{\frac{n_2(n_2-1)}2}}{(2\pi i)^{2n_2}}
\int_{\gamma_{r_1}^{n_2}}d^{n_2}z\int_{\gamma_{r_2}^{n_2}}d^{n_2}w
\det\left(z_j^{i-1}\right)_{1\le i,j\le n_2}\det\left(w_j^{i-1}\right)_{1\le i,j\le n_2}
\\
&\times \prod\limits_{j=1}^{n_2}\left(\frac{1-z_j}{z_j}\right)^{j}\left(\frac{w_j}{1-w_j}\right)^{j}
\frac 1{(1-qz_j)^{m_1}(1-z_j)^{n_2}z_j^{x_j}}\frac{w_j^{x_j-v_2-n_2-1}}{(1-qw_j)^{\Delta m}}.
\notag
\end{align}
In this expression we symmetrize in $\{z_j\}$ and  $\{w_j\}$. We find
\begin{align}\label{2.17}
P&=\sum\limits_{u=-\infty}^{v_1}\sum\limits_{\substack{x\in W_{n_2}\\x_{n_1}=u}}
\frac{(1-q)^{m_2n_2}(-1)^{\frac{n_2(n_2-1)}2}}{(2\pi i)^{2n_2}(n_2!)^2}
\int_{\gamma_{r_1}^{n_2}}d^{n_2}z\int_{\gamma_{r_2}^{n_2}}d^{n_2}w
\det\left(z_j^{i-1}\right)_{1\le i,j\le n_2}\det\left(w_j^{i-1}\right)_{1\le i,j\le n_2}
\\
&\times \prod\limits_{j=1}^{n_2}
\frac 1{(1-qz_j)^{m_1}(1-z_j)^{n_2}w_j^{v_2+n_2+1}(1-qw_j)^{\Delta m}}
\notag\\
&\times
\left(\sum\limits_{\sigma_1,\sigma_2\in S_n}\sgn(\sigma_1\sigma_2)\prod\limits_{j=1}^{n_2}\left(\frac{1-z_{\sigma_1(j)}}{z_{\sigma_1(j)}}\right)^j
\left(\frac{w_{\sigma_2(j)}}{1-w_{\sigma_2(j)}}\right)^j\left(\frac{w_{\sigma_2(j)}}{z_{\sigma_1(j)}}\right)^{x_j}\right).
\notag
\end{align}

Let $(S_j^-,S_j^+)$, $j=1,2$, be two partitions of $[1,n_2]=\{1,\dots n_2\}$, such that $|S_j^-|=n_1$ and  $|S_j^+|=\Delta n$.
For $\sigma_1,\sigma_2\in S_{n_2}$, we say that $\sigma_j\in S_{n_2}(S_j^-,S_j^+)$ if $\sigma_j([1,n_1])=S_j^-$ and
consequently $\sigma_j([n_1+1,n_2])=S_j^+$, $j=1,2$. Write
$$
\sigma_j^-=\sigma_j\left|_{[1,n_1]}\right.,\,\,\sigma_j^+=\sigma_j\left|_{[n_1+1,n_2]}\right.,\,\,j=1,2,
$$
for the restricted permutations. Given $S_j^-\,(S_j^+)$ we can identify $\sigma_j^-\,(\sigma_j^+)$ with a permutation
in $S_{n_1}\,(S_{\Delta n})$.
We do this by taking the order preserving bijection from $S_j^-\,(S_j^+)$ to $S_{n_1}\,(S_{\Delta n})$.
The signs are related by
\begin{equation}\label{2.18}
\sgn(\sigma_j)=(-1)^{\kappa(S_j^-,S_j^+)}\sgn(\sigma_j^-)\sgn(\sigma_j^+),
\end{equation}
where
$$
\kappa(U,V)=|\{(i,j)\,;\,i\in U,j\in V,i>j\}|.
$$
We will now choose our radii in the circles in the contour integrals depending on $S_j^\pm$. Recall that
$$
0<s_1<r_1<1,\,\,0<r_2<s_2<1,
$$
which we assumed in the proposition. Given $S_j^-$, $j=1,2$, we take
\begin{align}\label{2.19}
&|z_k|=s_1,\,k\in S_1^-,\,\,|w_k|=r_1,\,k\in S_2^-,
\\
&|z_k|=s_2,\,k\in S_1^+,\,\,|w_k|=r_2,\,k\in S_2^+.
\notag
\end{align}
We can write
\begin{align}\label{2.20}
&\prod\limits_{j=1}^{n_2}\left(\frac{1-z_{\sigma_1(j)}}{z_{\sigma_1(j)}}\right)^j
\left(\frac{w_{\sigma_2(j)}}{1-w_{\sigma_2(j)}}\right)^j\left(\frac{w_{\sigma_2(j)}}{z_{\sigma_1(j)}}\right)^{x_j}
\\
&=\prod\limits_{j=1}^{n_1}\left(\frac{1-z_{\sigma_1^-(j)}}{z_{\sigma_1^-(j)}}\right)^j
\left(\frac{w_{\sigma_2^-(j)}}{1-w_{\sigma_2^-(j)}}\right)^j\left(\frac{w_{\sigma_2^-(j)}}{z_{\sigma_1^-(j)}}\right)^{x_j}
\prod\limits_{j=n_1+1}^{n_2}\left(\frac{1-z_{\sigma_1^+(j)}}{z_{\sigma_1^+(j)}}\right)^j
\left(\frac{w_{\sigma_2^+(j)}}{1-w_{\sigma_2^+(j)}}\right)^j\left(\frac{w_{\sigma_2^+(j)}}{z_{\sigma_1^+(j)}}\right)^{x_j}
\notag
\end{align}
From (\ref{2.17}), (\ref{2.18}) and (\ref{2.20}) we find
\begin{align}\label{2.21}
P&=\sum\limits_{S_1^-,S_2^-}(-1)^{\kappa(S_1^-,S_1^+)+\kappa(S_2^-,S_2^+)}\sum\limits_{u=-\infty}^{v_1}
\frac{(1-q)^{m_2n_2}(-1)^{\frac{n_2(n_2-1)}2}}{(2\pi i)^{2n_2}(n_2!)^2}
\\
&\times\int_{\gamma_{s_1}^{n_1}}\prod\limits_{j\in S_1^-}dz_j\int_{\gamma_{s_2}^{\Delta n}}\prod\limits_{j\in S_1^+}dz_j
\int_{\gamma_{r_1}^{n_1}}\prod\limits_{j\in S_2^-}dw_j\int_{\gamma_{r_2}^{\Delta n}}\prod\limits_{j\in S_2^+}dw_j
\notag\\
&\times\det\left(z_j^{i-1}\right)_{1\le i,j\le n_2}\det\left(w_j^{i-1}\right)_{1\le i,j\le n_2}\prod\limits_{j=1}^{n_2}
\frac 1{(1-qz_j)^{m_1}(1-z_j)^{n_2}w_j^{v_2+n_2+1}(1-qw_j)^{\Delta m}}
\notag\\
&\times\sum\limits_{\substack{x\in W_{n_2}\\x_{n_1}=u}}\left(\sum_{\sigma_1^-,\sigma_2^-}\sgn(\sigma_1^-)\sgn(\sigma_2^-)
\prod\limits_{j=1}^{n_1}\left(\frac{1-z_{\sigma_1^-(j)}}{z_{\sigma_1^-(j)}}\right)^j
\left(\frac{w_{\sigma_2^-(j)}}{1-w_{\sigma_2^-(j)}}\right)^j\left(\frac{w_{\sigma_2^-(j)}}{z_{\sigma_1^-(j)}}\right)^{x_j}\right)
\notag\\
&\times\left(\sum_{\sigma_1^+,\sigma_2^+}\sgn(\sigma_1^+)\sgn(\sigma_2^+)\prod\limits_{j=n_1+1}^{n_2}\left(\frac{1-z_{\sigma_1^+(j)}}{z_{\sigma_1^+(j)}}\right)^j
\left(\frac{w_{\sigma_2^+(j)}}{1-w_{\sigma_2^+(j)}}\right)^j\left(\frac{w_{\sigma_2^+(j)}}{z_{\sigma_1^+(j)}}\right)^{x_j}\right).
\notag
\end{align}

The next step is to do the $x$-summations,
\begin{align}\label{2.22}
&\sum\limits_{x_1\le\dots\le x_{n_1-1}\le x_{n_1}=u}\prod\limits_{j=1}^{n_1}\left(\frac{w_{\sigma_2^-(j)}}{z_{\sigma_1^-(j)}}\right)^{x_j}\\
&=\prod\limits_{j=1}^{n_1}\left(\frac{w_{\sigma_2^-(j)}}{z_{\sigma_1^-(j)}}\right)^{u}
\sum\limits_{y_1\le\dots\le y_{n_1-1}\le u}\left(\frac{z_{\sigma_1^-(j)}}{w_{\sigma_2^-(j)}}\right)^{-y_j}\notag\\
&=\prod\limits_{j=1}^{n_1}\left(\frac{w_{\sigma_2^-(j)}}{z_{\sigma_1^-(j)}}\right)^u
\frac 1{\left(1-\frac{z_{\sigma_1^-(1)}}{w_{\sigma_2^-(1)}}\right)\left(1-\frac{z_{\sigma_1^-(1)}z_{\sigma_1^-(2)}}{w_{\sigma_2^-(1)}w_{\sigma_2^-(2)}}\right)\cdots
\left(1-\frac{z_{\sigma_1^-(1)}\cdots z_{\sigma_1^-(n_1-1)}}{w_{\sigma_2^-(1)}\cdots w_{\sigma_2^-(n_1-1)}}\right)},
\notag
\end{align}
by the same computation as (\ref{2.14}). Here we used the fact that $|z_{\sigma_1^-(j)}/w_{\sigma_2^-(j)}|=s_1/r_1<1$. Similarly,
\begin{align}\label{2.23}
&\sum\limits_{u\le x_{n_1+1}\le\dots\le x_{n_2}} \prod\limits_{j=n_1+1}^{n_2}\left(\frac{w_{\sigma_2^+(j)}}{z_{\sigma_1^+(j)}}\right)^{x_j}\\
&=\prod\limits_{j=n_1+1}^{n_2}\left(\frac{w_{\sigma_2^+(j)}}{z_{\sigma_1^+(j)}}\right)^u
\frac 1{\left(1-\frac{w_{\sigma_2^+(n_2)}}{z_{\sigma_1^+(n_2)}}\right)\left(1-\frac{w_{\sigma_2^+(n_2)}w_{\sigma_2^+(n_2-1)}}{z_{\sigma_1^+(n_2)}z_{\sigma_1^+(n_2-1)}}\right)\cdots
\left(1-\frac{w_{\sigma_2^+(n_2)}\cdots w_{\sigma_2^+(n_1+1)}}{z_{\sigma_1^+(n_2)}\cdots z_{\sigma_1^+(n_1+1)}}\right)},
\notag
\end{align}
since $|w_{\sigma_2^+(j)}/z_{\sigma_1^+(j)}|=r_2/s_2<1$.

We can now apply (\ref{2.10}) in lemma \ref{lem2.3} to see that
\begin{align}\label{2.24}
&\sum_{\sigma_1^-,\sigma_2^-}\sgn(\sigma_1^-)\sgn(\sigma_2^-)
\prod\limits_{j=1}^{n_1}\left(\frac{w_{\sigma_2^-(j)}(1-z_{\sigma_1^-(j)})}{z_{\sigma_1^-(j)}(1-w_{\sigma_2^-(j)})}\right)^j\\
&\times
\frac 1{\left(1-\frac{z_{\sigma_1^-(1)}}{w_{\sigma_2^-(1)}}\right)\left(1-\frac{z_{\sigma_1^-(1)}z_{\sigma_1^-(2)}}{w_{\sigma_2^-(1)}w_{\sigma_2^-(2)}}\right)\cdots
\left(1-\frac{z_{\sigma_1^-(1)}\cdots z_{\sigma_1^-(n_1-1)}}{w_{\sigma_2^-(1)}\cdots w_{\sigma_2^-(n_1-1)}}\right)}
\notag\\
&=\left(1-\prod\limits_{j\in S_1^-}z_j\prod\limits_{j\in S_2^-}\frac 1{w_j}\right)
\prod\limits_{j\in S_1^-}\frac{(1-z_j)^{n_1}}{z_j^{n_1}}\prod\limits_{j\in S_2^-}\frac{w_j^{n_1+1}}{(1-w_j)^{n_1}}
\det\left(\frac 1{w_j-z_i}\right)_{\substack{i\in S_1^-\\j\in S_2^-}}.
\notag
\end{align}
From (\ref{2.23}) we see that we also want to compute
\begin{align}\label{2.25}
&\sum_{\sigma_1^+,\sigma_2^+}\sgn(\sigma_1^+)\sgn(\sigma_2^+)
\prod\limits_{j=n_1+1}^{n_2}\left(\frac{w_{\sigma_2^+(j)}(1-z_{\sigma_1^+(j)})}{z_{\sigma_1^+(j)}(1-w_{\sigma_2^+(j)})}\right)^j\\
&\times
\frac 1{\left(1-\frac{w_{\sigma_2^+(n_2)}}{z_{\sigma_1^+(n_2)}}\right)\left(1-\frac{w_{\sigma_2^+(n_2)}w_{\sigma_2^+(n_2-1)}}{z_{\sigma_1^+(n_2)}z_{\sigma_1^+(n_2-1)}}\right)\cdots
\left(1-\frac{w_{\sigma_2^+(n_2)}\cdots w_{\sigma_2^+(n_1+1)}}{z_{\sigma_1^+(n_2)}\cdots z_{\sigma_1^+(n_1+1)}}\right)}.
\notag
\end{align}
Let $\tau(j)=n_2+1-j$, $1\le j\le\Delta n$ and $\tilde{\sigma}^+_i=\sigma^+_i\circ\tau$, $i=1,2$. Then,
$\tilde{\sigma}^+_i\,:\,[1,\Delta n]\to S_i^+$ and
\begin{equation*}
\sgn(\tilde{\sigma}^+_1)\sgn(\tilde{\sigma}^+_2)=\sgn(\sigma^+_1)\sgn(\sigma^+_2).
\end{equation*}
Also,
\begin{align*}
&\prod\limits_{j=n_1+1}^{n_2}\left(\frac{w_{\sigma_2^+(j)}(1-z_{\sigma_1^+(j)})}{z_{\sigma_1^+(j)}(1-w_{\sigma_2^+(j)})}\right)^j
=\prod\limits_{j=1}^{\Delta n}\left(\frac{w_{\tilde{\sigma}_2^+(j)}(1-z_{\tilde{\sigma}_1^+(j)})}{z_{\tilde{\sigma}_1^+(j)}(1-w_{\tilde{\sigma}_2^+(j)})}\right)^{n_2+1-j}\\
&=\prod\limits_{j\in S_1^+}\left(\frac{1-z_j}{z_j}\right)^{n_2+1}\prod\limits_{j\in S_2^+}\left(\frac{w_j}{1-w_j}\right)^{n_2+1}
\prod\limits_{j=1}^{\Delta n}\left(\frac{z_{\tilde{\sigma}_1^+(j)}(1-w_{\tilde{\sigma}_2^+(j)})}{w_{\tilde{\sigma}_2^+(j)}(1-z_{\tilde{\sigma}_1^+(j)})}\right)^j.
\end{align*}
Thus (\ref{2.25}) can be written
\begin{align*}
&\prod\limits_{j\in S_1^+}\left(\frac{1-z_j}{z_j}\right)^{n_2+1}\prod\limits_{j\in S_2^+}\left(\frac{w_j}{1-w_j}\right)^{n_2+1}
\sum_{\tilde{\sigma}_1^+,\tilde{\sigma}_2^+}\sgn(\tilde{\sigma}_1^+)\sgn(\tilde{\sigma}_2^+)
\prod\limits_{j=1}^{\Delta n}\left(\frac{z_{\tilde{\sigma}_1^+(j)}(1-w_{\tilde{\sigma}_2^+(j)})}{w_{\tilde{\sigma}_2^+(j)}(1-z_{\tilde{\sigma}_1^+(j)})}\right)^j\\
&\times
\frac 1{\left(1-\frac{w_{\tilde{\sigma}_2^+(n_2)}}{z_{\tilde{\sigma}_1^+(n_2)}}\right)
\left(1-\frac{w_{\tilde{\sigma}_2^+(n_2)}w_{\tilde{\sigma}_2^+(n_2-1)}}{z_{\tilde{\sigma}_1^+(n_2)}z_{\tilde{\sigma}_1^+(n_2-1)}}\right)\cdots
\left(1-\frac{w_{\tilde{\sigma}_2^+(n_2)}\cdots w_{\tilde{\sigma}_2^+(n_1+1)}}{z_{\tilde{\sigma}_1^+(n_2)}\cdots z_{\tilde{\sigma}_1^+(n_1+1)}}\right)}
\end{align*}
and by (\ref{2.10}) in lemma \ref{lem2.3} this equals
\begin{align*}
&\prod\limits_{j\in S_1^+}\left(\frac{1-z_j}{z_j}\right)^{n_2+1}\prod\limits_{j\in S_2^+}\left(\frac{w_j}{1-w_j}\right)^{n_2+1}
\prod\limits_{j\in S_1^+}\frac{z_j^{\Delta n+1}}{(1-z_j)^{\Delta n}}\prod\limits_{j\in S_2^+}\left(\frac{1-w_j}{w_j}\right)^{\Delta n}
\det\left(\frac 1{z_i-w_j}\right)_{\substack{i\in S_1^+\\j\in S_2^+}}\\
&=\prod\limits_{j\in S_1^+}\frac{(1-z_j)^{n_1+1}}{z_j^{n_1}}\prod\limits_{j\in S_2^+}\frac{w_j^{n_1+1}}{(1-w_j)^{n_1+1}}
\det\left(\frac 1{z_i-w_j}\right)_{\substack{i\in S_1^+\\j\in S_2^+}}.
\end{align*}
Using this, (\ref{2.22}), (\ref{2.23}) and (\ref{2.24}) in  (\ref{2.21}) we find
\begin{align}\label{2.26}
P&=\sum\limits_{S_1^-,S_2^-}(-1)^{\kappa(S_1^-,S_1^+)+\kappa(S_2^-,S_2^+)}\sum\limits_{u=-\infty}^{v_1}
\frac{(1-q)^{m_2n_2}(-1)^{\frac{n_2(n_2-1)}2}}{(2\pi i)^{2n_2}(n_2!)^2}
\\
&\times\int_{\gamma_{s_1}^{n_1}}\prod\limits_{j\in S_1^-}dz_j\int_{\gamma_{s_2}^{\Delta n}}\prod\limits_{j\in S_1^+}dz_j
\int_{\gamma_{r_1}^{n_1}}\prod\limits_{j\in S_2^-}dw_j\int_{\gamma_{r_2}^{\Delta n}}\prod\limits_{j\in S_2^+}dw_j
\notag\\
&\times\det\left(z_j^{i-1}\right)_{1\le i,j\le n_2}\det\left(w_j^{i-1}\right)_{1\le i,j\le n_2}
\det\left(\frac 1{w_j-z_i}\right)_{\substack{i\in S_1^-\\j\in S_2^-}}\det\left(\frac 1{z_i-w_j}\right)_{\substack{i\in S_1^+\\j\in S_2^+}}
\notag\\
&\times \left(1-\prod\limits_{j\in S_1^-}z_j\prod\limits_{j\in S_2^-}\frac 1{w_j}\right)
\prod\limits_{j\in S_1^-}\frac{(1-z_j)^{n_1}}{z_j^{n_1}}\prod\limits_{j\in S_2^-}\frac{w_j^{n_1+1}}{(1-w_j)^{n_1}}
\prod\limits_{j\in S_1^+}\frac{(1-z_j)^{n_1+1}}{z_j^{n_1}}\prod\limits_{j\in S_2^+}\frac{w_j^{n_1+1}}{(1-w_j)^{n_1+1}}
\notag\\
&\times \prod\limits_{j=1}^{n_2}\frac{w_j^u}{z_j^u}\prod\limits_{j=1}^{n_2}
\frac 1{(1-qz_j)^{m_1}(1-z_j)^{n_2}w_j^{v_2+n_2+1}(1-qw_j)^{\Delta m}}.
\notag
\end{align}

To see that the summation over $S_1^-,S_2^-$ in (\ref{2.26}) is actually trivial, in the sense that the summand does not depend on the choice
of $S_1^-,S_2^-$, we use the following observation.
Write
$$
\Delta_S(z)=\prod\limits_{j<k,j,k\in S}(z_k-z_j)
$$
for $S\subseteq [1,n_2]$. Then, by the standard formula for a Vandermonde determinant,
\begin{equation*}
\det\left(z_j^{i-1}\right)_{1\le i,j\le n_2}=\prod\limits_{1\le j<k\le n_2}(z_k-z_j)
=\Delta_{S_1^-}(z)\Delta_{S_1^+}(z)\prod\limits_{j\in S_1^-, k\in S_1^+}(z_k-z_j)
(-1)^{\kappa(S_1^-,S_1^+)}.
\end{equation*}
If we insert this into (\ref{2.26}) for both $z$ and $w$ we see that we can relabel the indices
\begin{align}
&(z_j)_{j\in S_1^-}\to (z_j)_{j=1}^{n_1}\,\,,\,\,(z_j)_{j\in S_1^+}\to (z_j)_{j=n_1+1}^{n_2}
\notag\\
&(w_j)_{j\in S_2^-}\to (z_j)_{j=1}^{n_1}\,\,,\,\,(w_j)_{j\in S_2^+}\to (w_j)_{j=n_1+1}^{n_2}
\notag
\end{align}
and then the sums over $S_1^-,S_2^-$ become trivial. Note that
$$
\sum\limits_{S_i^-}1=\binom{n_2}{n_1},
$$
$i=1,2$. Formula (\ref{2.26}) then reduces to (\ref{2.7}) and we have proved the proposition.
\end{proof}

From proposition \ref{prop2.2} we can, by a limiting procedure, obtain a corresponding formula in the Brownian directed
polymer model.

Let $\Gamma_d$ denote the vertical straight line contour through $d\in\mathbb{R}$ oriented upwards, $\Gamma_d:t\to d+it$, $t\in\mathbb{R}$.
Define
\begin{align}\label{2.27}
&Q(h)=\frac{(-1)^{\frac{n_2(n_2-1)}2}}{(2\pi i)^{2n_2}n_1!\Delta n!}\int_{\Gamma_{d_1}^{n_1}}d^{n_1}z\int_{\Gamma_{d_2}^{\Delta n}}d^{\Delta n}z
\int_{\Gamma_{d_3}^{n_1}}d^{n_1}w\int_{\Gamma_{d_4}^{\Delta n}}d^{\Delta n}w\det\left(z_j^{i-1}\right)_{1\le i,j\le n_2}\det\left(w_j^{i-1}\right)_{1\le i,j\le n_2}
\\
&\times\prod\limits_{j=1}^{n_1}\frac{e^{\frac 12\mu_1z_j^2-\xi_1z_j+\frac 12\Delta\mu w_j^2-\Delta\xi w_j}}{z_j^{\Delta n}w_j^{n_1}}\left(\frac 1{z_j-w_j}-h\right)
\prod\limits_{j=n_1+1}^{n_2}\frac{e^{\frac 12\mu_1z_j^2-\xi_1z_j+\frac 12\Delta\mu w_j^2-\Delta\xi w_j}}{z_j^{\Delta n-1}w_j^{n_1+1}(w_j-z_j)}
\notag
\end{align}
where
\begin{equation}\label{2.28}
d_1<d_3<0\,\,,\,\,d_4<d_2<0.
\end{equation}
Here, we have written
\begin{equation}\label{2.29}
\Delta\xi=\xi_2-\xi_1\,\,,\,\,\Delta\mu=\mu_2-\mu_1.
\end{equation}
We can now state a proposition concerning the joint distribution function that we are interested in in theorem \ref{thm1.1}.

\begin{proposition}\label{prop2.4}
Let $H(\mu,n)$ be defined by (\ref{1.1}). Then
\begin{equation}\label{2.30}
\frac{\partial}{\partial\xi_1}\mathbb{P}\left[H(\mu_1,n_1)\le\xi_1,H(\mu_2,n_2)\le\xi_2\right]=\left.\frac{\partial}{\partial h}\right|_{h=0}Q(h).
\end{equation}
\end{proposition}

\begin{proof} Just as in (\ref{1.20}) we have the formula
\begin{align}\label{2.31}
&\mathbb{P}\left[H(\mu_1,n_1)\le\xi_1,H(\mu_2,n_2)\le\xi_2\right]
\\
&=\lim\limits_{T\to\infty}\mathbb{P}\left[G([\mu_1T],n_1)\le \frac{q}{1-q}[\mu_1T]+\xi_1\frac{\sqrt{q}}{1-q}\sqrt{T},
G([\mu_2T],n_2)\le \frac{q}{1-q}[\mu_2T]+\xi_2\frac{\sqrt{q}}{1-q}\sqrt{T}\right].
\notag
\end{align}
In the formula (\ref{2.7}) we assume that we have chosen $r_1,r_2, s_1, s_2$ so that
\begin{equation}\label{2.32}
(r_1/s_1)^{n_1}>(s_2/r_2)^{\Delta n},
\end{equation}
which can always be done for fixed $n_1, n_2$. We can then do the $u$-summation in (\ref{2.7}) to get
\begin{equation}\label{2.33}
\sum\limits_{u=-\infty}^{v_1}\left(\prod\limits_{j=1}^{n_2}\frac{w_j}{z_j}\right)^u=\frac{\prod\limits_{j=1}^{n_2}w_j^{v_1}/z_j^{v_1}}
{1-\prod\limits_{j=1}^{n_2}z_j/w_j}.
\end{equation}
Insert this into (\ref{2.7}), expand the Cauchy determinants and symmetrize. This gives
\begin{align}\label{2.34}
&P=\frac{(1-q)^{m_2n_2}(-1)^{n_2(n_2-1)/2}}{(2\pi i)^{2n_2}n_1!(\Delta n)!}
\int\limits_{\gamma_{s_1}^{n_1}}d^{n_1}z\int\limits_{\gamma_{s_2}^{\Delta n}}d^{\Delta n}z
\int\limits_{\gamma_{r_1}^{n_1}}d^{n_1}w\int\limits_{\gamma_{r_2}^{\Delta n}}d^{\Delta n}w \frac{1-\prod\limits_{j=1}^{n_1}z_j/w_j}
{1-\prod\limits_{j=1}^{n_2}z_j/w_j}\\
&\times\det \left((z_j-1)^{i-1}\right)_{1\le i,j\le n_2}\det \left((w_j-1)^{i-1}\right)_{1\le i,j\le n_2}
\prod\limits_{j=1}^{n_1}\frac 1{w_j-z_j}\prod\limits_{j=n_1+1}^{n_2}\frac 1{z_j-w_j}\notag\\
&\times\prod\limits_{j=n_1+1}^{n_2}\frac{1-z_j}{1-w_j}
\prod\limits_{j=1}^{n_2}\frac{1}
{z_j^{v_1+n_1}(1-z_j)^{\Delta n}(1-qz_j)^{m_1}w_j^{v_2-v_1+\Delta n}(1-w_j)^{n_1}(1-qw_j)^{\Delta m}}.
\notag
\end{align}
Here, we have also used the fact that
$
\det \left(z_j^{i-1}\right)=\det \left((z_j-1)^{i-1}\right),
$
by the standard product formula for a Vandermonde determinant.
We now want to take the limit in (\ref{2.31}) using the formula (\ref{2.34}), i.e. we let
\begin{equation}\label{2.36}
m_i=[\mu_iT]\,\,,\,\,v_i=\frac{q}{1-q}[\mu_iT]+\xi_i\frac{\sqrt{q}}{1-q}\sqrt{T},
\end{equation}
$i=1,2$. Let $\Gamma_d^{(T)}$ be given by $t\to d+it$, $|t|\le\pi(1-q)^{-1}\sqrt{qT}$. In (\ref{2.34}) we make the change of variables
\begin{align}\label{2.37}
z_j&=e^{(1-q)z_j'/\sqrt{qT}}\,\,,\,\,z_j'\in\Gamma_{d_1}^{(T)}\,\,,\,\,1\le j\le n_1,
\\
z_j&=e^{(1-q)z_j'/\sqrt{qT}}\,\,,\,\,z_j'\in\Gamma_{d_2}^{(T)}\,\,,\,\,n_1+1\le j\le n_2,
\notag\\
w_j&=e^{(1-q)w_j'/\sqrt{qT}}\,\,,\,\,w_j'\in\Gamma_{d_3}^{(T)}\,\,,\,\,1\le j\le n_1,
\notag\\
w_j&=e^{(1-q)w_j'/\sqrt{qT}}\,\,,\,\,w_j'\in\Gamma_{d_4}^{(T)}\,\,,\,\,n_1+1\le j\le n_2,
\notag
\end{align}
where the $d_i$ satisfy (\ref{2.28}).

The condition (\ref{2.32}) becomes
\begin{equation}\label{2.38}
n_1d_3+\Delta n d_4>n_1d_1+\Delta n d_2.
\end{equation}
From (\ref{2.37}) it follows using Taylor expansions and (\ref{2.36}) that
\begin{equation*}
\lim\limits_{T\to\infty}\frac{(1-qz_j)^{m_1}z_j^{v_1+n_1}}{(1-q)^{m_1}}=e^{-\frac 12\mu_1 z_j'^2+\xi_1z_j'},
\end{equation*}
\begin{equation*}
\lim\limits_{T\to\infty}\frac{\sqrt{qT}}{1-q}(z_j-1)=z_j',
\end{equation*}
and similar limits involving $w_j$ instead, and
\begin{equation*}
\lim\limits_{T\to\infty}\frac{1-\prod_{j=1}^{n_1}z_j/w_j}{1-\prod_{j=1}^{n_2}z_j/w_j}=\frac{\sum_{j=1}^{n_1}(w_j'-z_j')}{\sum_{j=1}^{n_2}(w_j'-z_j')}.
\end{equation*}

If we insert (\ref{2.37}) into (\ref{2.34}) and use these limits we can take the limit $T\to\infty$. To make the argument complete we also need some
estimates, but we omit the details. After some computation we find, using 
(\ref{2.31}), that (we have dropped the primes on the $z$- and $w$-variables)
\begin{align}\label{2.39}
&\mathbb{P}\left[H(\mu_1,n_1)\le\xi_1,H(\mu_2,n_2)\le\xi_2\right]\notag\\
&=\frac{(-1)^{n_2(n_2-1)/2}}{(2\pi i)^{2n_2}n_1!(\Delta n)!}\int_{\Gamma_{d_1}^{n_1}}d^{n_1}z\int_{\Gamma_{d_2}^{\Delta n}}d^{\Delta n}z
\int_{\Gamma_{d_3}^{n_1}}d^{n_1}w\int_{\Gamma_{d_4}^{\Delta n}}d^{\Delta n}w\det\left(z_j^{i-1}\right)_{1\le i,j\le n_2}\det\left(w_j^{i-1}\right)_{1\le i,j\le n_2}
\\
&\times \frac{\sum\limits_{j=1}^{n_1}w_j-z_j}{\sum\limits_{j=1}^{n_2}w_j-z_j}
\prod\limits_{j=1}^{n_1}\frac{e^{\frac 12\mu_1z_j^2-\xi_1z_j+\frac 12\Delta\mu w_j^2-\Delta\xi w_j}}{z_j^{\Delta n}w_j^{n_1}(z_j-w_j)}
\prod\limits_{j=n_1+1}^{n_2}\frac{e^{\frac 12\mu_1z_j^2-\xi_1z_j+\frac 12\Delta\mu w_j^2-\Delta\xi w_j}}{z_j^{\Delta n-1}w_j^{n_1+1}(w_j-z_j)}
\notag
\end{align}
From (\ref{2.27}) and (\ref{2.39}) we see that (\ref{2.30}) follows (recall that $\Delta\xi=\xi_2-\xi_1$). Note that in $Q$ the condition
(\ref{2.38}) is no longer important. This completes the proof.
\end{proof}

\section{Expansion}\label{sect3}

In order to use the formula (\ref{2.30}) to prove theorem \ref{thm1.1} we must rewrite $Q(h)$ given in (\ref{2.27}) further
so that we can expand it in a way appropriate for the asymptotic analysis. This expansion is similar in some ways to writing a
distribution function like (\ref{1.3}) as a Fredholm expansion. 
Behind this expansion there is a certain orthogonality related
to the orthogonality of the Hermite polynomials. However, this orthogonality is seen at the level of the generating function for
the Hermite polynomials. We will prove a lemma which is the first step towards the expansion and which uses an integral formula for 
the Hermite polynomials.

\begin{lemma}\label{lem3.1}
The function  $Q(h)$ defined by (\ref{2.27}) is also given by
\begin{align}\label{3.1}
&Q(h)=\frac{1}{(2\pi i)^{4n_2}n_1!(\Delta n)!}\int_{\Gamma_{d_1}^{n_1}}d^{n_1}z\int_{\Gamma_{d_2}^{\Delta n}}d^{\Delta n}z
\int_{\Gamma_{d_3}^{n_1}}d^{n_1}w\int_{\Gamma_{d_4}^{\Delta n}}d^{\Delta n}w\int_{\gamma_{\tau_1}^{n_2}}d^{n_2}\zeta
\int_{\gamma_{\tau_2}^{n_2}}d^{n_2}\omega
\\
&\times\det\left(\frac 1{\zeta_j^i}\right)_{1\le i,j\le n_2}\det\left(\frac 1{\omega_j^{n_2+1-i}}\right)_{1\le i,j\le n_2}
\notag\\
&\times
\prod\limits_{j=1}^{n_1}\frac{z_j^{n_1}w_j^{\Delta n}e^{\frac 12\mu_1z_j^2-\xi_1z_j+\frac 12\Delta\mu w_j^2-\Delta\xi w_j}}
{e^{\frac 12\mu_1\zeta_j^2-\xi_1\zeta_j+\frac 12\Delta\mu \omega_j^2-\Delta\xi \omega_j}(\zeta_j-z_j)(\omega_j-w_j)}\left(
\frac 1{z_j-w_j}-h\right)
\notag\\
&
\times\prod\limits_{j=n_1+1}^{n_2}\frac{z_j^{n_1+1}w_j^{\Delta n-1}e^{\frac 12\mu_1z_j^2-\xi_1z_j+\frac 12\Delta\mu w_j^2-\Delta\xi w_j}}
{e^{\frac 12\mu_1\zeta_j^2-\xi_1\zeta_j+\frac 12\Delta\mu \omega_j^2-\Delta\xi \omega_j}(w_j-z_j)(\zeta_j-z_j)(\omega_j-w_j)},
\notag
\end{align}
where
\begin{equation}\label{3.1'}
d_1<d_3<-\max(\tau_1,\tau_2)<0\,\,,\,\,d_4<d_2<-\max(\tau_1,\tau_2)<0.
\end{equation}
\end{lemma}

\begin{proof}
Reversing the order of the rows in a determinant of size $n$ gives a sign factor $(-1)^{n(n-1)/2}$. If we do this in a Vandermonde determinant
we get the identity
\begin{equation}\label{vandermondeidentity}
\det \left(z_j^{i-1}\right)_{1\le i,j\le n_2}=(-1)^{\frac{n_2(n_2-1)}2}\det \left(z_j^{n_2-i}\right)_{1\le i,j\le n_2}
=(-1)^{\frac{n_2(n_2-1)}2}\left(\prod\limits_{j=1}^{n_2}z_j^{n_2}\right)
\det \left(\frac 1{z_j^i}\right)_{1\le i,j\le n_2}.
\end{equation}
Provided $\re z_j<0$, we see that for $i\ge 1$
$$
\int_0^\infty \frac{u_j^{i-1} e^{u_jz_j}}{(i-1)!}\,du_j=\frac {(-1)^i}{z_j^i}.
$$
It follows from these two identities that
\begin{align}
&\det \left(z_j^{i-1}\right)_{1\le i,j\le n_2}=(-1)^{\frac{n_2(n_2-1)}2}\prod\limits_{j=1}^{n_2}z_j^{n_2}
\det\left(\frac{(-1)^i}{(i-1)!}\int_0^\infty u_j^{i-1} e^{u_jz_j}\,du_j\right)_{1\le i,j\le n_2}
\notag\\
&=(-1)^{\frac{n_2(n_2-1)}2}\prod\limits_{j=1}^{n_2}z_j^{n_2}\det\left(\frac{(-1)^i}{(i-1)!}\int_{-\infty}^a (a-x_j)^{i-1} e^{(a-x_j)z_j}\,dx_j\right)_{1\le i,j\le n_2}
\notag\\
&=(-1)^{\frac{n_2(n_2+1)}2}\prod\limits_{j=1}^{n_2}z_j^{n_2}\int\limits_{(-\infty,a]^{n_2}}\prod\limits_{j=1}^{n_2}e^{(a-x_j)z_j}
\det\left(\frac{(x_j-a)^{i-1}}{(i-1)!}\right)_{1\le i,j\le n_2}\,d^{n_2}x
\notag
\end{align}
for any $a\in\mathbb{R}$. From this we see that
\begin{equation}\label{3.2}
\det \left(z_j^{i-1}\right)_{1\le i,j\le n_2}=(-1)^{\frac{n_2(n_2+1)}2}\prod\limits_{j=1}^{n_2}z_j^{n_2}\int\limits_{(-\infty,a]^{n_2}}\prod\limits_{j=1}^{n_2}e^{(a-x_j)z_j}
\det\left(\frac{(x_j-a)^{i-1}}{(i-1)!}\right)_{1\le i,j\le n_2}\,d^{n_2}x.
\end{equation}
Similarly, we get
\begin{equation}\label{3.3}
\det \left(w_j^{i-1}\right)_{1\le i,j\le n_2}=(-1)^{\frac{n_2(n_2+1)}2}\prod\limits_{j=1}^{n_2}w_j^{n_2}\int\limits_{(-\infty,b]^{n_2}}\prod\limits_{j=1}^{n_2}e^{(b-y_j)w_j}
\det\left(\frac{(y_j-b)^{i-1}}{(i-1)!}\right)_{1\le i,j\le n_2}\,d^{n_2}y,
\end{equation}
for any $b\in\mathbb{R}$. Choose $a=\xi_1$ and $b=\Delta\xi$. Using the identities (\ref{3.2}) and (\ref{3.3}) in (\ref{2.27}) we obtain
\begin{align}\label{3.4}
&Q(h)=\frac{(-1)^{\frac{n_2(n_2-1)}2}}{(2\pi i)^{2n_2}n_1!(\Delta n)!}\int_{\Gamma_{d_1}^{n_1}}d^{n_1}z\int_{\Gamma_{d_2}^{\Delta n}}d^{\Delta n}z
\int_{\Gamma_{d_3}^{n_1}}d^{n_1}w\int_{\Gamma_{d_4}^{\Delta n}}d^{\Delta n}w
\\
&\int\limits_{(-\infty,\xi_1]^{n_2}}d^{n_2}x\int\limits_{(-\infty,\Delta\xi]^{n_2}}d^{n_2}y
\det\left(\frac{x_j^{i-1}}{(i-1)!}\right)_{1\le i,j\le n_2}\det\left(\frac{y_j^{i-1}}{(i-1)!}\right)_{1\le i,j\le n_2}
\notag\\
&\times
\prod\limits_{j=1}^{n_1}z_j^{n_1}w_j^{\Delta n}e^{\frac 12\mu_1z_j^2-x_jz_j+\frac 12\Delta\mu w_j^2-y_jw_j}
\left(\frac 1{z_j-w_j}-h\right)
\notag\\
&
\times\prod\limits_{j=n_1+1}^{n_2}\frac{z_j^{n_1+1}w_j^{\Delta n-1}e^{\frac 12\mu_1z_j^2-x_jz_j+\frac 12\Delta\mu w_j^2-y_jw_j}}
{w_j-z_j}.
\notag
\end{align}
Here we also used the fact that
\begin{equation*}
\det\left(\frac{(x_j-a)^{i-1}}{(i-1)!}\right)_{1\le i,j\le n_2}=
\det\left(\frac{x_j^{i-1}}{(i-1)!}\right)_{1\le i,j\le n_2}
\end{equation*}
by the standard formula for the Vandermonde determinant, and similarly for the other determinant.
Let $H_k(x)=2^kx^k+\dots$, $k\ge 0$, be the standard Hermite polynomials so that, for any $a>0$,
\begin{align}
\det\left(\frac{x_j^{i-1}}{(i-1)!}\right)_{1\le i,j\le n_2}&=\frac 1{a^{\frac{n_2(n_2-1)}2}}\det\left(\frac{(ax_j)^{i-1}}{(i-1)!}\right)_{1\le i,j\le n_2}
=\frac 1{a^{\frac{n_2(n_2-1)}2}}\det\left(\frac{H_{i-1}(ax_j)}{2^{i-1}(i-1)!}\right)_{1\le i,j\le n_2}
\notag\\
&=\frac 1{(a\sqrt{2\mu_1})^{\frac{n_2(n_2-1)}2}}\det\left(\frac 1{2\pi i}\int_{\gamma_{\tau_1}}\frac{e^{a\sqrt{2\mu_1}x_j\zeta_j-\frac 12\mu_1\zeta_j^2}}
{\zeta_j^i}d\zeta_j\right)_{1\le i,j\le n_2},
\notag
\end{align}
where we have chosen $\tau_1$ so that (\ref{3.1'}) holds. Take $a=1/\sqrt{2\mu_1}$. We have shown that
\begin{equation}\label{3.5}
\det\left(\frac{x_j^{i-1}}{(i-1)!}\right)_{1\le i,j\le n_2}=\frac 1{(2\pi i)^{n_2}}\int_{\gamma_{\tau_1}^{n_2}}d^{n_2}\zeta
\prod\limits_{j=1}^{n_2}e^{x_j\zeta_j-\frac 12\mu_1\zeta_j^2}\det\left(\frac 1{\zeta_j^i}\right)_{1\le i,j\le n_2}.
\end{equation}
Similarly,
\begin{equation}\label{3.6}
\det\left(\frac{y_j^{i-1}}{(i-1)!}\right)_{1\le i,j\le n_2}=\frac{(-1)^{\frac{n_2(n_2-1)}2}}  {(2\pi i)^{n_2}}\int_{\gamma_{\tau_2}^{n_2}}d^{n_2}\omega
\prod\limits_{j=1}^{n_2}e^{y_j\omega_j-\frac 12\mu_1\omega_j^2}\det\left(\frac 1{\omega_j^{n_2+1-i}}\right)_{1\le i,j\le n_2},
\end{equation}
where $\tau_2$ satisfies (\ref{3.1'}). If we insert (\ref{3.5}) and (\ref{3.6}) into (\ref{3.4}) the $x_j$- and $y_j$-integrations
become
\begin{equation}
\int_{-\infty}^{\xi_1}e^{x_j(\zeta_j-z_j)}dx_j=\frac{e^{\xi_1(\zeta_j-z_j)}}{\zeta_j-z_j}\,\,,\,\,
\int_{-\infty}^{\Delta\xi}e^{y_j(\omega_j-w_j)}dy_j=\frac{e^{\Delta\xi(\omega_j-w_j)}}{\omega_j-w_j},
\notag
\end{equation}
where the integrals converge because of (\ref{3.1'}). The resulting formula is (\ref{3.1}).

\end{proof}

Write
\begin{equation}\label{3.7}
G_{n,\mu,\xi}(z)=z^ne^{\frac 12\mu z^2-\xi z}.
\end{equation}
Recall (\ref{3.1'}). For $1\le k,\ell\le n_2$, we define $A_h(\ell,k)$ by
\begin{align}\label{3.8}
&\delta_{\ell k}1(\ell\le n_1)+A_h(\ell,k)
\\
&=\frac 1{(2\pi i)^4}\int_{\Gamma_{d_1}}dz\int_{\Gamma_{d_3}}dw\int_{\gamma_{\tau_1}}d\zeta\int_{\gamma_{\tau_2}}d\omega
\frac{G_{n_1,\mu_1,\xi_1}(z)G_{\Delta n,\Delta\mu,\Delta\xi}(w)}{G_{k,\mu_1,\xi_1}(\zeta)G_{n_2+1-\ell,\Delta\mu,\Delta\xi}(\omega)}
\frac 1{(z-\zeta)(w-\omega)}\left(\frac 1{z-w}-h\right),
\notag
\end{align}
and $B(\ell,k)$ by
\begin{align}\label{3.9}
&\delta_{\ell k}1(\ell>n_1)+B(\ell,k)
\\
&=-\frac 1{(2\pi i)^4}\int_{\Gamma_{d_2}}dz\int_{\Gamma_{d_4}}dw\int_{\gamma_{\tau_1}}d\zeta\int_{\gamma_{\tau_2}}d\omega
\frac{G_{n_1+1,\mu_1,\xi_1}(z)G_{\Delta n-1,\Delta\mu,\Delta\xi}(w)}{G_{k,\mu_1,\xi_1}(\zeta)G_{n_2+1-\ell,\Delta\mu,\Delta\xi}(\omega)}
\frac 1{(z-w)(z-\zeta)(w-\omega)}.
\notag
\end{align}
Expand the determinants in (\ref{3.1}),
\begin{align}
\det\left(\frac 1{\zeta_j^i}\right)&=\sum\limits_{\sigma\in S_{n_2}}\sgn(\sigma)\prod\limits_{j=1}^{n_2}\frac 1{\zeta_j^{\sigma(j)}},
\notag\\
\det\left(\frac 1{\omega_j^{n_2+1-i}}\right)&=\sum\limits_{\tau\in S_{n_2}}\sgn(\tau)\prod\limits_{j=1}^{n_2}\frac 1{\omega_j^{n_2+1-\tau(j)}}.
\notag
\end{align}
From (\ref{3.8}) and (\ref{3.9}) it then follows that we can write
\begin{align}\label{3.10}
Q(h)=\frac 1{n_1!\Delta n!}\sum\limits_{\sigma,\tau\in S_{n_2}}\sgn(\sigma\tau)&\prod\limits_{j=1}^{n_1}\left(\delta_{\tau(j),\sigma(j)}1(\tau(j)\le n_1)+A_h(\tau(j),\sigma(j))\right)\\
\times&\prod\limits_{j=n_1+1}^{n_2}\left(\delta_{\tau(j),\sigma(j)}1(\tau(j)>n_1)+B(\tau(j),\sigma(j))\right).
\notag
\end{align}
This way of writing $Q(h)$ is useful because (\ref{3.10}) leads to a determinant expansion of $Q(h)$, and $A_h$ and $B$ can be rewritten in a way that is useful for taking limits, see lemma \ref{lem4.1}.

Write 
$$
[a,b]_<^n=\{(x_1,\dots,x_n)\in[a,b]^n\,;\,x_1<\dots<x_n\},
$$
which is empty if $n=0$,
and recall the notation (\ref{1.15}). By expanding (\ref{3.10}) we can prove
\begin{proposition}\label{prop3.2}
We have the formula
\begin{equation}\label{3.11}
Q(h)=\sum\limits_{r=0}^{\min(n_1,\Delta n)}\sum\limits_{s=0}^{n_1}\sum\limits_{t=0}^{\Delta n}\sum\limits_{\substack{\mathbf{c}\in[1,n_1]_<^r\\
\mathbf{c'}\in[1,n_1]_<^s}}\sum\limits_{\substack{\mathbf{d}\in[n_1+1,n_2]_<^r\\
\mathbf{d'}\in[n_1+1,n_2]_<^t}}
\det\left(\begin{matrix} B(\mathbf{c},\mathbf{c}) &B(\mathbf{c},\mathbf{c'}) &B(\mathbf{c},\mathbf{d}) &B(\mathbf{c},\mathbf{d'}) \\
A_h(\mathbf{c'},\mathbf{c}) &A_h(\mathbf{c'},\mathbf{c'}) &A_h(\mathbf{c'},\mathbf{d}) &A_h(\mathbf{c'},\mathbf{d'})\\
A_h(\mathbf{d},\mathbf{c}) &A_h(\mathbf{d},\mathbf{c'}) &A_h(\mathbf{d},\mathbf{d}) &A_h(\mathbf{d},\mathbf{d'})\\
B(\mathbf{d'},\mathbf{c}) &B(\mathbf{d'},\mathbf{c'}) &B(\mathbf{d'},\mathbf{d}) &B(\mathbf{d'},\mathbf{d'}) \end{matrix}\right).
\end{equation}
\end{proposition}

\begin{proof}
Set
\begin{equation}\label{3.12}
E_h(j;\ell,k)=\begin{cases}
   \delta_{\ell k}1(\ell\le n_1)+A_h(\ell,k) & \text{if } 1\le j\le n_1 \\
   \delta_{\ell k}1(\ell >n_1)+B(\ell,k)     & \text{if } n_1<j\le n_2
  \end{cases}.
\end{equation}
Then, by (\ref{3.10}),
\begin{equation}\label{3.13}
Q(h)=\frac 1{n_1!\Delta n!}\sum\limits_{\sigma,\tau\in S_{n_2}}\sgn(\sigma\tau)\prod\limits_{j=1}^{n_2}E_h(j;\tau(j),\sigma(j)).
\end{equation}
By reordering the product we get
\begin{align*}
Q(h)&=\frac 1{n_1!\Delta n!}\sum\limits_{\sigma,\tau\in S_{n_2}}\sgn(\sigma\tau)\prod\limits_{j=1}^{n_2}E_h(\tau^{-1}(j);j,\sigma(\tau^{-1}(j)))\\
&=\frac 1{n_1!\Delta n!}\sum\limits_{\sigma,\tau\in S_{n_2}}\sgn(\sigma\tau^{-1})\prod\limits_{j=1}^{n_2}E_h(\tau^{-1}(j);j,\sigma(\tau^{-1}(j))),
\end{align*}
since $\sgn(\sigma\tau^{-1})=\sgn(\sigma\tau)$.
If we replace $\sigma\tau^{-1}$ by $\sigma$ and then $\tau^{-1}$ by $\tau$, we see that
\begin{equation}\label{3.14}
Q(h)=\frac 1{n_1!\Delta n!}\sum\limits_{\sigma,\tau\in S_{n_2}}\sgn(\sigma)\prod\limits_{j=1}^{n_2}E_h(\tau(j);j,\sigma(j)).
\end{equation}
Let $J_-\subseteq[1,n_2]$, $|J_-|=n_1$ and $J_+=[1,n_2]\setminus J_-$. Then, by (\ref{3.12}) and (\ref{3.14}),
\begin{align}\label{3.15}
Q(h)&=\frac 1{n_1!\Delta n!}\sum\limits_{J_-}\sum\limits_{\sigma\in S_{n_2}}\sgn(\sigma)\sum\limits_{\tau\in S_{n_2};\tau(J_-)=[1,n_1]}
\prod\limits_{j\in J_-}E_h(\tau(j);j,\sigma(j))\prod\limits_{j\in J_+}E_h(\tau(j);j,\sigma(j))
\\
&=\sum\limits_{J_-}\sum\limits_{\sigma\in S_{n_2}}\sgn(\sigma) \prod\limits_{j\in J_-}(\delta_{j,\sigma(j)}1(j\le n_1)+A_h(j,\sigma(j)))
\prod\limits_{j\in J_+}(\delta_{j,\sigma(j)}1(j>n_1)+B(j,\sigma(j))),
\notag
\end{align}
since
$$
\sum\limits_{\tau\in S_{n_2}:\tau(J_-)=[1,n_1]}1=n_1!\Delta n!.
$$
We can rewrite (\ref{3.15}) as
\begin{align}\label{3.16}
Q(h)
=\sum\limits_{J_-}\sum\limits_{\sigma\in S_{n_2}}\sgn(\sigma) &\prod\limits_{j\in J_-\cap[1,n_1]}(\delta_{j,\sigma(j)}+A_h(j,\sigma(j)))
\prod\limits_{j\in J_-\cap[n_1+1,n_2]}A_h(j,\sigma(j))
\\
\times&\prod\limits_{j\in J_+\cap[1,n_1]}B(j,\sigma(j))
\prod\limits_{j\in J_+\cap[n_1+1,n_2]}(\delta_{j,\sigma(j)}+B(j,\sigma(j))).
\notag
\end{align}

We want to expand the products involving the Kronecker deltas. Let
$$
\gamma=J_+\cap[1,n_1]\,\,,\,\,\delta=J_-\cap[n_1+1,n_2].
$$
Set $r=|J_+\cap[1,n_1]|$. Then, $|J_-\cap[1,n_1]|=n_1-r$ and we see that $0\le r\le n_1$. Since $|J_-|=n_1$, we get
$$
|J_-\cap[n_1+1,n_2]|=n_1-|J_-\cap[1,n_1]|=r.
$$
Thus, $|\gamma|=|\delta|=r$. Given $\gamma,\delta$ we see that $J_-=\delta\cup([1,n_1]\setminus \gamma)$, so $J_-$ is uniquely
determined by $\gamma,\delta$. Hence, (\ref{3.16}) can be written as
\begin{align}\label{3.17}
Q(h)
=\sum\limits_{r=0}^{\min(n_1,\Delta n)}\sum\limits_{\gamma,\delta}\sum\limits_{\sigma\in S_{n_2}}\sgn(\sigma) &\prod\limits_{j\in [1,n_1]\setminus\gamma}
\left(\delta_{j,\sigma(j)}+A_h(j,\sigma(j))\right)
\prod\limits_{j\in\delta}A_h(j,\sigma(j))
\\
\times&\prod\limits_{j\in\gamma}B(j,\sigma(j))
\prod\limits_{j\in [n_1+1,n_2]\setminus\delta}\left(\delta_{j,\sigma(j)}+B(j,\sigma(j))\right),
\notag
\end{align}
where we sum over all $\gamma,\delta$ such that $\gamma\subseteq [1,n_1]$, $\delta\subseteq [n_1+1,n_2]$, $|\gamma|=|\delta|=r$.

Now,
\begin{equation}\label{3.18}
\prod\limits_{j\in [1,n_1]\setminus\gamma}(\delta_{j,\sigma(j)}+A_h(j,\sigma(j)))
=\sum\limits_{\gamma'\subseteq [1,n_1]\setminus\gamma}\prod\limits_{j\in[1,n_1]\setminus(\gamma\cup\gamma')}
\delta_{j,\sigma(j)}\prod\limits_{j\in\gamma'}A_h(j,\sigma(j))
\end{equation}
and
\begin{equation}\label{3.19}
\prod\limits_{j\in [n_1+1,n_2]\setminus\delta}(\delta_{j,\sigma(j)}+B(j,\sigma(j)))
=\sum\limits_{\delta'\subseteq [n_1+1,n_2]\setminus\delta}\prod\limits_{j\in[n_1+1,n_2]\setminus(\delta\cup\delta')}
\delta_{j,\sigma(j)}\prod\limits_{j\in\delta'}B(j,\sigma(j)).
\end{equation}
Inserting (\ref{3.18}) and (\ref{3.19}) into (\ref{3.17}) yields
\begin{align}\label{3.20}
&Q(h)
=\sum\limits_{r=0}^{\min(n_1,\Delta n)}\sum\limits_{s=0}^{n_1-r}\sum\limits_{t=0}^{\Delta n-r}
\sum\limits_{\gamma,\gamma',\delta,\delta'}\sum\limits_{\sigma\in S_{n_2}}\sgn(\sigma) 
\\
&\times\prod\limits_{j\in[1,n_1]\setminus(\gamma\cup\gamma')}\delta_{j,\sigma(j)}\prod\limits_{j\in[n_1+1,n_2]\setminus(\delta\cup\delta')}\delta_{j,\sigma(j)}
\prod\limits_{j\in\gamma}B(j,\sigma(j))\prod\limits_{j\in\gamma'}A_h(j,\sigma(j))\prod\limits_{j\in\delta}A_h(j,\sigma(j))
\prod\limits_{j\in\delta'}B(j,\sigma(j))
\notag
\end{align}
where we sum over all $\gamma,\gamma',\delta,\delta'$ such that
\begin{align}\label{3.21}
&\gamma,\gamma'\subseteq[1,n_1]\,,\,\delta,\delta'\subseteq[n_1+1,n_2]\,,\,\gamma\cap\gamma'=\emptyset\,,\,\delta\cap\delta'=\emptyset,
\\
&|\gamma|=|\delta|=r\,,\,|\gamma'|=s\,,\,|\delta'|=t.
\notag
\end{align}
Let $\Lambda=\gamma\cup\gamma'\cup\delta\cup\delta'$ and $L=|\Delta|=2r+t+s$. Terms in (\ref{3.20}) are $\neq 0$ only if $\sigma(j)=j$ for
$j\in[1,n_2]\setminus\Lambda$. The permutation $\sigma$ is then reduced to a bijection $\tilde{\sigma}:\Lambda\to\Lambda$. 
Let $\Lambda=\{\lambda_1,\dots,\lambda_L\}$, where $\lambda_1<\dots<\lambda_L$. Then $\tilde{\sigma}$ is a permutation of $\Lambda$ and we have
$\sgn(\tilde{\sigma})=\sgn(\sigma)$. Thus,
\begin{align}\label{3.22}
&Q(h)
=\sum\limits_{r=0}^{\min(n_1,\Delta n)}\sum\limits_{s=0}^{n_1-r}\sum\limits_{t=0}^{\Delta n-r}
\sum\limits_{\gamma,\gamma',\delta,\delta'}\sum\limits_{\tilde{\sigma}:\Lambda\to\Lambda}\sgn(\tilde{\sigma}) 
\\
&\times
\prod\limits_{j\in\gamma}B(j,\tilde{\sigma}(j))\prod\limits_{j\in\gamma'} A_h(j,\tilde{\sigma}(j))\prod\limits_{j\in\delta} A_h(j,\tilde{\sigma}(j))
\prod\limits_{j\in\delta'}B(j,\tilde{\sigma}(j)).
\notag
\end{align}
Define,
\begin{equation}
T_h(j;\ell,k)=\begin{cases}
  B(\ell,k)  & \text{if } j\in\gamma \\
  A_h(\ell,k) & \text{if } j\in\gamma' \\
  A_h(\ell,k)  & \text{if } j\in\delta \\
  B(\ell,k)  & \text{if } j\in\delta'.
  \end{cases}
\end{equation}
Then, by (\ref{3.22}),
\begin{equation}
Q(h)
=\sum\limits_{r=0}^{\min(n_1,\Delta n)}\sum\limits_{s=0}^{n_1-r}\sum\limits_{t=0}^{\Delta n-r}
\sum\limits_{\gamma,\gamma',\delta,\delta'}\sum\limits_{\tilde{\sigma}:\Lambda\to\Lambda}\sgn(\tilde{\sigma}) \prod\limits_{j\in\Lambda} T_h(j;j,\tilde{\sigma}(j)).
\end{equation}
Define $\tau\in S_L$ by $\tilde{\sigma}(\lambda_j)=\lambda_{\tau(j)}$, $1\le j\le L$. Then
$\sgn(\tilde{\sigma})=\sgn(\tau)$ and we find
\begin{align}\label{3.23}
Q(h)
&=\sum\limits_{r=0}^{\min(n_1,\Delta n)}\sum\limits_{s=0}^{n_1-r}\sum\limits_{t=0}^{\Delta n-r}
\sum\limits_{\gamma,\gamma',\delta,\delta'}\sum\limits_{\tau\in S_L}\sgn(\tau)\prod\limits_{i=1}^L T_h(\lambda_i;\lambda_i,\lambda_{\tau(i)})
\\
&=\sum\limits_{r=0}^{\min(n_1,\Delta n)}\sum\limits_{s=0}^{n_1-r}\sum\limits_{t=0}^{\Delta n-r}
\sum\limits_{\gamma,\gamma',\delta,\delta'}\det\left(T_h(\lambda_i;\lambda_i,\lambda_{j})\right).
\notag
\end{align}
Let
\begin{align}
\gamma&=\{c_1,\dots,c_r\}\,,\, \mathbf{c}=(c_1,\dots,c_r)\in[1,n_1]_<^r,
\notag\\
\gamma'&=\{c'_1,\dots,c'_s\}\,,\, \mathbf{c'}=(c'_1,\dots,c'_r)\in[1,n_1]_<^s,
\notag\\
\delta&=\{d_1,\dots,d_r\}\,,\, \mathbf{d}=(d_1,\dots,d_r)\in[n_1+1,n_2]_<^r,
\notag\\
\delta&=\{d'_1,\dots,d'_t\}\,,\, \mathbf{d'}=(d'_1,\dots,d'_t)\in[n_1+1,n_2]_<^t.
\notag
\end{align}
Notice that the determinant in (\ref{3.23}) is unchanged under permutations of the $\lambda_i$'s. Thus we can reorder the 
$\lambda_i$'s in (\ref{3.23}) so that we get the order
$c_1,\dots,c_r,c'_1,\dots,c'_s,d_1,\dots,d_r,d'_1,\dots,d'_t$. Also, notice that if $c_i=c_j'$ or $d_i=d_j'$
for some $i,j$, then the determinant is $=0$. Hence, we can remove the restrictions $\gamma\cap\gamma'=\emptyset$ and $\delta\cap\delta'=\emptyset$
in (\ref{3.21}). Note that if e.g. $s>n_1-r$, then we must have $c_i=c_j'$ for some $i,j$. Thus, the right side in (\ref{3.23}) equals the right side in (\ref{3.11}).
\end{proof}

We now want to give expressions for $A_h$ and $B$ that will be useful in the asymptotic analysis. First, we need some definitions.
Recall the notation (\ref{3.7}). Let $0<\tau_1,\tau_2<D_1<D_2$ and define
\begin{equation}\label{3.24}
a_{0,1}(\ell,k)=\frac 1{(2\pi i)^4}\int_{\Gamma_{D_1}}dz\int_{\Gamma_{D_2}}dw\int_{\gamma_{\tau_1}}d\zeta\int_{\gamma_{\tau_2}}d\omega
\frac{G_{n_1,\mu_1,\xi_1}(z)G_{\Delta n,\Delta\mu,\Delta\xi}(w)}{G_{k,\mu_1,\xi_1}(\zeta)G_{n_2+1-\ell,\Delta\mu,\Delta\xi}(\omega)(z-w)(z-\zeta)(w-\omega)},
\end{equation}
\begin{equation}\label{3.25}
b_1(\ell,k)=\frac 1{(2\pi i)^4}\int_{\Gamma_{D_2}}dz\int_{\Gamma_{D_1}}dw\int_{\gamma_{\tau_1}}d\zeta\int_{\gamma_{\tau_2}}d\omega
\frac{G_{n_1+1,\mu_1,\xi_1}(z)G_{\Delta n-1,\Delta\mu,\Delta\xi}(w)}{G_{k,\mu_1,\xi_1}(\zeta)G_{n_2+1-\ell,\Delta\mu,\Delta\xi}(\omega)(z-w)(z-\zeta)(w-\omega)}.
\end{equation}
Let $0<\tau<D$ and define
\begin{equation}\label{3.26}
c_2(\ell,k)=\frac 1{(2\pi i)^2}\int_{\Gamma_{D}}dw\int_{\gamma_{\tau}}d\omega\frac{G_{n_2-k,\Delta\mu,\Delta\xi}(w)}
{G_{n_2+1-\ell,\Delta\mu,\Delta\xi}(\omega)(w-\omega)},
\end{equation}
\begin{equation}\label{3.27}
c_3(\ell,k)=\frac 1{(2\pi i)^2}\int_{\Gamma_{D}}dz\int_{\gamma_{\tau}}d\zeta\frac{G_{\ell-1,\mu_1,\xi_1}(z)}
{G_{k,\mu_1,\xi_1}(\zeta)(z-\zeta)}.
\end{equation}
We now set,
\begin{subequations}\label{3.28}
\begin{align}
        a_{0,2}(\ell,k)&=-1(k>n_1)c_2(\ell,k)\\
        a_{0,3}(\ell,k)&=1(\ell\le n_1)c_3(\ell,k)\\
        b_2(\ell,k)&=-1(k>n_1+1)c_2(\ell,k)\\
        b_3(\ell,k)&=1(\ell\le n_1+1)c_3(\ell,k)\\
        a_2^\ast(\ell)&=c_2(\ell,n_1)\\
        a_3^\ast(k)&=c_3(n_1+1,k),
\end{align}
\end{subequations}
and finally, we define
\begin{subequations}\label{3.29}
\begin{align}
        a_0(\ell,k)&=a_{0,1}(\ell,k)-a_{0,2}(\ell,k)-a_{0,3}(\ell,k)\\
        b(\ell,k)&=-b_1(\ell,k)+ b_2(\ell,k)+b_3(\ell,k)\\
        A_0^\ast(\ell,k)&=-(\delta_{k,n_1+1}-a_3^\ast(k))(\delta_{\ell,n_1}-a_2^\ast(\ell))\label{3.29c}.
\end{align}
\end{subequations}
With this notation we can formulate our next lemma.

\begin{lemma}\label{lem3.3} If $A_h(\ell,k)$ and $B(\ell,k)$, $1\le\ell,k\le n_2$, are defined by 
(\ref{3.8}) and (\ref{3.9}) then
\begin{equation}\label{3.30}
A_0(\ell,k)=a_0(\ell,k),
\end{equation}
\begin{equation}\label{3.31}
B(\ell,k)=-\delta_{k,n_1+1}\delta_{\ell,n_1+1}+b(\ell,k)
\end{equation}
and 
\begin{equation}\label{3.32}
\left.\frac{\partial}{\partial h}\right|_{h=0}A_h(\ell,k)=A_0^\ast(\ell,k).
\end{equation}
\end{lemma}

\begin{proof}
Recall the condition (\ref{3.1'}),
\begin{equation}\label{3.33}
d_1<d_3<-\max(\tau_1,\tau_2)<0\,\,,\,\,d_4<d_2<-\max(\tau_1,\tau_2)<0.
\end{equation}
Choose $D_1,D_2,r_1,r_2,\tau_1,\tau_2$ so that
\begin{equation}\label{3.34}
0<\tau_2<\tau_1<r_1<r_2<D_1<D_2.
\end{equation}
In the integral in the right side of (\ref{3.8}) we can deform $\Gamma_{d_3}$ to $\Gamma_{D_2}$ and $-\gamma_{r_1}$, and then 
$\Gamma_{d_1}$ to $\Gamma_{D_1}$ and $-\gamma_{r_2}$ without passing any poles. This gives
\begin{align}\label{3.35}
&\delta_{\ell k}1(\ell\le n_1)+A_h(\ell,k)
=\frac 1{(2\pi i)^4}\left(\int_{\Gamma_{D_1}}dz\int_{\Gamma_{D_2}}dw-\int_{\Gamma_{D_1}}dz\int_{\gamma_{r_1}}dw
-\int_{\gamma_{r_2}}dz\int_{\Gamma_{D_2}}dw\right)
\\&\times
\int_{\gamma_{\tau_1}}d\zeta\int_{\gamma_{\tau_2}}d\omega
\frac{G_{n_1,\mu_1,\xi_1}(z)G_{\Delta n,\Delta\mu,\Delta\xi}(w)}{G_{k,\mu_1,\xi_1}(\zeta)G_{n_2+1-\ell,\Delta\mu,\Delta\xi}(\omega)}
\frac 1{(z-\zeta)(w-\omega)}\left(\frac 1{z-w}-h\right)
\notag\\
&+\frac 1{(2\pi i)^4}\int_{\gamma_{r_2}}dz\int_{\gamma_{r_1}}dw\int_{\gamma_{\tau_1}}d\zeta\int_{\gamma_{\tau_2}}d\omega
\frac{G_{n_1,\mu_1,\xi_1}(z)G_{\Delta n,\Delta\mu,\Delta\xi}(w)}{G_{k,\mu_1,\xi_1}(\zeta)G_{n_2+1-\ell,\Delta\mu,\Delta\xi}(\omega)}
\frac 1{(z-\zeta)(w-\omega)}\left(\frac 1{z-w}-h\right).
\notag
\end{align}
Consider the last integral in (\ref{3.35}). The $w$-integral has its only pole in $w=\omega$ and hence it equals
\begin{equation}\label{3.36}
\frac 1{(2\pi i)^3}\int_{\gamma_{r_2}}dz\int_{\gamma_{\tau_1}}d\zeta\int_{\gamma_{\tau_2}}d\omega
\frac{G_{n_1,\mu_1,\xi_1}(z)}{G_{k,\mu_1,\xi_1}(\zeta)\omega^{n_1+1-\ell}(z-\zeta)}\left(\frac 1{z-\omega}-h\right).
\end{equation}
In this integral the $z$-integral has poles at $z=\zeta$ and at $z=\omega$, which gives
\begin{equation}\label{3.37}
\frac 1{(2\pi i)^2}\int_{\gamma_{\tau_1}}d\zeta\int_{\gamma_{\tau_2}}d\omega
\frac{\zeta^{n_1-k}}{\omega^{n_1+1-\ell}}\left(\frac 1{\zeta-\omega}-h\right)+
\frac 1{(2\pi i)^2}\int_{\gamma_{\tau_1}}d\zeta\int_{\gamma_{\tau_2}}d\omega
\frac{G_{\ell-1,\mu_1,\xi_1}(\omega)}{G_{k,\mu_1,\xi_1}(\zeta)(\omega-\zeta)}.
\end{equation}
Since $\tau_2<\tau_1$, the $\zeta$-integral in the second integral in (\ref{3.37}) is $=0$. The first integral in (\ref{3.37}) equals
$\delta_{\ell, k}1(\ell\le n_1)-h\delta_{k,n_1+1}\delta_{\ell,n_1}$. Combined with (\ref{3.35}) this gives,
\begin{equation}\label{3.38}
A_h(\ell,k)=-h\delta_{k,n_1+1}\delta_{\ell,n_1}+a_{h,1}(\ell,k)-a_{h,2}(\ell,k)-a_{h,3}(\ell,k),
\end{equation}
where
\begin{equation}\label{3.39}
a_{h,1}(\ell,k)=\frac 1{(2\pi i)^4}\int_{\Gamma_{D_1}}dz\int_{\Gamma_{D_2}}dw\int_{\gamma_{\tau_1}}d\zeta\int_{\gamma_{\tau_2}}d\omega
\frac{G_{n_1,\mu_1,\xi_1}(z)G_{\Delta n,\Delta\mu,\Delta\xi}(w)\left(\frac 1{z-w}-h\right)}{G_{k,\mu_1,\xi_1}(\zeta)G_{n_2+1-\ell,\Delta\mu,\Delta\xi}(\omega)
(z-\zeta)(w-\omega)}
\end{equation}
\begin{equation}\label{3.40}
a_{h,2}(\ell,k)=\frac 1{(2\pi i)^4}\int_{\gamma_{r_2}}dz\int_{\Gamma_{D_2}}dw\int_{\gamma_{\tau_1}}d\zeta\int_{\gamma_{\tau_2}}d\omega
\frac{G_{n_1,\mu_1,\xi_1}(z)G_{\Delta n,\Delta\mu,\Delta\xi}(w)\left(\frac 1{z-w}-h\right)}{G_{k,\mu_1,\xi_1}(\zeta)G_{n_2+1-\ell,\Delta\mu,\Delta\xi}(\omega)
(z-\zeta)(w-\omega)}
\end{equation}
and
\begin{equation}\label{3.41}
a_{h,3}(\ell,k)=\frac 1{(2\pi i)^4}\int_{\Gamma_{D_1}}dz\int_{\gamma_{r_1}}dw\int_{\gamma_{\tau_1}}d\zeta\int_{\gamma_{\tau_2}}d\omega
\frac{G_{n_1,\mu_1,\xi_1}(z)G_{\Delta n,\Delta\mu,\Delta\xi}(w)\left(\frac 1{z-w}-h\right)}{G_{k,\mu_1,\xi_1}(\zeta)G_{n_2+1-\ell,\Delta\mu,\Delta\xi}(\omega)
(z-\zeta)(w-\omega)}.
\end{equation}
We see that $a_{0,1}(\ell,k)$ in (\ref{3.39}) agrees with  (\ref{3.24}). Also
\begin{equation}\label{3.43}
a_{0,2}(\ell,k)=\frac 1{(2\pi i)^4}\int_{\gamma_{r_2}}dz\int_{\Gamma_{D_2}}dw\int_{\gamma_{\tau_1}}d\zeta\int_{\gamma_{\tau_2}}d\omega
\frac{G_{n_1,\mu_1,\xi_1}(z)G_{\Delta n,\Delta\mu,\Delta\xi}(w)}{G_{k,\mu_1,\xi_1}(\zeta)G_{n_2+1-\ell,\Delta\mu,\Delta\xi}(\omega)
(z-w)(z-\zeta)(w-\omega)}.
\end{equation}
The $z$-integral in (\ref{3.43}) has its only pole in $z=\zeta$ and hence
\begin{equation}
a_{0,2}(\ell,k)=\frac 1{(2\pi i)^3}\int_{\Gamma_{D_2}}dw\int_{\gamma_{\tau_1}}d\zeta\int_{\gamma_{\tau_2}}d\omega
\frac{G_{\Delta n,\Delta\mu,\Delta\xi}(w)}{\zeta^{k-n_1}G_{n_2+1-\ell,\Delta\mu,\Delta\xi}(\omega)
(\zeta-w)(w-\omega)}.
\notag
\end{equation}
The $\zeta$-integral is $=0$ unless $k>n_1$, and if $k>n_1$ the $\zeta$-integral has $\zeta=w$ as its only pole outside $\gamma_{\tau_1}$.
Thus,
\begin{equation}
a_{0,2}(\ell,k)=-\frac {1(k>n_1)}{(2\pi i)^2}\int_{\Gamma_{D_2}}dw\int_{\gamma_{\tau_2}}d\omega
\frac{G_{n_2-k,\Delta\mu,\Delta\xi}(w)}{G_{n_2+1-\ell,\Delta\mu,\Delta\xi}(\omega)(w-\omega)}
=-1(k>n_1)c_2(\ell,k).
\notag
\end{equation}
Similarly, we can show that
\begin{equation}
a_{0,3}(\ell,k)=\frac {1(\ell\le n_1)}{(2\pi i)^2}\int_{\Gamma_{D_1}}dz\int_{\gamma_{\tau_1}}d\zeta
\frac{G_{\ell-1,\mu_1,\xi_1}(z)}{G_{k,\mu_1,\xi_1}(\zeta)(z-\zeta)}=1(\ell\le n_1)c_3(\ell,k).
\notag
\end{equation}
This proves (\ref{3.30}).
Now,
\begin{align}\label{3.44}
&\left.\frac{\partial}{\partial h}\right|_{h=0}a_{h,1}(\ell,k)
\\
&=-\left(\frac 1{(2\pi i)^2}\int_{\Gamma_{D_1}}dz\int_{\gamma_{\tau_1}}d\zeta\frac{G_{n_1,\mu_1,\xi_1}(z)}{G_{k,\mu_1,\xi_1}(\zeta)(z-\zeta)}\right)
\left(\frac 1{(2\pi i)^2}\int_{\Gamma_{D_2}}dw\int_{\gamma_{\tau_2}}d\omega\frac{G_{\Delta n,\Delta\mu,\Delta\xi}(w)}{G_{n_2+1-\ell,\Delta\mu,\Delta\xi}(\omega)
(w-\omega)}\right)
\notag\\
&=-c_3(n_1+1,k)c_2(\ell,n_1)=-a_3^\ast(k)a_2^\ast(\ell).
\notag
\end{align}
Similarly
\begin{align}\label{3.45}
&\left.\frac{\partial}{\partial h}\right|_{h=0}a_{h,2}(\ell,k)
\\
&=-\left(\frac 1{(2\pi i)^2}\int_{\gamma_{r_2}}dz\int_{\gamma_{\tau_1}}d\zeta\frac{G_{n_1,\mu_1,\xi_1}(z)}{G_{k,\mu_1,\xi_1}(\zeta)(z-\zeta)}\right)
\left(\frac 1{(2\pi i)^2}\int_{\Gamma_{D_2}}dw\int_{\gamma_{\tau_2}}d\omega\frac{G_{\Delta n,\Delta\mu,\Delta\xi}(w)}{G_{n_2+1-\ell,\Delta\mu,\Delta\xi}(\omega)
(w-\omega)}\right)
\notag\\
&=-\left(\frac 1{2\pi i}\int_{\gamma_{\tau_1}}\zeta^{n_1-k}d\zeta\right)c_2(\ell,n_1)==-\delta_{k,n_1+1}a_2^\ast(\ell),
\notag
\end{align}
and
\begin{equation}\label{3.46}
\left.\frac{\partial}{\partial h}\right|_{h=0}a_{h,3}(\ell,k)=-\delta_{\ell,n_1}a_3^\ast(k).
\end{equation}
If we use (\ref{3.44}) - (\ref{3.46}) in (\ref{3.38}) we see that we have proved (\ref{3.32}).

Consider next $B(\ell,k)$. In the integral in the right side
of (\ref{3.9}) we deform  $\Gamma_{d_2}$ to $\Gamma_{D_2}$ and $-\gamma_{r_1}$, and then
$\Gamma_{d_4}$ to $\Gamma_{D_1}$ and $-\gamma_{r_2}$, which can be done without passing any poles. We obtain
\begin{align}\label{3.47}
&\delta_{\ell k}1(\ell> n_1)+B(\ell,k)
=\frac 1{(2\pi i)^4}\left(-\int_{\Gamma_{D_2}}dz\int_{\Gamma_{D_1}}dw+\int_{\Gamma_{D_2}}dz\int_{\gamma_{r_2}}dw
+\int_{\gamma_{r_1}}dz\int_{\Gamma_{D_1}}dw\right)
\\&\times
\int_{\gamma_{\tau_1}}d\zeta\int_{\gamma_{\tau_2}}d\omega
\frac{G_{n_1+1,\mu_1,\xi_1}(z)G_{\Delta n-1,\Delta\mu,\Delta\xi}(w)}{G_{k,\mu_1,\xi_1}(\zeta)G_{n_2+1-\ell,\Delta\mu,\Delta\xi}(\omega)}
\frac 1{(z-w)(z-\zeta)(w-\omega)}
\notag\\
&-\frac 1{(2\pi i)^4}\int_{\gamma_{r_1}}dz\int_{\gamma_{r_2}}dw\int_{\gamma_{\tau_1}}d\zeta\int_{\gamma_{\tau_2}}d\omega
\frac{G_{n_1+1,\mu_1,\xi_1}(z)G_{\Delta n-1,\Delta\mu,\Delta\xi}(w)}{G_{k,\mu_1,\xi_1}(\zeta)G_{n_2+1-\ell,\Delta\mu,\Delta\xi}(\omega)}
\frac 1{(z-w)(z-\zeta)(w-\omega)}.
\notag
\end{align}
Consider the last integral in (\ref{3.47}). The $z$-integral has its only pole at $z=\zeta$ and hence it equals
\begin{equation}\label{3.48}
\frac 1{(2\pi i)^3}\int_{\gamma_{r_2}}dw\int_{\gamma_{\tau_1}}d\zeta\int_{\gamma_{\tau_2}}d\omega
\frac{G_{\Delta n-1,\Delta\mu,\Delta\xi}(w)}{\zeta^{k-(n_1+1)}G_{n_2+1-\ell,\Delta\mu,\Delta\xi}(\omega)(w-\omega)(w-\zeta)}.
\end{equation}
The $w$-integral has poles at $w=\omega$ and $w=\zeta$ and consequently (\ref{3.48}) equals
\begin{equation}\label{3.49}
\frac 1{(2\pi i)^2}\int_{\gamma_{\tau_1}}d\zeta \int_{\gamma_{\tau_2}}d\omega
\frac{\zeta^{n_1+1-k}\omega^{\ell-(n_1+2)}}{\omega-\zeta}+
\frac 1{(2\pi i)^2}\int_{\gamma_{\tau_1}}d\zeta\int_{\gamma_{\tau_2}}d\omega
\frac{G_{n_2-k,\Delta\mu,\Delta\xi}(\zeta)}{G_{n_2+1-\ell,\Delta\mu,\Delta\xi}(\omega)(\zeta-\omega)}.
\end{equation}
The first integral in (\ref{3.49}) equals $-\delta_{\ell,k}1(\ell\le n_1+1)$ and in the second one the $\zeta$-integral
has its only pole at $\zeta=\omega$ and hence equals $\delta_{\ell,k}$. Thus, the integral in (\ref{3.49}) equals 
$\delta_{\ell,k}1(\ell>n_1+1)$ and we see from (\ref{3.47}) that
\begin{align}
B(\ell,k)
&=-\delta_{k,n_1+1}\delta_{\ell,n_1+1}+
\frac 1{(2\pi i)^4}\left(-\int_{\Gamma_{D_2}}dz\int_{\Gamma_{D_1}}dw+\int_{\Gamma_{D_2}}dz\int_{\gamma_{r_2}}dw
+\int_{\gamma_{r_1}}dz\int_{\Gamma_{D_1}}dw\right)
\notag\\
&\times\int_{\gamma_{\tau_1}}d\zeta\int_{\gamma_{\tau_2}}d\omega
\frac{G_{n_1+1,\mu_1,\xi_1}(z)G_{\Delta n-1,\Delta\mu,\Delta\xi}(w)}{G_{k,\mu_1,\xi_1}(\zeta)G_{n_2+1-\ell,\Delta\mu,\Delta\xi}(\omega)}
\frac 1{(z-w)(z-\zeta)(w-\omega)}.
\notag
\end{align}
This leads to the formula (\ref{3.31}) by using an argument that is analogous to how we proved (\ref{3.30}).
\end{proof}

Before we can carry out the asymptotic analysis of the expression for $Q(h)$ in (\ref{3.11}) we have to rewrite it further. Define
\begin{equation}\label{3.492}
\tilde{a}_0(\ell,n_1)=a_0(\ell,n_1)+a_2^\ast(\ell)=
a_{0,1}(\ell,n_1)+c_2(\ell,n_1)-1(\ell\le n_1)c_3(\ell,n_1).
\end{equation}
Set
\begin{equation}\label{3.493}
V(\mathbf{c},\mathbf{c'},\mathbf{d},\mathbf{d'})=
\left( \begin{matrix}
b(\mathbf{c},\mathbf{c}) &b(\mathbf{c},\mathbf{c'}) &b(\mathbf{c},n_1) &b(\mathbf{c},\mathbf{d}) &b(\mathbf{c},\mathbf{d'})
\\
a_0(\mathbf{c'},\mathbf{c}) &a_0(\mathbf{c'},\mathbf{c'}) &\tilde{a}_0(\mathbf{c'},n_1) &a_0(\mathbf{c'},\mathbf{d}) &a_0(\mathbf{c'},\mathbf{d'})
\\
b(n_1+1,\mathbf{c}) &b(n_1+1,\mathbf{c'}) &b(n_1+1,n_1) &b(n_1+1,\mathbf{d}) &b(n_1+1,\mathbf{d'})
\\
a_0(\mathbf{d},\mathbf{c}) &a_0(\mathbf{d},\mathbf{c'}) &\tilde{a}_0(\mathbf{d},n_1) &a_0(\mathbf{d},\mathbf{d}) &a_0(\mathbf{d},\mathbf{d'})
\\
b(\mathbf{d'},\mathbf{c}) &b(\mathbf{d'},\mathbf{c'}) &b(\mathbf{d'},n_1) &b(\mathbf{d'},\mathbf{d}) &b(\mathbf{d'},\mathbf{d'})
\end{matrix}\right),
\end{equation}
\begin{equation}\label{3.494}
U(\mathbf{c},\mathbf{c'},\mathbf{d},\mathbf{d'})=
\left( \begin{matrix}
b(\mathbf{c},\mathbf{c}) &b(\mathbf{c},\mathbf{c'}) &b(\mathbf{c},n_1) &b(\mathbf{c},\mathbf{d}) &b(\mathbf{c},\mathbf{d'})
\\
a_0(\mathbf{c'},\mathbf{c}) &a_0(\mathbf{c'},\mathbf{c'}) &\tilde{a}_0(\mathbf{c'},n_1) &a_0(\mathbf{c'},\mathbf{d}) &a_0(\mathbf{c'},\mathbf{d'})
\\
a_0(n_1+1,\mathbf{c}) &a_0(n_1+1,\mathbf{c'}) &\tilde{a}_0(n_1+1,n_1) &a_0(n_1+1,\mathbf{d}) &a_0(n_1+1,\mathbf{d'})
\\
a_0(\mathbf{d},\mathbf{c}) &a_0(\mathbf{d},\mathbf{c'}) &\tilde{a}_0(\mathbf{d},n_1) &a_0(\mathbf{d},\mathbf{d}) &a_0(\mathbf{d},\mathbf{d'})
\\
b(\mathbf{d'},\mathbf{c}) &b(\mathbf{d'},\mathbf{c'}) &b(\mathbf{d'},n_1) &b(\mathbf{d'},\mathbf{d}) &b(\mathbf{d'},\mathbf{d'})
\end{matrix}\right),
\end{equation}
and
\begin{equation}\label{3.50}
M_h(\mathbf{c},\mathbf{c'},\mathbf{d},\mathbf{d'})=
\left( \begin{matrix}
b(\mathbf{c},\mathbf{c}) &b(\mathbf{c},\mathbf{c'}) &b(\mathbf{c},\mathbf{d}) &b(\mathbf{c},\mathbf{d'})
\\
A_h(\mathbf{c'},\mathbf{c}) &A_h(\mathbf{c'},\mathbf{c'}) &A_h(\mathbf{c'},\mathbf{d}) &A_h(\mathbf{c'},\mathbf{d'})
\\
A_h(\mathbf{d},\mathbf{c}) &A_h(\mathbf{d},\mathbf{c'}) &A_h(\mathbf{d},\mathbf{d}) &A_h(\mathbf{d},\mathbf{d'})
\\
b(\mathbf{d'},\mathbf{c}) &b(\mathbf{d'},\mathbf{c'}) &b(\mathbf{d'},\mathbf{d}) &b(\mathbf{d'},\mathbf{d'})
\end{matrix}\right),
\end{equation}
Define
\begin{equation}\label{3.501}
Q'_1(0)=\sum\limits_{r=0}^{\min(n_1,\Delta n)}\sum\limits_{s=0}^{n_1}\sum\limits_{t=0}^{\Delta n-1}
\sum\limits_{\substack{\mathbf{c}\in [1,n_1]_<^r\\\mathbf{c'}\in [1,n_1]_<^s}}\sum\limits_{\substack{\mathbf{d}\in [n_1+2,n_2]_<^r\\\mathbf{d'}\in [n_1+2,n_2]_<^t}}
\det V(\mathbf{c},\mathbf{c'},\mathbf{d},\mathbf{d'})
\end{equation}
and
\begin{equation}\label{3.502}
Q'_2(0)=\sum\limits_{r=1}^{\min(n_1,\Delta n)}\sum\limits_{s=0}^{n_1}\sum\limits_{t=0}^{\Delta n-1}
\sum\limits_{\substack{\mathbf{c}\in [1,n_1]_<^r\\\mathbf{c'}\in [1,n_1]_<^s}}\sum\limits_{\substack{\mathbf{d}\in [n_1+2,n_2]_<^{r-1}\\\mathbf{d'}\in [n_1+2,n_2]_<^t}}
\det U(\mathbf{c},\mathbf{c'},\mathbf{d},\mathbf{d'}).
\end{equation}
If $A$ is an $n\times n$-matrix and $1\le i,j\le n$, we let $A(\{i\}',\{j\}')$ denote the matrix $A$ with row $i$ and column $j$ removed.
Set (recall $L=2r+s+t$)
\begin{equation}\label{3.503}
Q'_3(0)=\sum\limits_{r=0}^{\min(n_1,\Delta n)}\sum\limits_{s=1}^{n_1}\sum\limits_{t=1}^{\Delta n}
\sum\limits_{\substack{\mathbf{c}\in [1,n_1]_<^r\\\mathbf{c'}\in [1,n_1]_<^s\\c'_s=n_1}}\sum\limits_{\substack{\mathbf{d}\in [n_1+2,n_2]_<^{r}\\\mathbf{d'}\in [n_1+1,n_2]_<^t\\d'_1=n_1+1}}
\sum\limits_{j=1}^L(-1)^{r+s+j}a_3^\ast(f_j)\det M_0(\{r+s\}',\{j\}'),
\end{equation}
where we use the notation
\begin{equation}\label{3.504}
f_j=\begin{cases}
c_j   & \text{if }1\le j\le r
\\
  c'_{j-r}     & \text{if } r<j\le r+s
\\
d_{j-r-s}       & \text{if } r+s<j\le 2r+s
\\
 d'_{j-2r-s}      & \text{if } 2r+s<j\le L
  \end{cases}.
\end{equation}
Also, set
\begin{equation}\label{3.505}
Q'_4(0)=\sum\limits_{r=0}^{\min(n_1,\Delta n)}\sum\limits_{\substack{s=0\\r+s\ge 1}}^{n_1}\sum\limits_{t=1}^{\Delta n}
\sum\limits_{\substack{\mathbf{c}\in [1,n_1]_<^r\\\mathbf{c'}\in [1,n_1]_<^s}}\sum\limits_{\substack{\mathbf{d}\in [n_1+2,n_2]_<^{r}\\\mathbf{d'}\in [n_1+1,n_2]_<^t\\d'_1=n_1+1}}
\sum\limits_{i=r+1}^{2r+s}\sum\limits_{j=1}^L(-1)^{i+j+1}a_2^\ast(f_i)a_3^\ast(f_j)\det M_0(\{i\}',\{j\}').
\end{equation}
Similarly, we define
\begin{equation}\label{3.506}
Q'_5(0)=\sum\limits_{r=1}^{\min(n_1,\Delta n)}\sum\limits_{s=1}^{n_1}\sum\limits_{t=0}^{\Delta n}
\sum\limits_{\substack{\mathbf{c}\in [1,n_1]_<^r\\\mathbf{c'}\in [1,n_1]_<^s\\c'_s=n_1}}\sum\limits_{\substack{\mathbf{d}\in [n_1+1,n_2]_<^{r}\\\mathbf{d'}\in [n_1+2,n_2]_<^t\\d_1=n_1+1}}
\sum\limits_{j=1}^L(-1)^{r+s+j}a_3^\ast(f_j)\det M_0(\{r+s\}',\{j\}'),
\end{equation}
and
\begin{equation}\label{3.506'}
Q'_6(0)=\sum\limits_{r=1}^{\min(n_1,\Delta n)}\sum\limits_{s=0}^{n_1}\sum\limits_{t=0}^{\Delta n}
\sum\limits_{\substack{\mathbf{c}\in [1,n_1]_<^r\\\mathbf{c'}\in [1,n_1]_<^s}}\sum\limits_{\substack{\mathbf{d}\in [n_1+1,n_2]_<^{r}\\\mathbf{d'}\in [n_1+2,n_2]_<^t\\d_1=n_1+1}}
\sum\limits_{i=r+1}^{2r+s}\sum\limits_{j=1}^L(-1)^{i+j}a_2^\ast(f_i)a_3^\ast(f_j)\det M_0(\{i\}',\{j\}').
\end{equation}
Note that the expressions $Q_1'(0)$ and $Q_2'(0)$ have a very similar structure, and this is true also for
$Q_3'(0)$ and $Q_5'(0)$, as well as for $Q_4'(0)$ and $Q_6'(0)$.

In section \ref{sect4} we will compute the asymptotics of $Q'_k(0)$, $1\le k\le 6$, which is all we need because of the next lemma
and proposition \ref{prop2.4}.

\begin{lemma}\label{lem3.4} We have the formula
\begin{equation}\label{3.507}
\left.\frac{\partial}{\partial h}\right|_{h=0} Q(h)=\sum\limits_{k=1}^6 Q'_k(0),
\end{equation}
with $Q'_k(0)$ as defined above.
\end{lemma}

\begin{proof}
From (\ref{3.31}) we see that $B(\ell,k)=b(\ell,k)$ unless $k=\ell=n_1+1$ in which case $B(n_1+1,n_1+1)=-1+b(n_1+1,n_1+1)$. This case can 
occur in the formula (\ref{3.11}) for $Q(h)$ if and only if $d'_1=n_1+1$, which requires $t\ge 1$. Let $E_{2r+s+1}$ be the matrix which is zero everywhere
except at position $(2r+s+1,2r+s+1)$ where it is $=1$. In the sum in (\ref{3.11}) we can assume that $d_1\neq d'_1$ since otherwise the 
determinant is $=0$. Hence, by (\ref{3.11}) we can write
\begin{equation}\label{3.51}
Q(h)=q_0(h)+q_1(h)+q_2(h),
\end{equation}
where
\begin{equation}\label{3.52}
q_0(h)=\sum\limits_{r=0}^{\min(n_1,\Delta n)}\sum\limits_{s=0}^{n_1}\sum\limits_{t=0}^{\Delta n}
\sum\limits_{\substack{\mathbf{c}\in [1,n_1]_<^r\\\mathbf{c'}\in [1,n_1]_<^s}}\sum\limits_{\substack{\mathbf{d}\in [n_1+2,n_2]_<^r\\\mathbf{d'}\in [n_1+2,n_2]_<^t\\d_1\neq d'_1}}
\det M_h,
\end{equation}
where $M_h$ is given by (\ref{3.50}),
\begin{equation}\label{3.53}
q_1(h)=\sum\limits_{r=0}^{\min(n_1,\Delta n)}\sum\limits_{s=0}^{n_1}\sum\limits_{t=1}^{\Delta n}
\sum\limits_{\substack{\mathbf{c}\in [1,n_1]_<^r\\\mathbf{c'}\in [1,n_1]_<^s}}\sum\limits_{\substack{\mathbf{d}\in [n_1+2,n_2]_<^r\\\mathbf{d'}\in [n_1+1,n_2]_<^t\\d'_1=n_1+1}}
\det\left(-E_{2r+s+1}+M_h\right),
\end{equation}
and
\begin{equation}\label{3.54}
q_2(h)=\sum\limits_{r=1}^{\min(n_1,\Delta n)}\sum\limits_{s=0}^{n_1}\sum\limits_{t=0}^{\Delta n}
\sum\limits_{\substack{\mathbf{c}\in [1,n_1]_<^r\\\mathbf{c'}\in [1,n_1]_<^s}}\sum\limits_{\substack{\mathbf{d}\in [n_1+1,n_2]_<^r\\\mathbf{d'}\in [n_1+2,n_2]_<^t\\d_1=n_1+1}}
\det M_h.
\end{equation}
We see from (\ref{3.53}) that
\begin{align}\label{3.55}
q_1(h)&=\sum\limits_{r=0}^{\min(n_1,\Delta n)}\sum\limits_{s=0}^{n_1}\sum\limits_{t=1}^{\Delta n}
\sum\limits_{\substack{\mathbf{c}\in [1,n_1]_<^r\\\mathbf{c'}\in [1,n_1]_<^s}}
\sum\limits_{\substack{\mathbf{d}\in [n_1+2,n_2]_<^r\\\mathbf{d'}\in [n_1+1,n_2]_<^t\\d'_1=n_1+1}}
\det M_h
\\
&-\sum\limits_{r=0}^{\min(n_1,\Delta n)}\sum\limits_{s=0}^{n_1}\sum\limits_{t=0}^{\Delta n-1}
\sum\limits_{\substack{\mathbf{c}\in [1,n_1]_<^r\\\mathbf{c'}\in [1,n_1]_<^s}}
\sum\limits_{\substack{\mathbf{d}\in [n_1+2,n_2]_<^r\\\mathbf{d'}\in [n_1+2,n_2]_<^t}}
\det M_h:=q_3(h)-q_4(h).
\notag
\end{align}
Note that
\begin{equation}
q_0(h)-q_4(h)=\sum\limits_{r=0}^{\min(n_1,\Delta n)}\sum\limits_{s=0}^{n_1}
\sum\limits_{\substack{\mathbf{c}\in [1,n_1]_<^r\\\mathbf{c'}\in [1,n_1]_<^s}}
\sum\limits_{\substack{\mathbf{d}\in [n_1+2,n_2]_<^r\\\mathbf{d'}\in [n_1+2,n_2]_<^{\Delta n}}}\det M_h=0
\notag
\end{equation}
since $[n_1+2,n_2]_<^{\Delta n}=\emptyset$. Thus, by (\ref{3.51}),
\begin{equation}\label{3.56}
Q(h)=q_2(h)+q_3(h).
\end{equation}

If $A$ is a matrix and $\mathbf{v}$ a row vector, $(A|\mathbf{v})_{\text{row\,}(i)}$ will denote the matrix obtained by replacing row $i$ in $A$
with $\mathbf{v}$. Similarly, if  $\mathbf{v}$ is a column vector, 
$(A|\mathbf{v})_{\text{col\,}(j)}$ will denote the matrix obtained by replacing column $j$ in $A$
with $\mathbf{v}$.
Let
\begin{equation}\label{3.57}
\mathbf{v}_i=\left(\begin{matrix} A_0^\ast(f_i,\mathbf{c}) &A_0^\ast(f_i,\mathbf{c'}) &A_0^\ast(f_i,\mathbf{d}) &A_0^\ast(f_i,\mathbf{d'})
\end{matrix}\right),
\end{equation}
where $A_0^\ast$ is given by (\ref{3.29}); recall  (\ref{3.32}). We see then that
\begin{align}\label{3.58}
q_3'(0)&=\left.\frac{\partial}{\partial h}\right|_{h=0} q_3(h)
\\
&=\sum\limits_{r=0}^{\min(n_1,\Delta n)}\sum\limits_{\substack{s=0\\r+s\ge 1}}^{n_1}\sum\limits_{t=1}^{\Delta n}
\sum\limits_{\substack{\mathbf{c}\in [1,n_1]_<^r\\\mathbf{c'}\in [1,n_1]_<^s}}
\sum\limits_{\substack{\mathbf{d}\in [n_1+2,n_2]_<^r\\\mathbf{d'}\in [n_1+1,n_2]_<^t\\d'_1=n_1+1}}
\sum\limits_{i=r+1}^{2r+s}\det (M_0|\mathbf{v}_i)_{\text{row\,}(i)}.
\notag
\end{align}
We have to have $r+s\ge 1$ to get a non-zero contribution when taking the $h$-derivative. Similarly,
\begin{align}\label{3.60}
q_2'(0)&=\left.\frac{\partial}{\partial h}\right|_{h=0} q_2(h)
\\
&=\sum\limits_{r=1}^{\min(n_1,\Delta n)}\sum\limits_{s=0}^{n_1}\sum\limits_{t=1}^{\Delta n}
\sum\limits_{\substack{\mathbf{c}\in [1,n_1]_<^r\\\mathbf{c'}\in [1,n_1]_<^s}}
\sum\limits_{\substack{\mathbf{d}\in [n_1+1,n_2]_<^r\\\mathbf{d'}\in [n_1+2,n_2]_<^t\\d_1=n_1+1}}
\sum\limits_{i=r+1}^{2r+s}\det (M_0|\mathbf{v}_i)_{\text{row\,}(i)}.
\notag
\end{align}
Expand the determinant in (\ref{3.58}) along row $i$. This gives
\begin{align}\label{3.61}
q_3'(0)=&\sum\limits_{r=0}^{\min(n_1,\Delta n)}\sum\limits_{\substack{s=0\\r+s\ge 1}}^{n_1}\sum\limits_{t=1}^{\Delta n}
\sum\limits_{\substack{\mathbf{c}\in [1,n_1]_<^r\\\mathbf{c'}\in [1,n_1]_<^s}}
\sum\limits_{\substack{\mathbf{d}\in [n_1+2,n_2]_<^r\\\mathbf{d'}\in [n_1+1,n_2]_<^t\\d'_1=n_1+1}}
\sum\limits_{i=r+1}^{2r+s}\sum\limits_{j=1}^{L}(-1)^{i+j+1}
\\
&\times\left(\delta_{f_i,n_1}-a_2^\ast(f_i)\right)\left(\delta_{f_j,n_1+1}-a_3^\ast(f_j)\right)\det M_0(\{i\}',\{j\}'),
\notag
\end{align}
where we have used (\ref{3.29c}) and (\ref{3.57}). Now,
$$
\left(\delta_{f_i,n_1}-a_2^\ast(f_i)\right)\left(\delta_{f_j,n_1+1}-a_3^\ast(f_j)\right)=
\delta_{f_i,n_1}\delta_{f_j,n_1+1}-\delta_{f_i,n_1}a_3^\ast(f_j)-\delta_{f_j,n_1+1}a_2^\ast(f_i)+a_2^\ast(f_i)a_3^\ast(f_j)
$$
leads to a corresponding decomposition
\begin{equation}\label{3.62}
q_3'(0)=q_{3,1}'(0)+q_{3,2}'(0)+q_{3,3}'(0)+q_{3,4}'(0).
\end{equation}
We will now show that $Q_1'(0)=q_{3,1}'(0)+q_{3,3}'(0)$, $Q_3'(0)=q_{3,2}'(0)$ and $Q_4'(0)=q_{3,4}'(0)$. A similar argument for (\ref{3.60}) will give $Q_2'(0)+Q_5'(0)+Q_6'(0)$.

The term $\delta_{f_i,n_1}\delta_{f_j,n_1+1}$ requires $j=2r+s+1$ and $f_{2r+s+1}=d'_1=n_1+1$, and $i=r+s$ and $f_{r+s}=c'_{s}=n_1$. Hence,
$s\ge 1$ and we obtain
\begin{equation}\label{3.63}
q_{3,1}'(0)=\sum\limits_{r=0}^{\min(n_1,\Delta n)}\sum\limits_{s=1}^{n_1}\sum\limits_{t=1}^{\Delta n}
\sum\limits_{\substack{\mathbf{c}\in [1,n_1]_<^r\\\mathbf{c'}\in [1,n_1]_<^s\\c'_s=n_1}}
\sum\limits_{\substack{\mathbf{d}\in [n_1+2,n_2]_<^r\\\mathbf{d'}\in [n_1+1,n_2]_<^t\\d'_1=n_1+1}}(-1)^r\det M_0(\{r+s\}',\{2r+s+1\}').
\end{equation}
The term $-\delta_{f_i,n_1}a_3^\ast(f_j)$ requires $i=r+s$, and gives
\begin{equation}\label{3.64}
q_{3,2}'(0)=\sum\limits_{r=0}^{\min(n_1,\Delta n)}\sum\limits_{s=1}^{n_1}\sum\limits_{t=1}^{\Delta n}
\sum\limits_{\substack{\mathbf{c}\in [1,n_1]_<^r\\\mathbf{c'}\in [1,n_1]_<^s\\c'_s=n_1}}
\sum\limits_{\substack{\mathbf{d}\in [n_1+2,n_2]_<^r\\\mathbf{d'}\in [n_1+1,n_2]_<^t\\d'_1=n_1+1}}
\sum\limits_{j=1}^{L}(-1)^{r+s+j}a_3^\ast(f_j)\det M_0(\{r+s\}',\{j\}'),
\end{equation}
which is equal to $Q_3'(0)$ as defined by (\ref{3.503}). The term $-\delta_{f_j,n_1+1}a_2^\ast(f_i)$ requires $j=2r+s+1$, which gives
\begin{equation}
q_{3,3}'(0)=\sum\limits_{r=0}^{\min(n_1,\Delta n)}\sum\limits_{\substack{s=0\\r+s\ge 1}}^{n_1}\sum\limits_{t=1}^{\Delta n}
\sum\limits_{\substack{\mathbf{c}\in [1,n_1]_<^r\\\mathbf{c'}\in [1,n_1]_<^s}}
\sum\limits_{\substack{\mathbf{d}\in [n_1+2,n_2]_<^r\\\mathbf{d'}\in [n_1+1,n_2]_<^t\\d'_1=n_1+1}}
\sum\limits_{i=r+1}^{2r+s}(-1)^{i+2r+s+1}a_2^\ast(f_i)\det M_0(\{i\}',\{2r+s+1\}').
\notag\end{equation}
If we write
\begin{equation}\label{3.65}
\mathbf{a}_2^\ast=\left(\begin{matrix}
0 \\
a_2^\ast(\mathbf{c'}) \\
a_2^\ast(\mathbf{d}) \\
0
\end{matrix}\right),
\end{equation}
where the blocks have length $r,s,r$ and $t$ respectively, we see that
\begin{equation}\label{3.67}
q_{3,3}'(0)=\sum\limits_{r=0}^{\min(n_1,\Delta n)}\sum\limits_{\substack{s=0\\r+s\ge 1}}^{n_1}\sum\limits_{t=1}^{\Delta n}
\sum\limits_{\substack{\mathbf{c}\in [1,n_1]_<^r\\\mathbf{c'}\in [1,n_1]_<^s}}
\sum\limits_{\substack{\mathbf{d}\in [n_1+2,n_2]_<^r\\\mathbf{d'}\in [n_1+1,n_2]_<^t\\d'_1=n_1+1}}
\det (M_0|\mathbf{a}_2^\ast)_{\text{col\,}(2r+s+1)}.
\end{equation}
Finally, we get
\begin{equation}
q_{3,4}'(0)=\sum\limits_{r=0}^{\min(n_1,\Delta n)}\sum\limits_{\substack{s=0\\r+s\ge 1}}^{n_1}\sum\limits_{t=1}^{\Delta n}
\sum\limits_{\substack{\mathbf{c}\in [1,n_1]_<^r\\\mathbf{c'}\in [1,n_1]_<^s}}
\sum\limits_{\substack{\mathbf{d}\in [n_1+2,n_2]_<^r\\\mathbf{d'}\in [n_1+1,n_2]_<^t\\d'_1=n_1+1}}
\sum\limits_{i=r+1}^{2r+s}\sum\limits_{j=1}^{L}(-1)^{i+j+1}a_2^\ast(f_i)a_3^\ast(f_j)\det M_0(\{i\}',\{j\}'),
\notag\end{equation}
which is $Q_4'(0)$. We can now split (\ref{3.60}) in the same way,
\begin{equation}\label{3.69}
q_2'(0)=q_{2,1}'(0)+q_{2,2}'(0)+q_{2,3}'(0)+q_{2,4}'(0),
\end{equation}
where
\begin{equation}\label{3.70}
q_{2,1}'(0)=\sum\limits_{r=1}^{\min(n_1,\Delta n)}\sum\limits_{s=1}^{n_1}\sum\limits_{t=0}^{\Delta n}
\sum\limits_{\substack{\mathbf{c}\in [1,n_1]_<^r\\\mathbf{c'}\in [1,n_1]_<^s\\c'_s=n_1}}
\sum\limits_{\substack{\mathbf{d}\in [n_1+1,n_2]_<^r\\\mathbf{d'}\in [n_1+2,n_2]_<^t\\d_1=n_1+1}}\det M_0(\{r+s\}',\{2r+s+1\}'),
\end{equation}
$q_{2,2}'(0)=Q_5'(0)$, with $Q_5'(0)$ given by (\ref{3.506}),
\begin{equation}\label{3.72}
q_{2,3}'(0)=\sum\limits_{r=1}^{\min(n_1,\Delta n)}\sum\limits_{s=0}^{n_1}\sum\limits_{t=0}^{\Delta n}
\sum\limits_{\substack{\mathbf{c}\in [1,n_1]_<^r\\\mathbf{c'}\in [1,n_1]_<^s}}
\sum\limits_{\substack{\mathbf{d}\in [n_1+1,n_2]_<^r\\\mathbf{d'}\in [n_1+2,n_2]_<^t\\d_1=n_1+1}}
\det (M_0|\mathbf{a}_2^\ast)_{\text{col\,}(r+s+1)},
\end{equation}
and, with $Q_6'(0)$ given by (\ref{3.506'}), $q_{2,4}'(0)=Q_6'(0)$.

From (\ref{3.56}), (\ref{3.62}) and  (\ref{3.69}) we see that
\begin{equation}\label{3.74}
\left.\frac{\partial}{\partial h}\right|_{h=0}Q(h)
= q_{3,1}'(0)+q_{3,3}'(0)+q_{2,1}'(0)+q_{2,3}'(0)+\sum\limits_{k=3}^6Q_k'(0).
\end{equation}
In order to prove the lemma it remains to show that
\begin{equation}\label{3.75}
Q_1'(0)=q_{3,1}'(0)+q_{3,3}'(0)\,\,,\,\,Q_2'(0)=q_{2,1}'(0)+q_{2,3}'(0).
\end{equation}
In the expression (\ref{3.63}) for $q_{3,1}'(0)$ we move row $2r+s+1$ to row $r+s+1$. This gives a sign change $(-1)^r$. 
We then shift the $s$-and $t$-summations by 1, using the fact that $d'_1=n_1+1$ and $c'_s=n_1$ are fixed. This gives
\begin{align}\label{3.81}
q_{3,1}'(0)=&\sum\limits_{r=0}^{\min(n_1,\Delta n)}\sum\limits_{s=0}^{n_1-1}\sum\limits_{t=0}^{\Delta n-1}
\sum\limits_{\substack{\mathbf{c}\in [1,n_1]_<^r\\\mathbf{c'}\in [1,n_1-1]_<^s}}
\sum\limits_{\substack{\mathbf{d}\in [n_1+2,n_2]_<^r\\\mathbf{d'}\in [n_1+2,n_2]_<^t}}
\\
&\det\left( \begin{matrix}
b(\mathbf{c},\mathbf{c}) &b(\mathbf{c},\mathbf{c'}) &b(\mathbf{c},n_1) &b(\mathbf{c},\mathbf{d}) &b(\mathbf{c},\mathbf{d'})
\\
a_0(\mathbf{c'},\mathbf{c}) &a_0(\mathbf{c'},\mathbf{c'}) &a_0(\mathbf{c'},n_1) &a_0(\mathbf{c'},\mathbf{d}) &a_0(\mathbf{c'},\mathbf{d'})
\\
b(n_1+1,\mathbf{c}) &b(n_1+1,\mathbf{c'}) &b(n_1+1,n_1) &b(n_1+1,\mathbf{d}) &b(n_1+1,\mathbf{d'})
\\
a_0(\mathbf{d},\mathbf{c}) &a_0(\mathbf{d},\mathbf{c'}) &a_0(\mathbf{d},n_1) &a_0(\mathbf{d},\mathbf{d}) &a_0(\mathbf{d},\mathbf{d'})
\\
b(\mathbf{d'},\mathbf{c}) &b(\mathbf{d'},\mathbf{c'}) &b(\mathbf{d'},n_1) &b(\mathbf{d'},\mathbf{d}) &b(\mathbf{d'},\mathbf{d'}).
\end{matrix}\right)
\notag
\end{align}
In the expression (\ref{3.67}) for $q_{3,3}'(0)$ we move row $2r+s+1$ to row $r+s+1$ and column $2r+s+1$ to column $r+s+1$. This
gives no net sign change. Note that if $r+s=0$ then $\mathbf{a}_2^\ast=0$ so we can remove the condition $r+s\ge 1$ in the summation
in (\ref{3.67}). Also, we shift the $t$-summation by 1. We obtain
\begin{align}\label{3.82}
q_{3,3}'(0)=&\sum\limits_{r=0}^{\min(n_1,\Delta n)}\sum\limits_{s=0}^{n_1}\sum\limits_{t=0}^{\Delta n-1}
\sum\limits_{\substack{\mathbf{c}\in [1,n_1]_<^r\\\mathbf{c'}\in [1,n_1]_<^s}}
\sum\limits_{\substack{\mathbf{d}\in [n_1+2,n_2]_<^r\\\mathbf{d'}\in [n_1+2,n_2]_<^t}}
\\
&\det\left( \begin{matrix}
b(\mathbf{c},\mathbf{c}) &b(\mathbf{c},\mathbf{c'}) &0    &b(\mathbf{c},\mathbf{d}) &b(\mathbf{c},\mathbf{d'})
\\
a_0(\mathbf{c'},\mathbf{c}) &a_0(\mathbf{c'},\mathbf{c'}) &a_2^\ast(\mathbf{c'}) &a_0(\mathbf{c'},\mathbf{d}) &a_0(\mathbf{c'},\mathbf{d'})
\\
b(n_1+1,\mathbf{c}) &b(n_1+1,\mathbf{c'}) &0    &b(n_1+1,\mathbf{d}) &b(n_1+1,\mathbf{d'})
\\
a_0(\mathbf{d},\mathbf{c}) &a_0(\mathbf{d},\mathbf{c'}) &a_2^\ast(\mathbf{d})   &a_0(\mathbf{d},\mathbf{d}) &a_0(\mathbf{d},\mathbf{d'})
\\
b(\mathbf{d'},\mathbf{c}) &b(\mathbf{d'},\mathbf{c'}) &0     &b(\mathbf{d'},\mathbf{d}) &b(\mathbf{d'},\mathbf{d'})
\end{matrix}\right).
\notag
\end{align}
Note that the $\mathbf{c'}$-summation in (\ref{3.81}) can be extended to  $[1,n_1]_<^s$, since if $c'_s=n_1$, then two columns in the
determinant in (\ref{3.81}) are equal. In fact, if $s\ge 1$ so that the sum is non-trivial, then extending the summation to $\mathbf{c'}\in[1,n_1]_<^s$ means that we also have the case
$c_s'=n_1$, and in this case the columns $r+s$ and $r+s+1$ are equal.
Also, we can extend the $s-$summation to $s=n_1$ since in that case we must have $c'_s=n_1$. We can thus add the two formulas 
(\ref{3.81}) and (\ref{3.82}) and this gives the first formula in (\ref{3.75}) with $\tilde{a}_0(\ell, n_1)=a_0((\ell, n_1)+a_2^\ast(\ell)$, 
which agrees with (\ref{3.492}). The proof of the second formula in (\ref{3.75}) is analogous.

\end{proof}

\section{Asymptotics and proof of the main theorem}\label{sect4}

We begin by recalling some notation from section \ref{sect1}, (\ref{1.8}). Let $\lambda_i=\eta_i-\nu_i^2$, $i=1,2$ and write
$$
\Delta\lambda=\lambda_2\left(\frac{t_2}{\Delta t}\right)^{1/3}-\lambda_1\left(\frac{t_1}{\Delta t}\right)^{1/3}.
$$
Then,
\begin{equation}\label{4.1}
\Delta\eta=\Delta\lambda+\Delta\nu^2,
\end{equation}
where $\Delta\nu$ is given by (\ref{1.8}). We will write
\begin{equation}\label{4.2}
N_1=t_1M\,\,,\,\, N_2=\Delta t M,
\end{equation}
where we will let $M\to\infty$ as in theorem \ref{thm1.1}. The scalings in  (\ref{scaling}) and in the arguments $\ell, k$ can then be written
\begin{align}\label{4.3}
n_1&=N_1+\nu_1N_1^{2/3}\,,\,n_2=t_2M+\nu_2(t_2M)^{2/3}\,,\,\Delta n=n_2-n_1=N_2+\Delta\nu N_2^{2/3}
\\
\mu_1&=N_1-\nu_1N_1^{2/3}\,,\,\mu_2=t_2M-\nu_2(t_2M)^{2/3}\,,\,\Delta \mu=\mu_2-\mu_1=N_2-\Delta\nu N_2^{2/3}
\notag\\
\xi_1&=2N_1+\lambda_1N_1^{1/3}\,,\,\xi_2=2t_2M+\lambda_2(t_2M)^{1/3}\,,\,\Delta \xi=\xi_2-\xi_1=2N_2+\Delta\lambda N_2^{2/3}
\notag\\
\ell&=n_1+1+xN_1^{1/3}\,,\,k=n_1+yN_1^{1/3},
\notag
\end{align}
where we have ignored integer parts.

We will now state two lemmas that we will need in order to prove theorem \ref{thm1.1} from proposition \ref{prop3.2} and lemma \ref{lem3.4}. The proofs
of the lemmas is postponed to section \ref{sect6}.

\begin{lemma}\label{lem4.1}
Recall (\ref{3.24}) to (\ref{3.27}). Under the scalings (\ref{4.3}) with $N_1,N_2$ given by (\ref{4.2}) we have the following limits, uniformly for 
$\nu_i, \eta_i,x,y$ in compact sets,
\begin{equation}\label{4.4}
\lim_{M\to\infty} N_1^{1/3}a_{0,1}(\ell,k)=\phi_1(x,y),
\end{equation}
\begin{equation}\label{4.5}
\lim_{M\to\infty} N_1^{1/3}b_1(\ell,k)=\psi_1(x,y),
\end{equation}
\begin{equation}\label{4.6}
\lim_{M\to\infty} N_1^{1/3}c_2(\ell,k)=\phi_2(x,y),
\end{equation}
\begin{equation}\label{4.7}
\lim_{M\to\infty} N_1^{1/3}c_3(\ell,k)=\phi_3(x,y),
\end{equation}
where $\phi_i,\psi_1$ are given by (\ref{1.10}) to  (\ref{1.13}).
\end{lemma}

We will also need some estimates in order to control the convergence of the whole expansion.

\begin{lemma}\label{lem4.2}
Assume that we have the scalings (\ref{4.3}) with $N_1,N_2$ given by (\ref{4.2}). There are constants $c,C>0$, which depend on 
$t_i, \nu_i, \eta_i$, such that for all $M\ge 1$,
\begin{equation}\label{4.8}
\left|N_1^{1/3}a_{0,1}(\ell,k)\right|\le Ce^{-c(x_+^{3/2}+(-y)_+^{3/2})+C(y_++(-x)_+)},
\end{equation}
\begin{equation}\label{4.9}
\left|N_1^{1/3}b_1(\ell,k)\right|\le Ce^{-c(x_+^{3/2}+(-y)_+^{3/2})+C(y_++(-x)_+)},
\end{equation}
\begin{equation}\label{4.10}
\left|N_1^{1/3}c_2(\ell,k)\right|\le Ce^{-c(x_+^{3/2}+y_+^{3/2})+C((-y)_+ +(-x)_+)},
\end{equation}
\begin{equation}\label{4.11}
\left|N_1^{1/3}c_3(\ell,k)\right|\le Ce^{-c((-x)_+^{3/2}+(-y)_+^{3/2})+C(y_++x_+)},
\end{equation}
for all $1\le\ell,k\le n_2$, where $a_+=\max(0,a)$.
\end{lemma}

As an immediate corollary of this lemma and the definitions (\ref{3.28}), (\ref{3.29}) and (\ref{3.492}), we obtain

\begin{corollary}\label{cor4.3}
Assume that we have the scalings (\ref{4.3}) with $N_1,N_2$ given by (\ref{4.2}). There are constants $c,C>0$, which depend on 
$t_i, \nu_i, \eta_i$, such that for all $M\ge 1$, and all $1\le\ell,k\le n_2$,
\begin{subequations}\label{4.12}
\begin{align}
\left|N_1^{1/3}a_{0}(\ell,k)\right|&\le Ce^{-c(x_+^{3/2}+(-y)_+^{3/2})+C(y_++(-x)_+)},
\\
\left|N_1^{1/3}b(\ell,k)\right|&\le Ce^{-c(x_+^{3/2}+(-y)_+^{3/2})+C(y_++(-x)_+)},
\\
\left|N_1^{1/3}\tilde{a}_{0}(\ell,n_1)\right|&\le Ce^{-cx_+^{3/2}+C(-x)_+},
\\
\left|N_1^{1/3}a_2^\ast(\ell)\right|&\le Ce^{-cx_+^{3/2}+C(-x)_+},
\\
\left|N_1^{1/3}a_3^\ast(k)\right|&\le Ce^{-c(-y)_+^{3/2}+Cy_+}.
\end{align}
\end{subequations}
\end{corollary}

 Recall the formula (\ref{3.507}) in lemma \ref{lem3.4}. We want to control the terms $Q_k'(0)$ asymptotically as $M\to\infty$.
\begin{lemma}\label{lem4.4}
We have the following limits.
\begin{equation}\label{4.14}
\lim_{M\to\infty} N_1^{1/3}Q_k'(0)=0
\end{equation}
for $3\le k\le 6$,
\begin{align}\label{4.16}
&\Psi^{(1)}(\eta_1,\eta_2):=\lim_{M\to\infty} N_1^{1/3}Q_1'(0)
\\
&=\sum\limits_{r,s,t=0}^\infty\frac 1{(r!)^2s!t!}\int\limits_{(-\infty,0]^r}d^rx\int\limits_{(-\infty,0]^s}d^sx'
\int\limits_{[0,\infty)^r}d^ry\int\limits_{[0,\infty)^t}d^ty' 
W_{r,s,r,t}^{(1)}(\mathbf{x},\mathbf{x'},\mathbf{y},\mathbf{y'}),
\notag
\end{align}
where $W_{r,s,r,t}^{(1)}$ is given by (\ref{1.16}), and
\begin{align}\label{4.17}
&\Psi^{(2)}(\eta_1,\eta_2):=\lim_{M\to\infty} N_1^{1/3}Q_2'(0)
\\
&=\sum\limits_{r=1,s,t=0}^\infty\frac 1{r!s!(r-1)!t!}\int\limits_{(-\infty,0]^r}d^rx\int\limits_{(-\infty,0]^s}d^sx'
\int\limits_{[0,\infty)^{r-1}}d^{r-1}y\int\limits_{[0,\infty)^t}d^ty' 
W_{r,s,r-1,t}^{(2)}(\mathbf{x},\mathbf{x'},\mathbf{y},\mathbf{y'}),
\notag
\end{align}
where $W_{r,s,r,t}^{(2)}$ is given by (\ref{1.17}).
\end{lemma}

\begin{proof}
Consider $M_0$ given by (\ref{3.50}) with $h=0$. Recall (\ref{3.30}). Let $[M_0]_i$ denote the $i$:th row in $M_0$, and $[M_0]^j$
the $j$:th column. We will use the following scalings
\begin{align}\label{4.18}
c_i&=n_1+x_iN_1^{1/3}\,\,,\,\,  c'_i=n_1+x'_iN_1^{1/3}
\\
d_i&=n_1+1+y_iN_1^{1/3}\,\,,\,\,  d'_i=n_1+1+y'_iN_1^{1/3}
\notag
\end{align}
so that $x_i\le 0$, $x'_i\le 0$, $y_i\ge 0$, and $y'_i\ge 0$.
Set
$$
Y_{\max}=\max\limits_{1\le j\le r} y_j+\max\limits_{1\le j\le t} y'_j\,,\,
X_{\max}=\max\limits_{1\le j\le r} (-x_j)+\max\limits_{1\le j\le s} (-x'_j).
$$
It follows from corollary \ref{cor4.3} that under the scaling (\ref{4.18}) there exist constants $c,C>0$ such that
\begin{equation}\label{4.19}
\begin{cases}
||N_1^{1/3}[M_0]_i||_2\le CL^{1/2}e^{C(Y_{\max}-x_i)}  & \text{if } 1\le i\le r \\
 ||N_1^{1/3}[M_0]_i||_2\le CL^{1/2}e^{C(Y_{\max}-x'_{i-r})}     & \text{if } r< i\le r+s\\
 ||N_1^{1/3}[M_0]_i||_2\le CL^{1/2}e^{-c y_{i-(r+s)}^{3/2}+CY_{\max}} & \text{if } r+s< i\le 2r+s \\
||N_1^{1/3}[M_0]_i||_2\le CL^{1/2}e^{-c {y'}_{i-(2r+s)}^{3/2}+CY_{\max}}  & \text{if } 2r+s< i\le L
  \end{cases},
\end{equation}
where $L=2r+s+t$, and
\begin{equation}\label{4.20}
\begin{cases}
||N_1^{1/3}[M_0]^j||_2\le  CL^{1/2}e^{-c (-x_j)_{i}^{3/2}+CX_{\max}} & \text{if } 1\le j\le r \\
 ||N_1^{1/3}[M_0]^j||_2\le CL^{1/2}e^{-c (-x_{j-r})_{i}^{3/2}+CX_{\max}}    & \text{if } r< j\le r+s\\
 ||N_1^{1/3}[M_0]^j||_2\le  CL^{1/2}e^{C(X_{\max}-y_{j-(r+s)})} & \text{if } r+s< j\le 2r+s \\
||N_1^{1/3}[M_0]^j||_2\le  CL^{1/2}e^{C(X_{\max}-{y'}_{j-(2r+s)})} & \text{if } 2r+s< j\le L
  \end{cases}.
\end{equation}

From Hadamard's inequality we get the estimates
$$
\left|\det \left(N_1^{1/3}M_0\right)\right|\le\prod\limits_{i=1}^L||N_1^{1/3}[M_0]_i||_2,
$$
and
$$
\left|\det \left(N_1^{1/3}M_0\right)\right|\le\prod\limits_{j=1}^L||N_1^{1/3}[M_0]^j||_2,
$$
from which it follows by taking the product that
\begin{equation}\label{4.21}
\left|\det \left(N_1^{1/3}M_0\right)\right|\le\prod\limits_{i=1}^L||N_1^{1/3}[M_0]_i||_2^{1/2}\prod\limits_{j=1}^L||N_1^{1/3}[M_0]^j||_2^{1/2}.
\end{equation}
If we use the estimates (\ref{4.19}) and (\ref{4.20}) in (\ref{4.21}) we see that there are constants $c,C>0$ such that
\begin{equation}\label{4.22}
\left|\det \left(N_1^{1/3}M_0\right)\right|\le C^LL^{L/2}\prod\limits_{j=1}^re^{-c(-x_j)^{3/2}}\prod\limits_{j=1}^se^{-c(-x'_j)^{3/2}}
\prod\limits_{j=1}^re^{-cy_j^{3/2}}\prod\limits_{j=1}^te^{-c{y'}_j^{3/2}}.
\end{equation}
Here we have also used the fact that given a constant $c>0$, there is a constant $C$ so that
\begin{equation*}
Y_{\max}\le C+\frac c2 \sum\limits_{j=1}^r y_j^{3/2}+\frac c2 \sum\limits_{j=1}^ty_j'^{3/2},
\end{equation*}
and an analogous estimate for $X_{\max}$.

Consider now the expression for $Q_3'(0)$ in (\ref{3.503}). If we use the estimate
$$
\left|N_1^{1/3}a_3^\ast(n_1+yN_1^{1/3})\right|\le Ce^{-c(-y)_+^{3/2}+Cy_+}
$$
from (\ref{4.12}) and the same estimates and arguments as above we see that
\begin{align}\label{4.23}
&\left|N_1^{L/3}\sum\limits_{j=1}^L(-1)^{r+s+j}a_3^\ast(f_j)\det M_0(\{r+s\}',\{j\}')\right|
\\
&\le C^LL^{L/2}\prod\limits_{j=1}^re^{-c(-x_j)^{3/2}}\prod\limits_{j=1}^se^{-c(-x'_j)^{3/2}}
\prod\limits_{j=1}^re^{-cy_j^{3/2}}\prod\limits_{j=1}^te^{-c{y'}_j^{3/2}}.
\notag
\end{align}
Note that in (\ref{3.503}), $y_1'=0$ and $x'_s=0$, so if we write
\begin{align}
N_1^{1/3}Q'_3(0)&=\frac 1{N_1^{1/3}}\sum\limits_{r=0}^{\min(n_1,\Delta n)}\sum\limits_{s=1}^{n_1}\sum\limits_{t=1}^{\Delta n}
\sum\limits_{\substack{\mathbf{c}\in [1,n_1]_<^r\\\mathbf{c'}\in [1,n_1]_<^s\\c'_s=n_1}}\frac 1{N_1^{(r+s-1)/3}}
\sum\limits_{\substack{\mathbf{d}\in [n_1+2,n_2]_<^{r}\\\mathbf{d'}\in [n_1+1,n_2]_<^t\\d'_1=n_1+1}}\frac 1{N_1^{(r+t-1)/3}}
\notag\\
&\times N_1^{L/3}\sum\limits_{j=1}^L(-1)^{r+s+j}a_3^\ast(f_j)\det M_0(\{r+s\}',\{j\}'),
\notag
\end{align}
we see that we can control the convergence of the Riemann sum using (\ref{4.23}) (note ordered variables instead of factorials), but
since we have the factor $1/N_1^{1/3}$ in front of the whole expression we see that it $\to 0$ as $M\to\infty$. From (\ref{3.505}) we can write
\begin{align}
N_1^{1/3}Q'_4(0)&=\frac 1{N_1^{1/3}}\sum\limits_{r=0}^{\min(n_1,\Delta n)}\sum\limits_{\substack{s=0\\r+s\ge 1}}^{n_1}\sum\limits_{t=1}^{\Delta n}
\sum\limits_{\substack{\mathbf{c}\in [1,n_1]_<^r\\\mathbf{c'}\in [1,n_1]_<^s}}\frac 1{N_1^{(r+s)/3}}
\sum\limits_{\substack{\mathbf{d}\in [n_1+2,n_2]_<^{r}\\\mathbf{d'}\in [n_1+1,n_2]_<^t\\d'_1=n_1+1}}\frac 1{N_1^{(r+t-1)/3}}
\notag\\
&\times N_1^{(L+1)/3}\sum\limits_{i=r+1}^{2r+s}\sum\limits_{j=1}^L(-1)^{i+j+1}a_2^\ast(f_i)a_3^\ast(f_j)\det M_0(\{i\}',\{j\}').
\notag
\end{align}
Using the estimates of $a_2^\ast$  and $a_3^\ast$ from corollary \ref{cor4.3} it follows that we can prove an estimate analogous to (\ref{4.23}) 
and again we see that $N_1^{1/3}Q'_4(0)\to 0$ as $M\to\infty$. This proves (\ref{4.14}) for $k=3,4$. The proof for $k=5,6$ is a analogous.

From the estimates in corollary \ref{cor4.3} we see that in analogy with the proof of (\ref{4.22}) we can prove
\begin{equation}\label{4.24}
\left|\det \left(N_1^{1/3}V\right)\right|\le C^{L+1}(L+1)^{(L+1)/2}\prod\limits_{j=1}^re^{-c(-x_j)^{3/2}}\prod\limits_{j=1}^se^{-c(-x'_j)^{3/2}}
\prod\limits_{j=1}^re^{-cy_j^{3/2}}\prod\limits_{j=1}^te^{-c{y'}_j^{3/2}}.
\end{equation}
where $V$ is given by (\ref{3.492}). From (\ref{3.501}) we can write
\begin{equation}\label{4.25}
N_1^{1/3}Q'_1(0)=\sum\limits_{r=0}^{\min(n_1,\Delta n)}\sum\limits_{s=0}^{n_1}\sum\limits_{t=0}^{\Delta n}
\sum\limits_{\substack{\mathbf{c}\in [1,n_1]_<^r\\\mathbf{c'}\in [1,n_1]_<^s}}\frac 1{N_1^{(r+s)/3}}
\sum\limits_{\substack{\mathbf{d}\in [n_1+2,n_2]_<^{r}\\\mathbf{d'}\in [n_1+2,n_2]_<^t}}\frac 1{N_1^{(r+t)/3}}
\det \left(N_1^{1/3}V\right).
\end{equation}
It follows from lemma \ref{lem4.1}, (\ref{3.28}), (\ref{3.29}) and (\ref{3.492}) that
$$
\lim_{M\to\infty} \det \left(N_1^{1/3}V\right) =W_{r,s,r,t}^{(1)}(\mathbf{x},\mathbf{x'},\mathbf{y},\mathbf{y'}).
$$
From the estimate (\ref{4.24}) we see that we can take the limit in (\ref{4.25}) and obtain (\ref{4.16}). The proof of (\ref{4.17})
is completely analogous.
\end{proof}

We now have all the results that we need to prove theorem \ref{thm1.1}.

\begin{proof} ({\it Proof of theorem \ref{thm1.1}}) 
Recall Proposition \ref{prop2.4}. In the scaling (\ref{4.3}) we see that
\begin{equation}\label{4.26}
\frac{\partial}{\partial\eta_1}\mathbb{P}\left[H(\mu_1,n_1)\le\xi_1,H(\mu_2,n_2)\le\xi_2\right]=\left.\frac{\partial}{\partial h}\right|_{h=0}
N_1^{1/3}Q(h).
\end{equation}
From lemma \ref{lem3.4} and lemma \ref{lem4.4} we see that
\begin{equation}\label{4.27}
\lim_{M\to\infty}\left.\frac{\partial}{\partial h}\right|_{h=0}N_1^{1/3}Q(h)=\Psi^{(1)}(\eta_1,\eta_2)+\Psi^{(2)}(\eta_1,\eta_2):=\Psi(\eta_1,\eta_2)
\end{equation}
uniformly for $\eta_1,\eta_2$ in a compact set. Let
$$
X_M=\frac{H(\mu_1,n_1)-2t_1M}{(t_1M)^{1/3}}+\nu_1^2\,,\,
Y_M=\frac{H(\mu_2,n_2)-2t_2M}{(t_2M)^{1/3}}+\nu_2^2.
$$
Then (\ref{4.26}) can be written
\begin{equation}
\frac{\partial}{\partial\eta_1}\mathbb{P}\left[X_M\le\eta_1,Y_M\le\eta_2\right]=\left.\frac{\partial}{\partial h}\right|_{h=0}
N_1^{1/3}Q(h)
\notag
\end{equation}
and for fixed $\eta_1^\ast$ and $\tilde{\eta}_1$ we see that
$$
\mathbb{P}\left[\eta_1^\ast<X_M\le\tilde{\eta}_1,Y_M\le\eta_2\right]=\int_{\eta_1^\ast}^{\tilde{\eta}_1}\left.\frac{\partial}{\partial h}\right|_{h=0}
N_1^{1/3}Q(h)\,d\eta_1.
$$
From (\ref{4.27}) it follows that
\begin{equation}\label{4.28}
\lim_{M\to\infty}\mathbb{P}\left[\eta_1^\ast<X_M\le\tilde{\eta}_1,Y_M\le\eta_2\right]=\int_{\eta_1^\ast}^{\tilde{\eta}_1}\Psi(\eta_1,\eta_2)\,d\eta_1.
\end{equation}
Now,
\begin{align}\label{4.29}
\mathbb{P}\left[\eta_1^\ast<X_M\le\tilde{\eta}_1,Y_M\le\eta_2\right]&
\le\mathbb{P}\left[\eta_1^\ast<X_M,Y_M\le\eta_2\right]
\\
&\le\mathbb{P}\left[\eta_1^\ast<X_M\le\tilde{\eta}_1,Y_M\le\eta_2\right]+\mathbb{P}\left[X_M>\tilde{\eta}_1\right].
\notag
\end{align}

From (\ref{1.4}), (\ref{4.28}) and (\ref{4.29}) we see that
\begin{align}\label{4.30}
&\int_{\eta_1^\ast}^{\tilde{\eta}_1}\Psi(\eta_1,\eta_2)\,d\eta_1\le \liminf\limits_{M\to\infty} \mathbb{P}\left[\eta_1^\ast<X_M,Y_M\le\eta_2\right]
\\
&\le \limsup\limits_{M\to\infty}\mathbb{P}\left[\eta_1^\ast<X_M,Y_M\le\eta_2\right]\le \int_{\eta_1^\ast}^{\tilde{\eta}_1}\Psi(\eta_1,\eta_2)\,d\eta_1
+1-F_2(\tilde{\eta}_1).
\notag
\end{align}
If we let $\tilde{\eta}_1\to\infty$ in (\ref{4.30}) we see that
$$
\lim_{M\to\infty} \mathbb{P}\left[\eta_1^\ast<X_M,Y_M\le\eta_2\right]=\int_{\eta_1^\ast}^{\infty}\Psi(\eta_1,\eta_2)\,d\eta_1,
$$
which is what we wanted to prove. 

Note that in order for this last argument to work we need an estimate of $\Psi(\eta_1,\eta_2)$
in terms of $\eta_1$. In fact, there are constants $c,C>0$ such that
\begin{equation}\label{4.31}
\left|\Psi(\eta_1,\eta_2)\right|\le Ce^{-c(\eta_1)_+^{3/2}}.
\end{equation}
We will only sketch the argument for (\ref{4.31}). Note that $\phi_1$, $\psi_1$ and $\phi_3$ all have a decay of the form
$e^{-c(\eta_1)_+^{3/2}}$ in $\eta_1$ by known estimates for the Airy function. Hence, the difficulty is in the presence of $\phi_2$.
If $r\ge 1$, the first column in $W_{r,s,r,t}^{(1)}$ does not depend on $\phi_2$ (we can assume $x_1<0$) and hence the
first column (in a Hadamard estimate) will give the right $\eta_1$-decay. If $r=0$, but $s\ge 1$, we can again consider 
the first column ($x'_1<0$), and get the right $\eta_1$-decay. If $r=s=0$,
$$
W_{0,0,0,t}^{(1)}(\mathbf{x},\mathbf{x'},\mathbf{y},\mathbf{y'})=\left|
\begin{matrix} \psi(0,0) &\psi(0,\mathbf{y'})\\
\psi(\mathbf{y'},0) &\psi(\mathbf{y'},\mathbf{y'})
\end{matrix}\right|
$$
and again the first column does not depend on $\phi_2$. The argument for $W_{r,s,r-1,t}^{(2)}$ is easier since we now always have $r\ge 1$.
\end{proof}

\section{Proof of the combinatorial identities}\label{sect5}

In this section we will prove lemma 2.3.

\begin{proof}
Consider first the identity (\ref{2.9}). We can write
\begin{equation*}
\prod_{i<j}\left(\frac 1{w_{\sigma(i)}w_{\sigma(j)}}-\frac 1{w_{\sigma(i)}}\right)=\prod_{i<j}\frac 1{w_{\sigma(i)}w_{\sigma(j)}}(1-w_{\sigma(j)}).
\end{equation*}
Now,
\begin{equation*}
\prod_{i<j}\frac 1{w_{\sigma(i)}w_{\sigma(j)}}=\prod\limits_{i=1}^{n-1}\frac 1{w_{\sigma(i)}^{n-i}}\prod\limits_{j=2}^n\frac 1{w_{\sigma(j)}^{j-1}}
=\prod\limits_{j=1}^{n}\frac 1{w_{\sigma(j)}^{n-1-j}}\prod\limits_{j=1}^n\frac 1{w_{\sigma(j)}^{j}}
\end{equation*}
and
\begin{equation*}
\prod_{i<j}(1-w_{\sigma(j)})=\prod\limits_{j=2}^n(1-w_{\sigma(j)})^{j-1}=\prod\limits_{j=1}^n\frac 1{1-w_{\sigma(j)}}\prod\limits_{j=1}^n(1-w_{\sigma(j)})^{j}.
\end{equation*}
Thus, we have the identity
$$
\prod_{i<j}\left(\frac 1{w_{\sigma(i)}w_{\sigma(j)}}-\frac 1{w_{\sigma(i)}}\right)=\prod\limits_{j=1}^n
\frac 1{(1-w_j)w_{\sigma(j)}^{n-1-j}}\left(\frac{1-w_{\sigma(j)}}{w_{\sigma(j)}}\right)^j.
$$
Hence, the left side of (\ref{2.9}) can be written
\begin{align}
&\sum\limits_{\sigma\in S_n}\sgn(\sigma)\prod_{i<j}\left(\frac 1{w_{\sigma(i)}w_{\sigma(j)}}-\frac 1{w_{\sigma(i)}}\right)
\frac{\prod\limits_{j=1}^n(1-w_j)w_{\sigma(j)}^{n-1-j}}
{(1-w_{\sigma(1)})\cdots(1-w_{\sigma(1)}\cdots w_{\sigma(n)})}
\notag\\
&=\sum\limits_{\sigma\in S_n}\sgn(\sigma)\prod_{i<j}\left(\frac 1{w_{\sigma(i)}w_{\sigma(j)}}-\frac 1{w_{\sigma(i)}}\right)
\frac{\prod\limits_{j=1}^nw_j^{-2}(1-w_j)}
{(\frac 1{w_{\sigma(1)}}-1)\cdots(\frac 1{w_{\sigma(1)}\cdots w_{\sigma(n)}}-1)}.
\notag
\end{align}
By the identity (1.7) in \cite{TrWi} with $p=0,q=1$ the last expression equals
$$
\prod\limits_{j=1}^nw_j^{-1}\left(\frac 1{w_j}-1\right)\frac 1{\frac 1{w_j}-1}\det\left(\frac 1{w_j^{i-1}}\right)=(-1)^{n(n-1)/2}
\prod\limits_{j=1}^n\frac 1{w_j^n}\det\left(w_j^{i-1}\right),
$$
where we also used (\ref{vandermondeidentity}). This proves (\ref{2.9}).

We now turn to the proof of (\ref{2.10}). Denote the left side of (\ref{2.10}) by $\omega_n(z,w)$. We will use induction on $n$. It
is easy to see that the identity is true for $n=1$. Fix $\sigma_1(n)=k$ and $\sigma_2(n)=\ell$. Then
\begin{align}\label{5.1}
&\sum\limits_{k,\ell=1}^n\sum\limits_{\substack{\sigma_1,\sigma_2\in S_n\\\sigma_1(n)=k,\sigma_2(n)=\ell}}\sgn(\sigma_1)\sgn(\sigma_2)
\left(\frac{1-z_k}{z_k}\right)^n\left(\frac{1-w_\ell}{w_\ell}\right)^n
\prod\limits_{j=1}^{n-1}\left(\frac{1-z_{\sigma_1(j)}}{z_{\sigma_1(j)}}\right)^j\left(\frac{w_{\sigma_2(j)}}{1-w_{\sigma_2(j)}}\right)^j
\\
&\times\frac 1{1-\frac{z_1\cdots z_n}{w_1\cdots w_n}}\frac 1{\left(1-\frac{z_{\sigma_1(1)}}{w_{\sigma_2(1)}}\right)
\left(1-\frac{z_{\sigma_1(1)}z_{\sigma_1(2)}}{w_{\sigma_2(1)}w_{\sigma_2(2)}}\right)
\cdots\left(1-\frac{z_{\sigma_1(1)}\cdots z_{\sigma_1(n-1)}}{w_{\sigma_2(1)}\cdots w_{\sigma_2(n-1)}}\right)}
\notag\\
&=\sum\limits_{k,\ell=1}^n\frac {(-1)^{n-k+n-\ell}}{1-\frac{z_1\cdots z_n}{w_1\cdots w_n}}\left(\frac{1-z_k}{z_k}\right)^n\left(\frac{w_\ell}{1-w_\ell}\right)^n
\omega_{n-1}(z_1,\dots,\hat{z_k},\dots, z_n,w_1,\dots,\hat{w_\ell},\dots, w_n),
\notag
\end{align}
where $\hat{z_k} (\hat{w_\ell})$ means that we leave out $z_k (w_\ell)$. By the induction hypothesis the last expression in (\ref{5.1}) equals
\begin{align}\label{5.2}
&\sum\limits_{k,\ell=1}^n\frac {(-1)^{k+\ell}}{1-\frac{z_1\cdots z_n}{w_1\cdots w_n}}\left(\frac{1-z_k}{z_k}\right)^n\left(\frac{w_\ell}{1-w_\ell}\right)^n
\prod\limits_{j\neq k}\frac{(1-z_j)^{n-1}}{z_j^{n-1}}\prod\limits_{j\neq \ell}\frac{w_j^{n}}{(1-w_j)^{n-1}}
\\
&\times\det\left(\frac 1{w_k-z_j}\right)_{1\le j,k\le n}\frac{\prod\limits_{j=1}^n(w_\ell-z_j)(w_j-z_k)}{(w_\ell-z_k)\prod\limits_{j=1}^{\ell-1}
(w_j-w_\ell)\prod\limits_{j=\ell+1}^{n}(w_\ell-w_j)\prod\limits_{j=1}^{k-1}
(z_k-z_j)\prod\limits_{j=k+1}^{n}(z_j-z_k)},
\notag
\end{align}
where we also used the Cauchy determinant formula. The expression in (\ref{5.2}) can be written
\begin{equation}\label{5.3}
\frac{\det\left(\frac 1{w_k-z_j}\right)}{1-\frac{z_1\cdots z_n}{w_1\cdots w_n}}\prod\limits_{j=1}^{n}\frac{w_j^n(1-z_j)^{n-1}}{z_j^{n-1}(1-w_j)^{n-1}}
\sum\limits_{k,\ell=1}^n\frac{(-1)^{n-1}(1-z_k)\prod\limits_{j=1}^n(w_\ell-z_j)(w_j-z_k)}{z_k(1-w_\ell)(w_\ell-z_k)\prod\limits_{j\neq \ell}
(w_\ell-w_j)\prod\limits_{j\neq k}(z_k-z_j)}.
\end{equation}
We see from (\ref{5.3}) and the final formula (\ref{2.10}) that in order to complete the proof we have to show the identity
\begin{align}\label{5.4}
&(-1)^{n-1}\sum\limits_{k,\ell=1}^n\frac{(1-z_k)}{z_k(1-w_\ell)(w_\ell-z_k)}\frac{\prod\limits_{j=1}^n(w_\ell-z_j)(w_j-z_k)}{\prod\limits_{j\neq \ell}
(w_\ell-w_j)\prod\limits_{j\neq k}(z_k-z_j)}
\\
&=\prod\limits_{j=1}^n\frac{w_j(1-z_j)}{z_j(1-w_j)}\left(1-\frac{z_1\dots z_n}{w_1\dots w_n}\right)=
\prod\limits_{j=1}^n\frac{w_j(1-z_j)}{z_j(1-w_j)}-\prod\limits_{j=1}^n\frac{1-z_j}{1-w_j}.
\notag
\end{align}
We can assume that $|z_j|, |w_j|<1$, $1\le j\le n$. Take $0<r_1<r_2<1$ such that $|z_j|<r_1$, $|w_j|<r_2$ for $1\le j\le n$. Consider the contour integral
\begin{align}\label{5.5}
&-\frac{1}{(2\pi i)^2}\int_{\gamma_{r_1}}dz\int_{\gamma_{r_2}}dw\frac{1-z}{z(1-w)(w-z)}\prod\limits_{j=1}^n
\frac{(w-z_j)(z-w_j)}{(w-w_j)(z-z_j)}=-\frac{1}{2\pi i}\int_{\gamma_{r_2}}dw\frac 1{(1-w)w}\prod\limits_{j=1}^n\frac{(w-z_j)w_j}{(w-w_j)z_j}
\\
&-\sum\limits_{k=1}^n\frac{1}{2\pi i}\int_{\gamma_{r_2}}dw\frac{1-z_k}{z_k(1-w)(w-z_k)}\prod\limits_{j=1}^n
\frac{(w-z_j)(z_k-w_j)}{w-w_j}\prod\limits_{j\neq k}\frac 1{z_k-z_j},
\notag
\end{align}
where we have computed the $z$-integral. The first expression in the right side of (\ref{5.5}) can be computed by noticing that the only pole 
outside $\gamma_{r_2}$
(including $\infty$) is at $w=1$ and this gives
$$
-\prod\limits_{j=1}^n\frac{w_j(1-z_j)}{z_j(1-w_j)}
$$
The second expression in the right side of (\ref{5.5}) equals
$$
(-1)^{n-1}\sum\limits_{k,\ell=1}^n\frac{(1-z_k)}{z_k(1-w_\ell)(w_\ell-z_k)}\frac{\prod\limits_{j=1}^n(w_\ell-z_j)(w_j-z_k)}{\prod\limits_{j\neq \ell}
(w_\ell-w_j)\prod\limits_{j\neq k}(z_k-z_j)}
$$
and thus by comparing (\ref{5.4}) and (\ref{5.5}) we see that it remains to show 
\begin{equation}\label{5.6}
\frac{1}{(2\pi i)^2}\int_{\gamma_{r_1}}dz\int_{\gamma_{r_2}}dw\frac{1-z}{z(1-w)(w-z)}\prod\limits_{j=1}^n
\frac{(w-z_j)(w_j-z)}{(w-w_j)(z-z_j)}=\prod\limits_{j=1}^n\frac{1-z_j}{1-w_j}.
\end{equation}
The $w$-integral in (\ref{5.6}) has its only pole outside $\gamma_{r_2}$ at $w=1$ which gives
$$
\frac{1}{2\pi i}\int_{\gamma_{r_1}}\frac{dz}{z}\prod\limits_{j=1}^n\frac{w_j-z}{z_j-z}\prod\limits_{j=1}^n\frac{1-z_j}{1-w_j}
=\prod\limits_{j=1}^n\frac{1-z_j}{1-w_j}\frac{1}{2\pi i}\int_{\gamma_{1/r_1}}\frac{dz}{z}\prod\limits_{j=1}^n\frac{zw_j-1}{zz_j-1}
=\prod\limits_{j=1}^n\frac{1-z_j}{1-w_j},
$$
since the only pole in the last $z$-integral is at $z=0$.
\end{proof}

\section{Asymptotic analysis}\label{sect6}

In this section we will prove lemma \ref{lem4.1} and lemma \ref{lem4.2}. Recall the notations and scalings (\ref{4.1}) to (\ref{4.3}).
Define, with $k$ and $\ell$ as in (\ref{4.3}),
\begin{align}
f_1(z;x)&=(\ell-1)\log z+\frac 12\mu_1z^2- \xi_1 z
\notag\\
f_2(z;y)&=(n_2-k)\log z+\frac 12\Delta\mu z^2- \Delta\xi z
\notag
\end{align}
and note that $n_2-k=\Delta n-yN_1^{1/3}$. Recall the notation (\ref{3.7}) and the definitions (\ref{3.24})- (\ref{3.27}).
We have that
\begin{align}\label{6.1}
G_{n_1,\mu_1,\xi_1}(z)&=e^{f_1(z;0)}\,\,,\,\,G_{\Delta n,\Delta\mu,\Delta\xi}(w)=e^{f_2(w;0)}
\\
G_{k,\mu_1,\xi_1}(\zeta)&=e^{f_1(\zeta;y)}\,\,,\,\,G_{n_2+1-\ell,\Delta\mu,\Delta\xi}(\omega)=e^{f_2(\omega;x)}
\notag\\
G_{n_2-k,\Delta\mu,\Delta\xi}(w)&=e^{f_2(w;y)}\,\,,\,\,G_{\ell-1,\mu_1,\xi_1}(z)=e^{f_1(z;x)}
\notag
\end{align}
Let $d_i$, $1\le i\le 4$, be some positive parameters that will be chosen later. Introduce the following
contour parametrizations
\begin{align}\label{6.2}
z(t_1)&=1+(d_1+it_1)N_1^{-1/3}\,\,,\,\,t_1\in\mathbb{R},
\\
\zeta(s_1)&=(1-d_2N_1^{-1/3})e^{is_1N_1^{-1/3}}\,\,,\,\,s_1\in I_1=[-\pi N_1^{1/3},\pi N_1^{1/3}],
\notag\\
w(t_2)&=1+(d_3+it_2)N_2^{-1/3}\,\,,\,\,t_2\in\mathbb{R},
\notag\\
\omega(s_2)&=(1-d_4N_2^{-1/3})e^{is_2N_2^{-1/3}}\,\,,\,\,s_2\in I_1=[-\pi N_2^{1/3},\pi N_2^{1/3}].
\notag
\end{align}
Define
\begin{align}\label{6.3}
g_1(t_1;x)&=\re f_1(z(t_1);x)\,\,,\,\,h_1(s_1;x)=\re f_1(\zeta(s_1);x)
\\
g_2(t_2;x)&=\re f_2(w(t_2);y)\,\,,\,\,h_2(s_2;y)=\re f_1(\omega(s_2);y).
\notag
\end{align}
Let
\begin{align}\label{6.3'}
 \Delta_1&=d_1-\nu_1+\frac12(d_1^2-2\nu_1 d_1-x)N_1^{-1/3}-\frac 12\nu_1d_1^2N_1^{-2/3}\\
\Delta_2&=2(d_2+\nu_1)+(\eta_1-\nu_1^2-2\nu_1d_2)N_1^{-1/3},\notag
\\
\Delta_3&=d_3-\Delta\nu+\frac12(d_3^2-2\Delta\nu d_3+y)N_2^{-1/3}-\frac 12\Delta\nu d_3^2N_2^{-2/3}\notag\\
\Delta_4&=2(d_4+\Delta\nu)+(\Delta\eta-\Delta\nu^2-2\Delta\nu d_4)N_2^{-1/3}.
\notag
\end{align}

\begin{lemma}\label{lem6.1} Assume that, for $M$ large,
\begin{equation}\label{6.4}
1\le d_1\le N_1^{1/3}\,\,,\,\,1\le\Delta_1\le N_1^{1/3},
\end{equation}
\begin{equation}\label{6.5}
1\le d_2\le \frac 12N_1^{1/3}\,\,,\,\,\Delta_2\ge 1,
\end{equation}
\begin{equation}\label{6.14}
1\le d_3\le N_3^{1/3}\,\,,\,\,1\le\Delta_3\le N_3^{1/3},
\end{equation}
\begin{equation}\label{6.15}
1\le d_4\le \frac 12N_2^{1/3}\,\,,\,\,\Delta_4\ge 1.
\end{equation}
Then,
\begin{equation}\label{6.6}
g_1(t_1;x)-g_1(0;x)\le -\frac{\Delta_1}{20} t_1^2
\end{equation}
for all $t_1\in\mathbb{R}$, and
\begin{equation}\label{6.7}
h_1(s_1;x)-h_1(0;x)\ge \frac{\Delta_2}{20} s_1^2
\end{equation}
for all $s_1\in I_1$. Furthermore
\begin{equation}\label{6.16}
g_2(t_2;y)-g_2(0;y)\le -\frac{\Delta_3}{20} t_2^2
\end{equation}
for all $t_2\in\mathbb{R}$, and
\begin{equation}\label{6.17}
h_2(s_2;y)-h_2(0;y)\ge \frac{\Delta_4}{20} s_2^2
\end{equation}
for all $s_2\in I_2$.
\end{lemma}

\begin{proof}
By (\ref{6.2}) and (\ref{6.3}),
\begin{equation}\label{6.8}
g_1(t_1;x)=\frac{\ell-1}2\log\left((1+d_1N_1^{-1/3})^2+N_1^{-2/3}t_1^2\right)+\frac 12\mu_1\left((1+d_1N_1^{-1/3})^2-N_1^{-2/3}t_1^2\right)
-\xi_1(1+d_1N_1^{-1/3}).
\end{equation}
Thus,
$$
g_1'(t_1;x)=N_1^{-2/3}t_1\left(\frac{\ell-1-\mu_1\left((1+d_1N_1^{-1/3})^2+N_1^{-2/3}t_1^2\right)}
{(1+d_1N_1^{-1/3})^2+N_1^{-2/3}t_1^2}\right),
$$
and by introducing the scalings (\ref{4.3}) we obtain
\begin{equation}\label{6.9}
g_1'(t_1;x)=-t_1\left(\frac{2\Delta_1+\left(N_1^{-1/3}-\nu_1N_1^{-2/3}\right)t_1^2}
{(1+d_1N_1^{-1/3})^2+N_1^{-2/3}t_1^2}\right).
\end{equation}
If $0\le t_1\le N_1^{1/3}$, then $(1+d_1N_1^{-1/3})^2+N_1^{-2/3}t_1^2\le 5$ by (\ref{6.4}), and if $M$ is large enough $N_1^{-1/3}-\nu_1N_1^{-2/3}\ge 0$,
so (\ref{6.9}) gives
\begin{equation}\label{6.10}
g_1'(t_1;x)\le -\frac{\Delta_1}5 t_1.
\end{equation}
If $t_1\ge N_1^{1/3}$ and $M$ is sufficiently large, then (\ref{6.9}) gives
\begin{equation}
g_1'(t_1;x)\le -t_1\frac{\frac12N_1^{-1/3}t_1^2}{(1+d_1N_1^{-1/3})^2+N_1^{-2/3}t_1^2}\le-t_1\frac{\frac12N_1^{-1/3}t_1^2}{5N_1^{-2/3}t_1^2}
\le-\frac 1{10}N_1^{1/3}t_1\le-\frac{\Delta_1}{10}t_1,
\notag
\end{equation}
by (\ref{6.4}). Hence, (\ref{6.10}) holds for all $t_1\ge 0$, and we have proved (\ref{6.6}) for $t_1\ge 0$. The case 
$t_1\le 0$ follows by symmetry.

Consider now $h_1$. We have that
\begin{equation}\label{6.11}
h_1(s_1;x)=(\ell-1)\log(1-d_2N_1^{-1/3})+\frac 12\mu_1(1-d_2N_1^{-1/3})^2\cos 2N_1^{-1/3}s_1 - \xi_1(1-d_2N_1^{-1/3})\cos N_1^{-1/3}s_1 
\end{equation}
and hence
\begin{equation}\label{6.12}
h_1'(s_1;x)=N_1^{-1/3}(1-d_2N_1^{-1/3})\sin N_1^{-1/3}s_1 \left(\xi_1-2\mu_1(1-d_2N_1^{-1/3})\cos N_1^{-1/3}s_1\right).
\end{equation}
From the scaling (\ref{4.3}) we see that if $M$ is sufficiently large then
$\xi_1-2\mu_1(1-d_2N_1^{-1/3})\cos N_1^{-1/3}s_1\ge N_1^{2/3}\Delta_2$ and hence,
$$
h_1(s_1;x)-h_1(0;x)\ge N_1^{1/3}(1-d_2N_1^{-1/3})\Delta_2\int_0^{s_1}\sin N_1^{-1/3}t\,dt\ge\frac {\Delta_2}2N_1^{2/3}
(1-\cos N_1^{-1/3}s_1)
$$
by (\ref{6.5}) for all $s_1\in I_1$. If $|N_1^{-1/3}s_1|\in[0,\pi /2]$, then
$$
\frac 12(1-\cos N_1^{-1/3}s_1)=\sin^2\left(\frac12 N_1^{-1/3}s_1\right)\ge\frac 14N_1^{-2/3}s_1^2
$$ 
and hence $h_1(s_1;x)-h_1(0;x)\ge\frac 14\Delta_2s_1^2$. If $|N_1^{-1/3}s_1|\in[\pi /2,\pi]$, then $1-\cos N_1^{-1/3}s_1\ge 1$,
and
$$
h_1(s_1;x)-h_1(0;x)\ge \frac12\Delta_2N_1^{2/3}\ge\frac 1{2\pi^2}\Delta_2s_1^2\ge\frac {\Delta_2}{20}s_1^2.
$$
Exactly the same argument gives (\ref{6.16}) and (\ref{6.17}).
\end{proof}

We will now prove lemma \ref{lem4.1}.

\begin{proof} ({\it Proof of lemma \ref{lem4.1}})
All the limits below will be uniform for $\nu_i,\eta_i,x,y$ in compact sets. Write
$$
u_1(t_1)=d_1+it_1\,,\,u_2(t_2)=d_3+it_2\,,\,v_1(s_1)=-d_2+is_1\,,\,v_2(s_2)=-d_4+is_2.
$$
Since $\nu_i,\eta_i,x,y$ belongs to a compact set it is clear that we can choose $d_i$,$1\le i\le 4$, constant but so large that
(\ref{6.4}), (\ref{6.5}), (\ref{6.14}) and (\ref{6.15}) hold for all sufficiently large $M$. Recall the definition (\ref{1.7}) of $\alpha$. In (\ref{3.24}) we will use the
parametrizations (\ref{6.2}) and we choose $d_1$ and $d_3$ so that
\begin{equation}\label{6.18}
\alpha d_3-d_1\ge 1
\end{equation}
which ensures that the $z$- and $w$-contours have the right ordering.
If we let
$$
J(t_1,s_1,t_2,s_2)=\frac{(1-d_2N_1^{-1/3})(1-d_4N_1^{-1/3})e^{is_1N_1^{-1/3}+is_2N_2^{-1/3}}}
{N_1^{2/3}N_2^{1/3}(z(t_1)-w(t_2))(z(t_1)-\zeta(s_1))(w(t_2)-\omega(s_2))},
$$
then
\begin{equation}\label{6.19}
\frac{N_1^{1/3}dzdwd\zeta d\omega}{(z-w)(z-\zeta)(w-\omega)}=\alpha J(t_1,s_1,t_2,s_2)dt_1ds_1dt_2ds_2
\end{equation}
and 
\begin{equation}\label{6.20}
J(t_1,s_1,t_2,s_2)\to\frac 1{(u_1(t_1)-\alpha u_2(t_2))(u_1(t_1)-v_1(s_1))(u_2(t_2)-v_2(s_2))}
\end{equation}
as $M\to\infty$; also $J$ is bounded. Furthermore,
\begin{align}\label{6.21}
&f_1(z(t_1);x)-f_1(1;x)\to\frac 13u_1(t_1)^3-\nu_1u_1(t_1)^2-(\lambda_1-x)u_1(t_1),
\\
&f_1(\zeta(s_1);x)-f_1(1;x)\to\frac 13v_1(s_1)^3-\nu_1v_1(s_1)^2-(\lambda_1-x)v_1(s_1),
\\
&f_2(w(t_2);y)-f_2(1;y)\to\frac 13u_2(t_2)^3-\Delta\nu u_2(t_2)^2-(\Delta\lambda+\alpha y)u_2(t_2),
\\
&f_2(\omega(s_2);y)-f_2(1;y)\to\frac 13v_2(s_2)^3-\Delta\nu v_2(s_2)^2-(\Delta\lambda+\alpha y)v_2(s_2)
\notag
\end{align}
as $M\to\infty$.

It follows from (\ref{3.24}) and (\ref{6.1}) that
\begin{equation}\label{6.22}
N_1^{1/3}a_{0,1}(\ell,k)=\frac{\alpha}{(2\pi)^4}\int_{\mathbb{R}}dt_1\int_{I_1}ds_1\int_{\mathbb{R}}dt_2\int_{I_2}ds_2J(t_1,s_1,t_2,s_2)
\frac{e^{f_1(z(t_1);0)+f_2(w(t_2);0)}}{e^{f_1(\zeta(s_1);y)+f_2(\omega(s_2);x)}}.
\end{equation}
The integrand in (\ref{6.22}) is bounded by
\begin{align}
&Ce^{g_1(t_1;0)+g_2(t_2;0)-h_1(s_1;y)-h_2(s_2;x)}
\notag\\
&\le Ce^{g_1(0;0)+g_2(0;0)-h_1(0;y)-h_2(0;x)-\frac 1{20}(t_1^2+s_1^2+t_2^2+s_2^2)}\le
Ce^{-\frac 1{20}(t_1^2+s_1^2+t_2^2+s_2^2)},
\notag
\end{align}
where the first inequality follows from lemma \ref{lem6.1} since $\Delta_i\ge 1$, and the second inequality
follows from (\ref{6.21}) by letting $t_1=s_1=t_2=s_2=0$ and taking real parts. Thus, by the dominated convergence theorem we can take the
limit $M\to\infty$ in (\ref{6.22}) and get
\begin{align}\label{6.23}
&\lim_{M\to\infty}N_1^{1/3}a_{0,1}(\ell,k)
\\
&=\frac{\alpha}{(2\pi)^4}\int_{\mathbb{R}}dt_1\int_{\mathbb{R}}ds_1\int_{\mathbb{R}}dt_2\int_{\mathbb{R}}ds_2
\frac 1{(u_1(t_1)-\alpha u_2(t_2))(u_1(t_1)-v_1(s_1))(u_2(t_2)-v_2(s_2))}
\notag\\
&\times\frac{e^{\frac 13u_1(t_1)^3-\nu_1u_1(t_1)^2-\lambda_1u_1(t_1)+\frac 13u_2(t_2)^3-\Delta\nu u_2(t_2)^2-\Delta\lambda u_2(t_2)}}
{e^{\frac 13v_1(s_1)^3-\nu_1v_1(s_1)^2-(\lambda_1-y)v_1(s_1)+\frac 13v_2(s_2)^3-\Delta\nu v_2(s_2)^2-(\Delta\lambda+\alpha x)v_2(s_2)}}
\notag\\
&=\frac{\alpha}{(2\pi i)^4}\int_{\Gamma_{d_1}}dz\int_{\Gamma_{d_3}}dw\int_{\Gamma_{-d_2}}d\zeta\int_{\Gamma_{-d_4}}d\omega
\frac{e^{\frac 13z^3- \nu_1 z^2- \lambda_1 z+\frac 13w^3-\Delta\nu w^2-\Delta\lambda w}}
{(z-\alpha w)(z-\zeta)(w-\omega)e^{\frac 13\zeta^3- \nu_1 \zeta^2- (\lambda_1-y)\zeta+\frac 13\omega^3-\Delta\nu \omega^2-(\Delta\lambda+\alpha x)\omega}}
\notag\\
&=\phi_1(x,y),
\notag
\end{align}
where $\phi_1$ is given by (\ref{1.9}). Recall the condition in (\ref{6.23}). The last equality is a straightforward
rewriting of the contour integral in terms of Airy functions, see the end of this section. This proves (\ref{4.4}). The
limit of $N_1^{1/3}b_1(\ell,k)$ is the same as the right side of (\ref{6.23}), but we have the condition $d_1>\alpha d_3$ instead.
For $c_2$ we get
\begin{align}\label{6.232}
\lim_{M\to\infty}N_1^{1/3}c_2(\ell,k)&=\frac{\alpha}{(2\pi i)^2}\int_{\Gamma_{d_3}}dw\int_{\Gamma_{-d_4}}d\omega
\frac{e^{\frac 13w^3-\Delta\nu w^2-(\Delta\lambda+\alpha y)w}}
{(w-\omega)e^{\frac 13\omega^3-\Delta\nu \omega^2-(\Delta\lambda+\alpha x)\omega}}
\\
&=\phi_2(x,y),
\notag
\end{align}
and for $c_3$,
\begin{align}\label{6.233}
\lim_{M\to\infty}N_1^{1/3}c_3(\ell,k)&=\frac{\alpha}{(2\pi i)^2}\int_{\Gamma_{d_1}}dz\int_{\Gamma_{-d_2}}d\zeta
\frac{e^{\frac 13z^3- \nu_1 z^2- (\lambda_1-x)z}}
{(z-\zeta)e^{\frac 13\zeta^3- \nu_1 \zeta^2- (\lambda_1-y)\zeta}}
\\
&=\phi_3(x,y).
\notag
\end{align}
\end{proof}

We turn now to the proof of lemma \ref{lem4.2}.

\begin{proof} ({\it Proof of lemma \ref{lem4.2}})
To prove the estimate (\ref{4.8}) we will use (\ref{6.22}) but we will make appropriate choices
of the $d_i$'s in order to get the estimate. From (\ref{6.22}) we find
\begin{equation}\label{6.24}
\left|N_1^{1/3}a_{0,1}(\ell,k)\right|\le \frac{C}{|d_1-\alpha d_3|(d_1+d_2)(d_3+d_4)}
\int_{\mathbb{R}}dt_1\int_{I_1}ds_1\int_{\mathbb{R}}dt_2\int_{I_2}ds_2
e^{g_1(t_1;0)-h_1(s_1;y)+g_2(t_2;0)-h_2(s_2;x)}.
\end{equation}
We will choose $d_i$ so that the conditions (\ref{6.4}), (\ref{6.5}), (\ref{6.14}), (\ref{6.15}) and (\ref{6.18}) are satisfied. Hence,
it follows from (\ref{6.24}) and lemma \ref{lem6.1} that
\begin{equation}\label{6.25}
\left|N_1^{1/3}a_{0,1}(\ell,k)\right|\le 
Ce^{g_1(0;0)-h_1(0;y)+g_2(0;0)-h_2(s_2;x)}.
\end{equation}
From (\ref{6.8}), (\ref{6.11}) and the scalings (\ref{4.3}) we see that
\begin{align}\label{6.26}
g_1(0;x)&=(N_1+\nu_1N_1^{2/3}+xN^{1/3})\log(1+d_1N_1^{-1/3})\\&+\frac 12(N_1-\nu_1N_1^{2/3})(1+d_1N_1^{-1/3})^2
-(2N_1+\lambda_1N^{1/3})(1+d_1N_1^{-1/3})\notag
\end{align}
and
\begin{align}\label{6.27}
h_1(0;y)&=(N_1+\nu_1N_1^{2/3}+yN^{1/3})\log(1-d_2N_1^{-1/3})\\&+\frac 12(N_1-\nu_1N_1^{2/3})(1-d_2N_1^{-1/3})^2
-(2N_1+\lambda_1N^{1/3})(1-d_2N_1^{-1/3})\notag.
\end{align}
It is straightforward to show that
\begin{equation}\label{6.28}
\log(1+x)\le x-\frac{x^2}2+\frac{x^3}3
\end{equation}
for all $x\ge 0$, and
\begin{equation}\label{6.29}
\log(1-x)\ge -x-\frac{x^2}2-\frac{x^3}{3(1-x)^3}
\end{equation}
if $0\le x<1$. If we use the estimate (\ref{6.28}) in (\ref{6.26}) we get
\begin{equation}\label{6.30}
g_1(0;x)\le -\frac 32 N_1-\frac 12\nu_1N_1^{2/3}-\lambda_1N^{1/3}+\frac 13d_1^3\left(1+\nu_1N_1^{-1/3}+xN_1^{2/3}\right)
-\nu_1d_1^2-\lambda_1d_1+x\left(d_1-\frac 12d_1^2N_1^{-1/3}\right).
\end{equation} 
Similarly, using (\ref{6.29}) in (\ref{6.27}) we find
\begin{equation}\label{6.31}
h_1(0;y)\ge -\frac 32 N_1-\frac 12\nu_1N_1^{2/3}-\lambda_1N^{1/3}+\frac 13d_2^3\left(\frac{1+\nu_1N_1^{-1/3}+yN_1^{2/3}}
{(1-d_2N_1^{-1/3})^3}\right)
-\nu_1d_2^2+\lambda_1d_2-y\left(d_2-\frac 12d_2^2N_1^{-1/3}\right).
\end{equation} 
Combining (\ref{6.30}) and (\ref{6.31}) we obtain
\begin{align}\label{6.32}
&g_1(0;x)-h_1(0;y)\le\frac 13d_1^3\left(1+\nu_1N_1^{-1/3}+xN_1^{2/3}\right)-\nu_1d_1^2-\lambda_1d_1+x\left(d_1-\frac 12d_1^2N_1^{-1/3}\right)
\\
&+\frac 13d_2^3\left(\frac{1+\nu_1N_1^{-1/3}+yN_1^{2/3}}
{(1-d_2N_1^{-1/3})^3}\right)+\nu_1d_2^2-\lambda_1d_2+y\left(d_2+\frac 12d_2^2N_1^{-1/3}\right).
\notag
\end{align}
In an analogous way, we obtain
\begin{align}\label{6.33}
&g_2(0;y)-h_2(0;x)\le\frac 13d_3^3\left(1+\Delta\nu N_2^{-1/3}-yN_2^{2/3}\right)-\Delta\nu d_3^2-\Delta\lambda d_3-y\left(d_3-\frac 12d_3^2N_2^{-1/3}\right)
\\
&+\frac 13d_4^3\left(\frac{1+\Delta\nu N_2^{-1/3}-xN_2^{2/3}}
{(1-d_4N_2^{-1/3})^3}\right)+\Delta\nu d_4^2-\Delta\lambda d_4-x\left(d_4+\frac 12d_4^2N_2^{-1/3}\right).
\notag
\end{align}
We will use the estimates (\ref{6.32}) and (\ref{6.33}) in (\ref{6.25}). Take
\begin{equation}\label{6.34}
d_1=k_1\,,\,d_2=k_2+\delta_2(-y)_+^{1/2}\,,\,
d_3=k_3\,,\,d_4=k_4+\delta_4x_+^{1/2},
\end{equation} 
where $k_i$ and $\delta_i$ are to be specified.

Note that since $1\le \ell,k\le n_2$, there is a constant $k_0$ so that $|x|\le k_0N_1^{2/3}$ and $|y|\le k_0N_1^{2/3}$.
First choose $k_1$ large enough so that $\Delta_1\ge 1$ holds. Then (\ref{6.4}) will hold if $M$ is large enough. We can
choose $k_2$ so that $\Delta_2\ge 1$ and  $d_2\ge 1$ hold provided that $d_2\le \frac 12N_1^{1/3}$. Now,
$$
d_2=k_2+\delta_2(-y)_+^{1/2}\le k_2+k_0^{1/2}\delta_2 N_1^{1/3}\le \frac 12 N_1^{1/3}
$$
for large $M$ if we choose $\delta_2$ small enough. With these choices (\ref{6.4}) and (\ref{6.5}) are satisfied for large $M$.
In a similar way we can choose $k_3,k_4$ and $\delta_4$ so that (\ref{6.14}) and (\ref{6.15}) hold, and we can also choose $k_3$
so large that (\ref{6.18}) holds. Note that there is a constant $C$ so that
$$
\frac{1+\nu_1 N_1^{-1/3}+yN_1^{-2/3}}{(1-d_2N_1^{-1/3})^3}\le C\,,\,
\frac{1+\Delta\nu N_2^{-1/3}-xN_2^{-2/3}}{(1-d_4N_2^{-1/3})^3}\le C
$$
and consequently we see from (\ref{6.32}) and (\ref{6.33}) that
\begin{align}\label{6.35}
&g_1(0;0)-h_1(0;y)+g_2(0;0)-h_2(0;x)\le \frac 13d_1^3(1+\nu_1N_1^{-1/3})-\nu_1d_1^2-\lambda_1d_1+Cd_2^3+\nu_1d_2^2-\lambda_1d_2
\\
&+y(d_2+\frac 12d_2^21N_1^{-1/3})+\frac 13d_3^3(1+\Delta\nu N_2^{-1/3})-\Delta\nu d_3^2-\Delta\lambda d_3
+Cd_4^3+\Delta\nu d_4^2-\Delta\lambda d_4-x(d_4+\frac 12d_4^2N_2^{-1/3})
\notag\\
&\le C(1+d_2^3+d_4^3+d_2^2+d_4^2+d_2+d_4)+y(d_2+\frac 12d_2^21N_1^{-1/3})-x(d_4+\frac 12d_4^2N_2^{-1/3}),
\notag
\end{align}
since $d_1$ and $d_3$ are constants. From (\ref{6.34}) we see that
$d_2^3\le 4(k_2^3+\delta_2^3(-y)_+^{3/2})$, $d_2^2\le 2(k_2^2+\delta_2^2(-y)_+)$ and similarly for $d_4$. If $y\ge 0$, then
$$
y(d_2+\frac 12d_2^2N_1^{-1/3})\le Cy=Cy_+,
$$
since $d_2=k_2$, and if $y<0$, then
$$
y(d_2+\frac 12d_2^2N_1^{-1/3})\le d_2y=k_2y-\delta_2(-y)_+^{3/3}\le -\delta_2(-y)_+^{3/3}.
$$ 
Thus
$$
y(d_2+\frac 12d_2^2N_1^{-1/3})\le -\delta_2(-y)_+^{3/3}+Cy_+
$$ 
for all $y$. Similarly,
$$
-x(d_4+\frac 12d_4^2N_2^{-1/3})\le -\delta_4x_+^{3/3}+C(-x)_+.
$$ 
We can pick $\delta_2$ so small that
$$
C(\delta_2(-y)_+^{3/3}+\delta_2^2(-y)_+)-\delta_2(-y)_+^{3/3}\le -c(-y)_+^{3/2},
$$
$c>0$ is a small constant. A similar argument can be done for $\delta_4$. Using these estimates in (\ref{6.35}) it follows 
from (\ref{6.25}) that
$$
\left|N_1^{1/3}a_{0,1}(\ell,k)\right|\le Ce^{-c(x_+^{3/2}+(-y)_+^{3/2})+C(y_++(-x)_+)},
$$
which is what we wanted to prove. The estimates (\ref{4.9}), (\ref{4.10}) and (\ref{4.11}) can be proved in a similar way using
(\ref{6.32}) and (\ref{6.33}). We will not go into the details.
\end{proof}

Let us briefly indicate how we can go from the contour integral form of $\phi_1(x,y)$ in (\ref{6.23}) to the Airy form
in (\ref{1.9}). We use the fact that if $D>0$, then
\begin{equation}\label{6.36}
\frac 1{2\pi i}\int_{\Gamma_D}e^{\frac 13z^3+Az^2+Bz}dz=\Ai(-B+A^2)e^{-AB+\frac 23 A^3}
\end{equation}
and
\begin{equation}\label{6.37}
\frac 1{2\pi i}\int_{\Gamma_{-D}}e^{-\frac 13\zeta^3+A\zeta^2+B\zeta}d\zeta=\Ai(B+A^2)e^{AB+\frac 23 A^3}.
\end{equation}
Also, we write
\begin{equation}\label{6.38}
\frac 1{z-\alpha w}=-\int_0^\infty e^{\tau_1(z-\alpha w)}d\tau_1\,,\,
\frac 1{z-\zeta}=\int_0^\infty e^{-\tau_2(z-\zeta)}d\tau_2\,,\,
\frac 1{w-\omega}=\int_0^\infty e^{-\tau_3(w-\omega)}d\tau_3.
\end{equation}
If we insert (\ref{6.38}) into (\ref{6.23}) and use (\ref{1.6}), (\ref{6.36}) and (\ref{6.37}) we get
(\ref{1.9}) after some manipulations.


\end{document}